\documentclass[a4paper,onecolumn,10pt,aps,nofootinbib,superscriptaddress]{revtex4}

\usepackage[T1]{fontenc}
\usepackage{lmodern}
\usepackage{amsmath}
\usepackage{amssymb}  
\usepackage{amsthm}
\usepackage{amsfonts}
\usepackage{graphicx}
\usepackage{cancel}
\usepackage{bbm}
\usepackage{url}
\usepackage{enumerate} 
\usepackage{ upgreek }
\usepackage{dsfont}
\usepackage{color}
\usepackage{makecell}
\usepackage{diagbox}
\usepackage{wrapfig}
\usepackage{rotating}
\usepackage{tabularx}

\usepackage{hyperref}
\hypersetup{
    colorlinks=true,
    linkcolor=blue,
    citecolor=red,
    filecolor=magenta,      
    urlcolor=cyan,
}

\newcommand{\ket}[1]{| #1 \rangle}

\newcommand{\bra}[1]{\langle #1 |}

\newcommand{\braket}[2]{\langle #1 | #2 \rangle}

\newcommand{\proj}[2]{| #1 \rangle\!\langle #2 |}

\newcommand{\id}{\ensuremath{\mathds{1}}}

\def\beq{\begin{equation}}
\def\eeq{\end{equation}}
\def\bq{\begin{quote}}
\def\eq{\end{quote}}
\def\ben{\begin{enumerate}}
\def\een{\end{enumerate}}
\def\bit{\begin{itemize}}
\def\eit{\end{itemize}}

\def\ra{\rightarrow}

\def\lb{\left(}
\def\rb{\right)}
\def\lset{\lbrace}
\def\rset{\rbrace}

\def\r|{\right|}
\def\lbr{\left[}
\def\rbr{\right]}
\def\ident{\textnormal{id}}
\def\one{\id}

\newcommand\C{\mathbb{C}}

\newcommand\R{\mathbb{R}}
\newcommand\N{\mathbb{N}}
\newcommand\M{\mathcal{M}}
\newcommand\D{\mathcal{D}}

\newcommand\F{\mathbb{F}}

\DeclareMathOperator{\Dec}{Dec}
\DeclareMathOperator{\Enc}{Enc}
\DeclareMathOperator{\EC}{EC}
\DeclareMathOperator{\Loc}{Loc}
\DeclareMathOperator{\Tr}{Tr}
\DeclareMathOperator{\Ad}{Ad}

\newcommand{\PIID}{\pi}

\theoremstyle{plain}
\newtheorem{thm}{Theorem}[section]
\newtheorem{lem}[thm]{Lemma}
\newtheorem{cor}[thm]{Corollary}

\newtheorem{defn}[thm]{Definition}

\theoremstyle{definition}

\begin{document}
\title{\textbf{Fault-tolerant Coding for Quantum Communication}}
\author{Matthias Christandl}
\email{christandl@math.ku.dk}
\affiliation{Department of Mathematical Sciences, University of Copenhagen, Universitetsparken 5, 2100 Copenhagen, Denmark}
\author{Alexander M\"uller-Hermes}
\email{muellerh@math.uio.no, muellerh@posteo.net}
\affiliation{Institut Camille Jordan, Universit\'{e} Claude Bernard Lyon 1, 43 boulevard du 11 novembre 1918, 69622 Villeurbanne cedex, France}
\affiliation{Department of Mathematics, University of Oslo, P.O. box 1053, Blindern, 0316 Oslo, Norway}

\begin{abstract}
Designing encoding and decoding circuits to reliably send messages over many uses of a noisy channel is a central problem in communication theory. When studying the optimal transmission rates achievable with asymptotically vanishing error it is usually assumed that these circuits can be implemented using noise-free gates. While this assumption is satisfied for classical machines in many scenarios, it is not expected to be satisfied in the near term future for quantum machines where decoherence leads to faults in the quantum gates. As a result, fundamental questions regarding the practical relevance of quantum channel coding remain open. 

By combining techniques from fault-tolerant quantum computation with techniques from quantum communication, we initiate the study of these questions. We introduce fault-tolerant versions of quantum capacities quantifying the optimal communication rates achievable with asymptotically vanishing total error when the encoding and decoding circuits are affected by gate errors with small probability. Our main results are threshold theorems for the classical and quantum capacity: For every quantum channel $T$ and every $\epsilon>0$ there exists a threshold $p(\epsilon,T)$ for the gate error probability below which rates larger than $C-\epsilon$ are fault-tolerantly achievable with vanishing overall communication error, where $C$ denotes the usual capacity. Our results are not only relevant in communication over large distances, but also on-chip, where distant parts of a quantum computer might need to communicate under higher levels of noise than affecting the local gates. 

\end{abstract}

\date{\today}

\maketitle

\tableofcontents
\section{Introduction}

Shannon's theory of communication~\cite{shannon1948mathematical} from the 1940s is the theoretical foundation for the communication infrastructure we use today. While it is extremely successful, it also turned out to be incomplete as a theory of all communication physically possible, since it did not consider the transmission of quantum particles (e.g., photons). This was first noted by Holevo in the 1970s~\cite{holevo1972mathematical,holevo1973bounds} and since led to a theory of quantum communication, called quantum Shannon theory, which includes Shannon's communication theory as a special case~\cite{Nielsen2007,wilde2017quantum}.    

A basic problem in both classical and quantum Shannon theory is the transmission of messages via noisy communication channels. To maximize the transmission rate while maintaining a high probability of correct transmission the sender may encode the message to make it more resilient against the noise introduced by the channel. The receiver then uses a decoder to guess the transmitted message. This is depicted in Figure \ref{fig:capacityIdeal}. 

\begin{figure*}[hbt!]
        \center
        \includegraphics[scale=0.5]{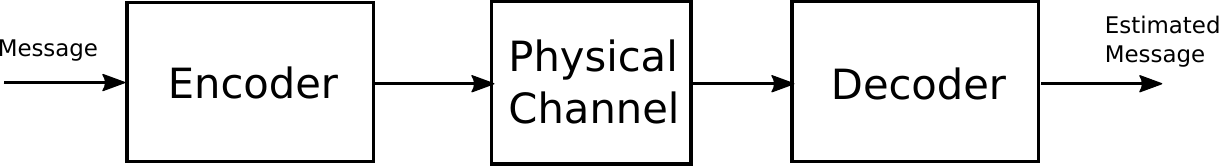}
        \caption{Basic communication setting.}
        \label{fig:capacityIdeal}
\end{figure*}

In Shannon theory it is often assumed that messages can be encoded into multiple uses of the same channel. Surprisingly, for many communication channels it is then possible to transmit messages at a positive rate and with a success probability approaching $1$ in the asymptotic limit of infinitely many channel uses. The communication channels for which this is possible and the optimal communication rates are precisely characterized by Shannon's famous capacity formula~\cite{shannon1948mathematical}. For proving this formula it is assumed that the encoder and the decoder can be executed perfectly, i.e., without introducing additional errors. This assumption is realistic in many applications of classical Shannon theory since the error rates of individual logic gates in modern computers are effectively zero (at least in non-hostile environments and at the time-scales relevant for communication~\cite{nicolaidis2010soft}). 

Different generalizations of Shannon's capacity are studied in quantum Shannon theory, since messages could be classical or quantum, or assisted by entanglement. All those scenarios have in common that encoding and decoding operations are assumed to be implemented without errors as in Shannon's original work~\cite{Nielsen2007,wilde2017quantum}. However, this is not a realistic assumption: At least in the near term future it is not expected that the error rates of individual quantum logic gates will effectively vanish~\cite{preskill2018quantum}. The practical relevance of capacities in quantum Shannon theory is therefore still unclear, and it is an open problem at what rates information can be send over a noisy quantum channel in the presence of gate errors.  

\subsection{Communication setting with gate errors}

Throughout this article, we will consider communication settings where the encoder and decoder (cf.~Figure \ref{fig:capacityIdeal}) are quantum circuits, i.e., compositions of elementary operations such as unitary operations from a universal gate set, measurements, preparations, and partial traces. Each elementary operation in a circuit is a location where a fault could happen with a certain probability. For simplicity we will consider the noise model of i.i.d.~Pauli noise, which we denote by $\mathcal{F}_\PIID(p)$, commonly used in the literature on fault tolerance (see~\cite{aharonov2008fault,aliferis2006quantum}). Here, locations are chosen i.i.d.~at random with some fixed probability $p>0$ and at each chosen location a random Pauli error is applied (see Definition \ref{defn:FaultPattern} and Definition \ref{defn:NoiseModel} below). Now, the basic problem is to ensure high communication rates in the presence of gate errors (happening with some fixed probability $p>0$ at each location in the coding circuits) while the probability of correct transmission of data converges towards $1$ in the limit of asymptotically many channel uses. We define the fault-tolerant classical capacity $C_{\mathcal{F}_\PIID(p)}(T)$ of a classical-quantum or quantum channel $T$, and the fault-tolerant quantum capacity $Q_{\mathcal{F}_\PIID(p)}(T)$ of a quantum channel $T$ as the suprema of achievable communication rates in the presence of gate errors modelled by the noise model $\mathcal{F}_\PIID(p)$.

\subsection{Nature of the problem}

The importance of analyzing noise in encoding and decoding operations has previously been noted in the context of long distance quantum communication~\cite{briegel1998quantum} and it is an important subject in the practical implementation of quantum communication~\cite{dur1999quantum,jiang2009quantum,muralidharan2014ultrafast,epping2016error}. However, in these works the overall success probability of the considered protocols does not approach $1$ at a finite communication rate in the limit of infinitely many channel uses. Therefore, it does not resolve the aforementioned issues.  

To deal with non-vanishing gate errors in the context of quantum computation, the notion of fault-tolerance has been developed using quantum error correcting codes. In particular, if the error rates of individual quantum logic gates are below a certain threshold, then it is possible to achieve a vanishing overall error rate for a quantum computation with only a small overhead~\cite{aharonov2008fault,aliferis2006quantum}. While our approach will use these techniques from fault-tolerance, we should emphasize that they will not directly lead to high communication rates over quantum channels in the presence of gate errors, and in most cases they cannot even be directly applied to the problem. The latter point is exemplified by the qubit depolarizing channel $\rho\mapsto (1-p)\rho + p(\sigma_x\rho \sigma_x + \sigma_y\rho \sigma_y + \sigma_z\rho \sigma_z)/3$ where $\sigma_x,\sigma_y,\sigma_z$ denote the Pauli matrices and $p\in \lbr 0,1\rbr$. It has been shown in~\cite{fern2008lower} that this channel has non-zero quantum capacity for $p\sim 0.19$, but fault-tolerance under this noise has only been shown~\cite{paetznick2012fault,chamberland2016thresholds} with thresholds $p\sim 0.001$. Therefore, it is not possible to treat the communication channel as circuit noise since it might cause faults with probabilities orders of magnitude above the threshold required for fault-tolerance.

Even if we want to send information over a depolarizing channel with parameter $p$ below the threshold it is not clear that the codes used in fault-tolerance would achieve non-vanishing communication rates. The concatenated codes considered in~\cite{aharonov2008fault,aliferis2006quantum} encode a single logical qubit in a growing number of physical qubits to achieve increasingly better protection against depolarizaing noise with parameter $p$ below threshold, but this scheme has a vanishing communication rate. While it might be possible to send quantum information at non-vanishing communication rates in the presence of gate errors by directly using constant-overhead fault-tolerance schemes~\cite{gottesman2014fault,fawzi2018constant}, this would again only work for communication channels close to identity. 

Finally, it might be possible to achieve good communication rates in the presence of gate errors by concatenating specific codes for a communication channel with quantum error correcting codes allowing for fault-tolerance. However, there is a fundamental problem with this approach and it is unclear how to resolve it: When the inner code is the code for the communication channel (which might be a random code of large dimension) it is in general unclear how to implement quantum circuits fault-tolerantly on the concatenated code. On the other hand, if the code for the communication channel is the outer code, then its codewords are encoded in the inner code and they cannot be directly transmitted through the quantum channel. Hence, it is not clear whether such schemes can lead to high communication rates in the presence of gate errors, when the success probability of the overall communication scheme should approach $1$ in the limit of asymptotically many channel uses.

\subsection{Our approach and main results}

To construct encoders and decoders for sending information over quantum channels at high rates in the presence of gate errors, we will use fault-tolerance, but combine it with techniques from quantum Shannon theory. Specifically, we consider a quantum error correcting code enabling fault-tolerant quantum computation and use it to implement certain encoding and decoding circuits fault-tolerantly. Note that this is similar to the idea of concatenating a channel code (as outer code) with a circuit code (as inner code) outlined above. We will then use interface quantum circuits translating between logical qubits and physical qubits (i.e., the qubits that the communication channel acts on). The overall setup is depicted in Figure \ref{fig:capacityReal}. Note that the interfaces are quantum circuits themself and therefore affected by gate errors. 

\begin{figure*}[hbt!]
        \center
        \includegraphics[scale=0.38]{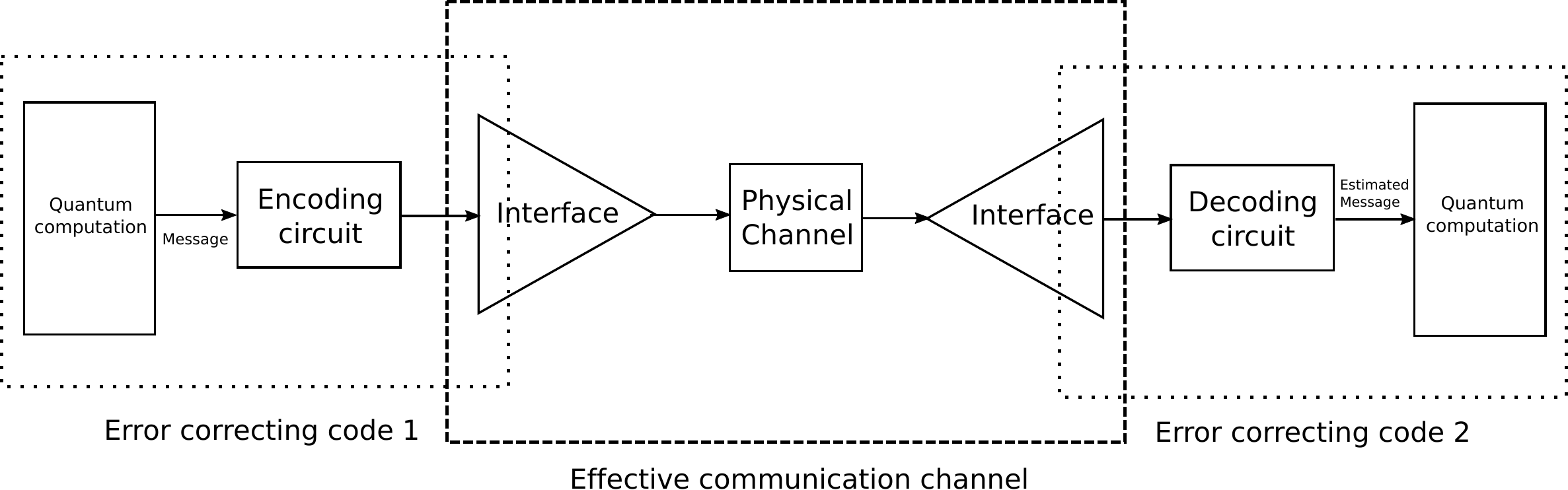}
        \caption{Our approach: Encoder and decoder are implemented in error correcting codes and interfaces are used to convert logical to physical states and vice-versa.}
        \label{fig:capacityReal}
\end{figure*} 

To find good lower bounds on the fault-tolerant capacities $C_{\mathcal{F}_\PIID(p)}(T)$ and $Q_{\mathcal{F}_\PIID(p)}(T)$ of a quantum channel $T$ under the i.i.d. Pauli noise model $\mathcal{F}_\PIID(p)$, we take the following steps: 

\begin{enumerate}
\item We show that there are interface circuits for the concatenated $7$-qubit Steane code used in~\cite{aliferis2006quantum}, that only fail with a probability linear in $p$ (regardless of the concatenation level), when affected by the noise model $\mathcal{F}_\PIID(p)$. 
\item We define an effective communication channel by combining the channel $T$ with the noisy interfaces (see Figure \ref{fig:capacityReal}). This channel will be close to the original channel $T$ with distance linear in $p$. 
\item We consider encoders and decoders achieving rates close to the capacity of the effective communication channel and by implementing them in the concatenated $7$-qubit Steane code we construct the overall coding scheme. 
\item By continuity of the ideal capacities $C(T)$ and $Q(T)$ (i.e., the classical and quantum capacity of $T$ without gate errors) we conclude that the achieved rates are also close to these quantities. 
\end{enumerate}

There are some pitfalls in the steps outlined above. Most importantly, a careful analysis of the interface circuits under noise will show that the effective communication channel might not be a tensor power of a well-defined quantum channel, i.e., it might not be i.i.d.~in general even though both the physical channel and circuit noise model are i.i.d. The reason for this are correlated errors produced within the quantum circuit at the encoder, e.g., through non-local CNOT gates. While such errors are correctable by the circuit code, failures in the interface (which happen with probabilities proportional to the gate error probability $p$) might cause non-i.i.d.~correlations to emerge in the effective channel. Instead, the effective channel takes the form of a fully-quantum arbitrarily varying channel~\cite{boche2018fully} with some additional structure, which we refer to as a quantum channel under arbitrarily varying perturbation (AVP). A quantum channel $T$ under AVP can be seen as a perturbation of the original channel $T$ by a quantum channel interacting with an environment system. When taking tensor powers of such a channel the environment systems corresponding to the tensor factors can be in an entangled quantum state leading to correlations between the different applications of the channel. We will analyze classical and quantum capacities of quantum channels under AVP in two settings: The first setting allows only classical correlations in the environment state and the second setting allows for arbitrary (including entangled) environment states. In both scenarios, we find coding schemes achieving rates close to the standard capacities by applying techniques from the quantum Shannon theory of non-i.i.d.~channel models (e.g., codes for arbitrarily varying channels, and postselection techniques). Finally, we show that the capacities under AVP are lower bounds to the fault-tolerant capacities which implies our main results, which are threshold theorems for fault-tolerant capacities of classical-quantum and quantum channels:
\begin{itemize}
	\item For every $\epsilon>0$ and every dimension $d\geq 2$ there exists a threshold $p(\epsilon,d)>0$ such that  
\[
C_{\mathcal{F}_\PIID(p)}(T) \geq C(T) - \epsilon
\]
for all $0\leq p\leq p(\epsilon,d)$ and for all classical-quantum channels $T:\mathcal{A}\ra \M_d$.
\item For every $\epsilon>0$ and every quantum channel $T:\M_{d_1}\ra \M_{d_2}$ there exists a $p(\epsilon,T)>0$ such that 
\[
C_{\mathcal{F}_\PIID(p)}(T) \geq C(T) - \epsilon  \quad\text{ and }\quad Q_{\mathcal{F}_\PIID(p)}(T)\geq Q(T) - \epsilon
\]
for all $0\leq p\leq p(\epsilon,T)$.
\end{itemize}
Our results show that communication at good rates and with vanishing communication error is possible in non-trivial cases and in realistic scenarios where local gates are affected by noise. This is an important validation of the practical relevance of quantum Shannon theory. Moreover, our results are not only relevant for communication over large distances, but also for applications within a quantum computer. Here, communication between distant parts of a quantum computing chip may be affected by a higher level of noise than the local gates. Our results could then be used to optimize the design of quantum hardware with on-chip communication.

Our article is organized as follows:
\begin{itemize}
\item In Section \ref{sec:Prelim} we will go through preliminaries needed for the rest of the article. 
\begin{itemize}
\item In Section \ref{sec:Notation} we introduce the most basic notation. 
\item In Section \ref{sec:QCircuitsAndNoise} we review the circuit model of quantum computation and we explain how to model the noise affecting quantum circuits. 
\item In Section \ref{sec:AnalyzingNoisyQuantumCircuits} we slightly simplify the formalism from~\cite{aliferis2006quantum}, and we explain how it can be used to analyze quantum circuits affected by noise. 
\item In Section \ref{sec:ConcCodesThresh} we review the threshold theorem from~\cite{aliferis2006quantum} for concatenated quantum codes.
\item In Section \ref{sec:CapQuChannel} we define the most basic capacities considered in quantum Shannon theory: The classical capacity of a classical-quantum or quantum channel, and the quantum capacity of a quantum channel.  
\end{itemize}
\item In Section \ref{sec:FTTechniques} we study quantum communication in the framework of fault-tolerant quantum circuits.
\begin{itemize}
\item In Section \ref{sec:NoisyInterface} we define interfaces between physical qubits and logical qubits encoded in stabilizer codes. For the concatenated $7$-qubit Steane code we analyze how an interface is affected by i.i.d.~Pauli noise. 
\item In Section \ref{sec:EncDecEffComm} we use interfaces for the concatenated $7$-qubit Steane code to study quantum protocols for communication over a quantum channel. In particular, we identify the effective communication channel (cf.~Figure \ref{fig:capacityReal}). 
\end{itemize}
\item In Section \ref{sec:CapAVP} we study transmission of classical or quantum information via quantum channels under arbitrarily varying perturbation (AVP). This includes the effective communication problem from Section \ref{sec:EncDecEffComm} as a special case.
\begin{itemize}
\item In Section \ref{sec:CapAVPFirst} we introduce quantum channels under AVP and we define the corresponding classical and quantum capacity.
\item In Section \ref{sec:CapAVPSep} we study quantum channels under AVP with fully separable environment states. This channel model reduces to the arbitrarily varying quantum channels studied in~\cite{ahlswede2013quantum}.
\item In Section \ref{sec:CapAVPAll} we study quantum channels under AVP with unrestricted environment states (including highly entangled states). We apply a postselection technique to prove a lower bound on these quantities.  
\end{itemize} 
\item In Section \ref{sec:FTCapacities} we introduce and study fault-tolerant versions of the classical capacity and the quantum capacity. 
\begin{itemize}
\item In Section \ref{sec:FTCapCQ} we study the fault-tolerant classical capacity $C_{\mathcal{F}_\PIID(p)}(T)$ under i.i.d.~Pauli noise of a classical-quantum channel. 
\item In Section \ref{sec:FTCapQC} we study the fault-tolerant classical capacity $C_{\mathcal{F}_\PIID(p)}(T)$ and the fault-tolerant quantum capacity $Q_{\mathcal{F}_\PIID(p)}(T)$ under i.i.d.~Pauli noise of a quantum channel. We show that these capacities are lower bounded by capacities under AVP.
\item In Section \ref{sec:GoodCodes} we show how asymptotically good codes can be used to design fault-tolerant coding schemes for quantum communication for quantum channels arising as convex combinations with the identity channel.
\end{itemize}
\item In Section \ref{sec:Conclusion} we conclude our article with some ideas for further research.
\end{itemize}

\section{Preliminaries}
\label{sec:Prelim}

\subsection{Notation}
\label{sec:Notation}

We will denote by $\M_d:=\M\lb\C^d\rb$ the matrix algebra of complex $d\times d$-matrices and by $\mathcal{D}(\C^d)$ the set of $d$-dimensional quantum states, i.e., the set of positive semidefinite matrices in $\M_d$ with trace equal to $1$. As is common in the mathematical literature, we will define quantum channels as completely positive and trace-preserving maps $T:\M_{d_A}\ra\M_{d_B}$. In analogy, we define channels with classical output as positive and entry-sum-preserving maps into $\C^{d}$, where we identify probability distributions on the set $\lset 1,\ldots ,d\rset$ with entrywise-positive vectors in $\C^d$ such that the sum of their entries equals $1$ (which is equivalent to diagonal matrices in $\mathcal{D}(\C^d)$). By the same reasoning, we may define channels with classical input as linear maps on $\C^d$, which are positive with respect to the entrywise-positive order on $\C^d$ and the order on the output, and such that they yield normalized outputs on vectors with entries summing to $1$. Sometimes, we will also use the more common definition as maps from a finite alphabet $\mathcal{A}$, and extended linearly to the space of probability distributions on $\mathcal{A}$ denoted by $\mathcal{P}(\mathcal{A})$.

\subsection{Quantum circuits and noise models}
\label{sec:QCircuitsAndNoise}

A quantum channel $\Gamma:\M^{\otimes n}_2\ra\M^{\otimes m}_2$ is called a \emph{quantum circuit} if it can be written as a composition of elementary operations~\cite{Nielsen2007}. These \emph{elementary operations} are:
\begin{itemize}
\item Elementary qubit gates: Pauli gates $\sigma_x,\sigma_y$ and $\sigma_z$, Hadamard gate $H$, and the $T$-gate.
\item Identity gate corresponding to a resting qubit.
\item Controlled-not (CNOT) gate.
\item Measurements and preparations in the computational basis.
\item Qubit trace, i.e., discarding a qubit. 
\end{itemize}
It is well known that the quantum circuits form a dense subset of the set of quantum channels $T:\M^{\otimes n}_2\ra\M^{\otimes m}_2$ (see for instance~\cite{boykin1999universal,barenco1995elementary,Nielsen2007}). Note that there may be many ways to construct the same quantum circuit (viewed as a quantum channel) from elementary operations. Moreover, after physically implementing the quantum circuit, its performance under noise might depend on the specific construction. To simplify our discussion we will assume every quantum circuit $\Gamma$ to be specified by a particular \emph{circuit diagram} $G_\Gamma$. Formally,  $G_\Gamma$ is an acyclic directed graph with vertices colored by elementary operations, and edges corresponding to qubits interacting at elementary gates. We define the set of \emph{locations} of $\Gamma$ denoted by $\Loc(\Gamma)$ as the set of vertices of $G_\Gamma$, and we will denote by $d_\text{out}(l)\in\lset 0,1,2\rset$ the outdegree of a location $l\in\Loc(\Gamma)$. Note that $d_\text{out}(l)=0$ if the location $l$ is a measurement or a trace, $d_\text{out}(l)=1$ if the location $l$ is an elementary qubit gate, identity gate, or a preparation, and $d_\text{out}(l)=2$ if the location $l$ is a CNOT gate.

Different models of noise affecting quantum circuits have been studied in the literature. Here, we are restricting to the simplest case of these models where Pauli noise affects locations of the circuit locally. To model such noise we will select subsets of locations in the circuit diagram where we will insert Pauli noise channels either before or after the executed elementary operation. Formally, this introduces a second coloring of the vertices of the circuit diagram with colors representing different Pauli channels describing the noise. It should be emphasized that whether noise channels are inserted before or after the elementary operations is a convention (except when elementary operations are preparations or measurements). Here, we choose the same convention as in~\cite{aliferis2006quantum}, but other choices might be reasonable as well.  

\begin{defn}[Pauli fault patterns and faulty circuits]\label{defn:FaultPattern}
Consider the single qubit Pauli channels $\Ad_{\sigma_x},\Ad_{\sigma_y},\Ad_{\sigma_z}:\M_2\ra\M_2$ defined as $\Ad_{\sigma_i}\lb X\rb = \sigma_i X\sigma_i$ for $i\in\lset x,y,z\rset$ and any $X\in\M_2$. A \emph{Pauli fault pattern} affecting a quantum circuit $\Gamma:\M^{\otimes n}_2\ra\M^{\otimes m}_2$ is a function $F:\Loc(\Gamma)\ra \lset \ident,x,y,z\rset\cup \lb\lset\ident,x,y,z\rset\times \lset\ident,x,y,z\rset\rb$ such that $F(l)\in\lset \ident,x,y,z\rset $ if $d_\text{out}(l)\in\lset 0,1\rset$ and $F(l)\in\lset\ident,x,y,z\rset\times \lset\ident,x,y,z\rset$ if $d_\text{out}(l)=2$. We denote by $\lbr\Gamma\rbr_{F}:\M^{\otimes n}_2\ra\M^{\otimes m}_2$ the quantum channel obtained by inserting the particular noise channels in the execution of the circuit diagram of the circuit $\Gamma$. Specifically, we do the following:
\begin{itemize}
\item We apply the Pauli channel $\Ad_{\sigma_{F(l)}}$ directly before the location $l\in\Loc(\Gamma)$ if $d_{\text{out}}(l)=0$.
\item We apply the Pauli channel $\Ad_{\sigma_{F(l)}}$ directly after the location $l\in\Loc(\Gamma)$ if $d_{\text{out}}(l)=1$.
\item We apply the Pauli channel $\Ad_{\sigma_{F(l)_1}}\otimes \Ad_{\sigma_{F(l)_2}}$ directly after the location $l\in\Loc(\Gamma)$ if $d_{\text{out}}(l)=2$.
\end{itemize}
When given a Pauli fault pattern $F$ affecting a composition $\Gamma_1\circ \Gamma_2$ of quantum circuits, we will often abuse notation and write $\lbr \Gamma_1\rbr_F$ and $\lbr \Gamma_2\rbr_F$ to denote the restriction of $F$ (as a function) to the locations of the circuits $\Gamma_1$ and $\Gamma_2$, respectively.
\end{defn}

In our simplified treatment, a \emph{noise model} $\mathcal{F}$ specifies a probability distribution over fault patterns to occur in a given quantum circuit $\Gamma$. We will denote by $\lbr\Gamma\rbr_\mathcal{F}$ the circuit affected by the noise model, i.e., the quantum channel obtained after selecting a fault pattern at random according to the noise model and inserting it into the circuit diagram of the circuit. Again it will be convenient to sometimes abuse notation, and use the notation $\lbr\cdot\rbr_\mathcal{F}$ for combinations of quantum circuits and quantum communication channels. For example, we might write $\lbr\Gamma_D\circ T^{\otimes m}\circ \Gamma_E\rbr_\mathcal{F}$, where $\Gamma_E$ and $\Gamma_D$ are quantum circuits and $T$ denotes a quantum channel that is not viewed as a quantum circuit. By this notation we mean the overall quantum channel obtained from applying the noise model to the quantum circuits $\Gamma_E$ and $\Gamma_D$ that have well-defined locations, but the quantum channel $T$ is not affected by this.  

In the following, we will restrict to a very basic type of noise.

\begin{defn}[I.i.d.~Pauli noise model]\label{defn:NoiseModel}
Consider a quantum circuit $\Gamma:\M^{\otimes n}_2\ra\M^{\otimes m}_2$. The \emph{i.i.d.~Pauli noise model} $\mathcal{F}_{\PIID}(p)$ selects a Pauli fault pattern $F:\Loc(\Gamma)\ra \lset \ident,x,y,z\rset\cup \lb\lset\ident,x,y,z\rset\times \lset\ident,x,y,z\rset\rb$ with the probability
\[
\text{P}\lb F\rb = (1-p)^{l_{\ident}}(p/3)^{l_x}(p/3)^{l_y}(p/3)^{l_z}
\] 
where
\begin{align*}
l_i &:= \Big|\lset l\in \text{Loc}\lb\Gamma\rb~:~d_{\text{out}}(l)\in\lset 0,1\rset \text{ and }F(l)= i\rset\Big|+\Big|\lset l\in \text{Loc}\lb\Gamma\rb~:~d_{\text{out}}(l)=2 \text{ and }F(l)_1=i\rset\Big| \\
&\quad\quad\quad\quad\quad+\Big|\lset l\in \text{Loc}\lb\Gamma\rb~:~d_{\text{out}}(l)=2 \text{ and }F(l)_2=i\rset\Big|,
\end{align*} 
for any $i\in\lset \ident,x,y,z\rset$.
\end{defn}

It is straightforward to define other examples of i.i.d.~noise models or more general local noise models similar to the previous definitions. The i.i.d.~Pauli noise model corresponds to inserting the depolarizing channel $D_p:\M_2\ra \M_2$ given by $D_p(X) = (1-p)X + p \sum_{i\in \lset x,y,z\rset} \sigma_i X \sigma_i$ on the output systems of every location in the circuit. Our main results also hold (with slightly modified proofs) for noise models where arbitrary qubit quantum channels of the form $X\mapsto (1-p)X + pN(X)$ are inserted after each location, and where the qubit quantum channel $N$ may be different for each location. Different techniques might be required for local noise models that are not i.i.d.~, or for noise even more exotic. We will comment on this further at the appropriate places.

\subsection{Analyzing noisy quantum circuits}
\label{sec:AnalyzingNoisyQuantumCircuits}

To protect a quantum circuit against noise it can be implemented in a quantum error correcting code. Here, we will describe how to analyze noisy quantum circuits following the ideas of~\cite{aliferis2006quantum}. We will first introduce idealized encoding and decoding operations that will select a specified basis in which faults are interpreted. These operations are unitary and can be inserted into a quantum circuit affected by a specific fault pattern. When the fault pattern is nice enough it will then be possible to transform the faulty quantum circuit into the ideal quantum circuit by decoupling the data from the noise encoded in a quantum state corresponding to possible syndromes. In our discussion we will restrict to stabilizer codes encoding a single qubit where these constructions can be done in a starightforward way.   

Let $\mathcal{C}\subset (\C^2)^{\otimes K}$ denote the \emph{code space} of a $2$-dimensional stabilizer code, i.e., the common eigenspace for eigenvalue $+1$ of a collection $\lset g_1,\ldots , g_{K-1}\rset$ of commuting product-Pauli matrices generating a matrix group not containing $-\id$. In such a setting we denote by $W_s\subset (\C^{2})^{\otimes K}$ for $s\in\F^{K-1}_2$ the space of common eigenvectors for the eigenvalues $(-1)^{s_i}$ with respect to each $g_i$. By definition it is clear, that $W_s\perp W_{s'}$ for $s\neq s'$ and that $\text{dim}\lb W_s\rb=\text{dim}\lb\mathcal{C}\rb=2$ for any $s\in\F^{K-1}_2$. Therefore, we have 
\[
\lb\C^2\rb^{\otimes K} = \bigoplus_{s\in\F^{K-1}_2} W_s 
\]
by a simple dimension counting argument, and $\mathcal{C} = W_{(1,1,\ldots ,1)}$. 

We will denote by $\lset\ket{\overline{0}},\ket{\overline{1}}\rset\subset\mathcal{C}$ an orthonormal basis of the code space, i.e., the encoded computational basis. For each $s\in \F^{K-1}_2$ we can select a product-Pauli operator\footnote{using that product-Pauli operators either commute or anticommute} $E_s:\lb\C^2\rb^{\otimes K}\ra\lb\C^2\rb^{\otimes K}$ such that 
\[
W_s = E_s\lb\mathcal{C}\rb. 
\]
Note that in this way $W_{00\cdots 0} = \mathcal{C}$. To follow the usual convention we will call the operator $E_s$ the Pauli error associated with the \emph{syndrome} $s\in\F^{K-1}_2$. By the previous discussion the set
\begin{equation}
\bigcup_{s\in\F^{K-1}_2} \lset E_s\ket{\overline{0}}, E_s\ket{\overline{1}}\rset
\label{equ:ErrorBasis}
\end{equation}
forms a basis of $\lb\C^2\rb^{\otimes K}$ and it will be convenient to analyze noisy quantum circuits with respect to this basis. This approach follows closely the analysis of~\cite[Section 5.2.2]{aliferis2006quantum}, but makes it slightly more precise. We start by defining a linear map $D:(\C^2)^{\otimes K}\ra \C^2\otimes (\C^2)^{\otimes K-1}$ acting as
\[
D\lb E_s\ket{\overline{i}}\rb = \ket{i}\otimes \ket{s}, 
\]
on the basis from $\eqref{equ:ErrorBasis}$ and extend it linearly. Here we used the abbreviation $\ket{s}:=\ket{s_1}\otimes\ket{s_2}\otimes\cdots \otimes\ket{s_{K-1}}$ for $s\in\F^{K-1}_2$. Clearly, $D$ is a unitary change of basis when identifying $\C^2\otimes (\C^2)^{\otimes K-1}\simeq (\C^2)^{\otimes K}$. We can therefore define the unitary quantum channel $\Dec^*:\M^{\otimes K}_{2}\ra\M^{\otimes K}_2$ by 
\begin{equation}
\Dec^*(X) = DXD^\dagger,
\label{equ:DecIdeal}
\end{equation} 
and its inverse $\Enc^*:\M^{\otimes K}_{2}\ra\M^{\otimes K}_2$ by 
\begin{equation}
\Enc^*(X) = D^\dagger X D, 
\label{equ:EncIdeal}
\end{equation}
for any $X\in \M^{\otimes K}_2$. The quantum channel $\Dec^*$ is also called the \emph{ideal decoder}, and $\Enc^*$ is called the \emph{ideal encoder}. Note that these quantum channels in the case of concatenated codes appear in~\cite[p.17]{aliferis2006quantum}, where they are called the ``$k-^{*}$decoder'' and the ``$k-^{*}$encoder''. Finally we define a quantum channel $\EC^*:\M^{\otimes K}_{2}\ra \M^{\otimes K}_{2}$ by 
\begin{equation}
\EC^* =  \Enc^*\circ\lbr\ident_2\otimes \lb \proj{0}{0}\Tr\rb\rbr\circ \Dec^*, 
\label{equ:IdealEC}
\end{equation}
where $\ket{0}\in(\C^2)^{\otimes (K-1)}$ corresponds to the zero syndrome, and $\Tr:\M^{\otimes (K-1)}_2\ra \C$. The quantum channel $\EC^*$ corrects errors on the data and is called the \emph{ideal error correction}. In the following, we will sometimes use a superscript $*$ in order to identify ideal operations or gates. Such operations or gates should be thought of as either mathematical objects appearing in the analysis of noisy quantum circuits (in the case of $\Enc^*$, $\Dec^*$ and $\EC^*$), or as logical operations or logical gates acting on logical states. In any case, these ideal objects are never subject to noise.

To implement quantum circuits in a stabilizer code with code space $\mathcal{C}\subset(\C^2)^K$ of dimension $\text{dim}(\mathcal{C})=2$ we will assume that there are quantum circuits, called \emph{gadgets}, implementing the elementary operation on the code space. 

\begin{defn}[Implementation of a quantum circuit]
Let $\mathcal{C}\subset(\C^2)^{\otimes K}$ satisfying $\text{dim}(\mathcal{C})=2$ be the code space of a stabilizer code, and let $\proj{0}{0}\in\M^{\otimes (K-1)}_2$ denote the pure state corresponding to the zero syndrome. We assume that certain elementary quantum circuits called gadgets are given: 
\begin{enumerate}
\item For each elementary single qubit operation $G^*:\M_2\ra\M_2$, we have a gadget $G:\M^{\otimes K}_2\ra \M^{\otimes K}_2$ such that 
\[
\Dec^*\circ G\circ \Enc^*\lb \cdot \otimes \proj{0}{0}\rb = G^{*}\lb\cdot\rb\otimes \proj{0}{0}.
\]
\item For the CNOT gate $G^*_{\text{CNOT}}:\M_4\ra\M_4$, we have a gadget $G_{\text{CNOT}}:\M^{\otimes K}_2\otimes\M^{\otimes K}_2 \ra \M^{\otimes K}_2\otimes \M^{\otimes K}_2$ such that 
\[
(\Dec^*)^{\otimes 2}\circ G_{\text{CNOT}}\circ (\Enc^*)^{\otimes 2}\lb \cdot \otimes (\proj{0}{0})^{\otimes 2}\rb = G_{\text{CNOT}}^{*}\lb\cdot\rb\otimes (\proj{0}{0})^{\otimes 2}.
\]
\item For a measurement $G^{*}:\M_2\ra\text{Diag}_2$ in the computational basis we have a gadget $G:\M^{\otimes K}_2\ra\text{Diag}_2$ such that
\[
G\circ \Enc^*\lb \cdot \otimes \proj{0}{0}\rb = G^*(\cdot).
\]
\item For a preparation in the computational basis $G^{*}:\C\ra \M_2$ we have a gadget $G:\C\ra \M^{\otimes K}_2$ such that 
\[
\Dec^*\circ G  = G^*\otimes \proj{0}{0}.
\]
\item For the trace $G^{*}:\M_2\ra\C$ we have a gadget $G:\M^{\otimes K}_2\ra\C$ such that 
\[
G\circ \Enc^*\lb \cdot \otimes \proj{0}{0}\rb = G^*(\cdot).
\]
\end{enumerate}
Besides the gadgets defined above, we consider a quantum circuit $\EC:\M^{\otimes K}_2\ra\M^{\otimes K}_2$ realizing the quantum channel $\EC^*:\M^{\otimes K}_2\ra\M^{\otimes K}_2$ from \eqref{equ:IdealEC}. Then, we define the \emph{rectangle} of an elementary operation $G^*$ with corresponding gadget $G$ to be the quantum circuit given by
\[
R_G = \begin{cases}
\EC\circ G &\text{ if $G^*$ is a single qubit operation, or a preparation.}\\
\EC^{\otimes 2}\circ G&\text{ if $G^*$ is a CNOT gate.}\\
G  &\text{ if $G^*$ is a measurement or a trace.}
\end{cases}
\]
Given a quantum circuit $\Gamma:\M^{\otimes n}_2\ra\M^{\otimes m}_2$ we define its \emph{implementation} $\Gamma_{\mathcal{C}}:\M^{\otimes nK}_2\ra\M^{\otimes mK}_2$ as the quantum circuit obtained by replacing each qubit in the circuit $\Gamma$ by a block of $K$ qubits, and each elementary operation in the circuit diagram of $\Gamma$ by the corresponding rectangle. 
\label{defn:Impl}
\end{defn}

Implementations of quantum circuits can be analyzed using the ideal encoder $\Enc^*$ and the ideal decoder $\Dec^*$ introduced previously. To illustrate this, consider a quantum error correcting code $\mathcal{C}\subset(\C^2)^K$ satisfying $\text{dim}(\mathcal{C})=2$ and an elementary single qubit operation $G^*:\M_2\ra\M_2$ with corresponding rectangle $R_G$. Using that $\Dec^*$ and $\Enc^*$ are inverse to each other, we can compute 
\begin{align*}
\Dec^*\circ R_G\circ \Enc^*\lb \cdot \otimes \proj{0}{0}\rb &= \Dec^*\circ \EC\circ G\circ \Enc^*\lb \cdot \otimes \proj{0}{0}\rb\\
& = \Dec^*\circ \EC\circ \Enc^*\circ \Dec^*\circ G\circ \Enc^*\lb \cdot \otimes \proj{0}{0}\rb \\
& = G^*(\cdot)\otimes \proj{0}{0},
\end{align*}
where we used the assumptions from Definition \ref{defn:Impl} and the fact that as quantum channels and without noise we have $\EC=\EC^*$ with the ideal error correction from \eqref{equ:IdealEC}. In a similar way, it can be checked that   
\[
(\Dec^*)^{\otimes m}\circ \Gamma_\mathcal{C}\circ (\Enc^*)^{\otimes n}\lb \cdot \otimes \proj{0}{0}^{\otimes n}\rb = \Gamma\lb\cdot\rb\otimes \proj{0}{0}^{\otimes m}
\]  
for any quantum circuit $\Gamma:\M^{\otimes n}_2\ra\M^{\otimes m}_2$, where $\Gamma_\mathcal{C}$ denotes its implementation as in Definition \ref{defn:Impl}. 

So far, we have not considered any noise affecting the elementary gates of a quantum circuit. In Definition \ref{defn:Impl} we have only described how to implement a given quantum circuit within a certain quantum error correcting code. It should be noted that the gadgets required for this construction are always easy to construct. However, it is more challenging to construct these gadgets in a way such that implementations of quantum circuits become fault-tolerant, which is one of the main achievements of \cite{aliferis2006quantum}. We will not go into the details of these constructions, but only review the concepts needed for our main results. 

Intuitively, a quantum circuit affected by noise should be called ``correct'' if its behaviour matches that of the same circuit without noise. Moreover, the probability of a quantum circuit being ``correct''  under the noise model that is considered should be high after implementing the circuit in a quantum error correcting code. It is the central idea in~\cite{aliferis2006quantum} to derive ``correctness'' of noisy quantum circuits from conditions that are satisfied with high probability by the rectangles making up the circuit. However, this approach requires some care. When the input to a rectangle is unconstrained it might already include a high level of noise. Then, a single additional fault occurring with probability $p$ under the noise model $\mathcal{F}_\PIID(p)$ in the rectangle could cause the accumulated faults to be uncorrectable and the overall circuit not to agree with its desired output. As a consequence the quantum circuit containing the rectangle would not be correct. To avoid this problem it makes sense to derive ``correctness'' from properties of \emph{extended rectangle} combining the rectangle and the error correction preceeding it. By designing the error correction in a certain way~\cite{aliferis2006quantum} its output can be controlled even if there are faults, and a meaningful condition can be defined. In the following, we will make these notions precise.

\begin{defn}[Extended rectangles]
Let $\mathcal{C}\subset(\C^2)^{\otimes K}$ satisfying $\text{dim}(\mathcal{C})=2$ be the code space of a stabilizer code where a gadget is defined for every elementary operation and for the error correction operation as in Definition \ref{defn:Impl}. The \emph{extended rectangle} corresponding to an elementary operation $G^*$ with corresponding gadget $G$ is the quantum circuit given by
\[
E_G = \begin{cases}
\EC\circ G\circ \EC &\text{ if $G^*$ is a single qubit operation.}\\
\EC^{\otimes 2}\circ G\circ \EC^{\otimes 2} &\text{ if $G^*$ is a CNOT gate.}\\
G\circ \EC &\text{ if $G^*$ is a measurement or a partial trace.} \\
\EC\circ G &\text{ if $G^*$ is a preparation.} \\
\end{cases}
\]
\label{defn:ExRec}
\end{defn}

In~\cite[p.10]{aliferis2006quantum} a combinatorial condition called ``goodness'' is introduced for extended rectangles affected by fault patterns. This condition behaves as follows: First, using the ideal decoder and encoder it can be shown that an implemented quantum circuit with classical input and output and affected by noise can be transformed into the ideal quantum circuit without noise whenever all its extended rectangles are ``good''. Secondly, the probability of an extended rectangle being ``good'' under the Pauli i.i.d.~noise model $\mathcal{F}_\PIID(p)$ is very high. By using the union bound, it can then be concluded that an implemented quantum circuit affected by the noise model $\mathcal{F}_\PIID(p)$ is ``correct'' with high probability.  Unfortunately, the precise definition of ``goodness'' is quite cumbersome, and we have chosen to avoid it here. Instead, we define when extended rectangles are \emph{well-behaved} under a fault-pattern, which is inspired by the transformation rules stated in~\cite[p.11]{aliferis2006quantum} for ``good'' extended rectangles. In particular, a well-behaved quantum circuit, i.e., a quantum circuit in which all extended rectangles are well-behaved, can be transformed in the same way as in~\cite{aliferis2006quantum} leading to an ideal quantum circuit when input and output are classical. Moreover, the notion of ``goodness'' from~\cite[Section 3.1.]{aliferis2006quantum} implies our notion of ``well-behaved''.    

\begin{defn}[Well-behaved extended rectangles and quantum circuits]
Let $\mathcal{C}\subset(\C^2)^{\otimes K}$ satisfying $\text{dim}(\mathcal{C})=2$ be the code space of a stabilizer code where a gadget is defined for every elementary operation as in Definition \ref{defn:Impl}. Let $G^*$ be an elementary operation with corresponding extended rectangle $E_G$ according to Definition \ref{defn:ExRec}, and let $F$ be a Pauli fault pattern on the quantum circuit $E_G$. We will call the extended rectangle $E_G$ \emph{well-behaved} under the fault pattern $F$ if the corresponding condition holds:
\begin{enumerate}
\item The operation $G^*$ is a single qubit operation and we have 
\[
\Dec^*\circ \lbr E_G\rbr_F = (G^*\otimes S^F_G)\circ \Dec^*\circ \lbr\EC\rbr_F,
\]
for some quantum channel $S^F_G:\M^{\otimes (K-1)}_2\ra \M^{\otimes (K-1)}_2$ on the syndrome space.
\item The operation $G^*$ is a CNOT gate and we have 
\[
(\Dec^*)^{\otimes 2}\circ \lbr E_G\rbr_F = (G^*\otimes S^F_G)\circ (\Dec^*)^{\otimes 2}\circ \lbr(\EC)^{\otimes 2}\rbr_F,
\]
for some quantum channel $S^F_G:\M^{\otimes 2(K-1)}_2\ra \M^{\otimes 2(K-1)}_2$ on the syndrome space.
\item The operation $G^*$ is a measurement or a trace and we have 
\[
\lbr E_G\rbr_F = (G^*\otimes \Tr)\circ \Dec^*\circ \lbr\EC\rbr_F.
\]
\item The operation $G^*$ is a preparation and we have 
\[
\Dec^*\circ \lbr E_G\rbr_F = G^*\otimes \sigma^F_G,
\]
for some quantum state $\sigma^F_G\in\M^{\otimes (K-1)}_2$.
\end{enumerate} 
Similarily, we will call the implementation $\Gamma_{\mathcal{C}}:\M^{\otimes nK}_2\ra\M^{\otimes mK}_2$ of a quantum circuit $\Gamma:\M^{\otimes n}_2\ra\M^{\otimes m}_2$ \emph{well-behaved} under the Pauli fault pattern $F$ affecting the quantum circuit $\Gamma_{\mathcal{C}}$ if all extended rectangles contained in $\lbr\Gamma_{\mathcal{C}}\rbr_F$ are well-behaved.

\label{defn:ExRecCorr}
\end{defn} 

To analyze a faulty but well-behaved implementation of a quantum circuit we can use the transformation rules from Definition \ref{defn:ExRecCorr} repeatedly. First, we either introduce an ideal decoder after the final step of the implemented circuit, or if the quantum circuit has classical output we use the transformation rules for measurements or traces in its final step and thereby obtain an ideal decoder. Second, we move the ideal decoder towards the beginning of the quantum circuit using the transformation rules from Definition \ref{defn:ExRecCorr} repeatedly. In Figure \ref{fig:dancing} we depict a schematic of this process. 

\begin{figure*}[hbt!]
        \center
        \includegraphics[scale=0.5]{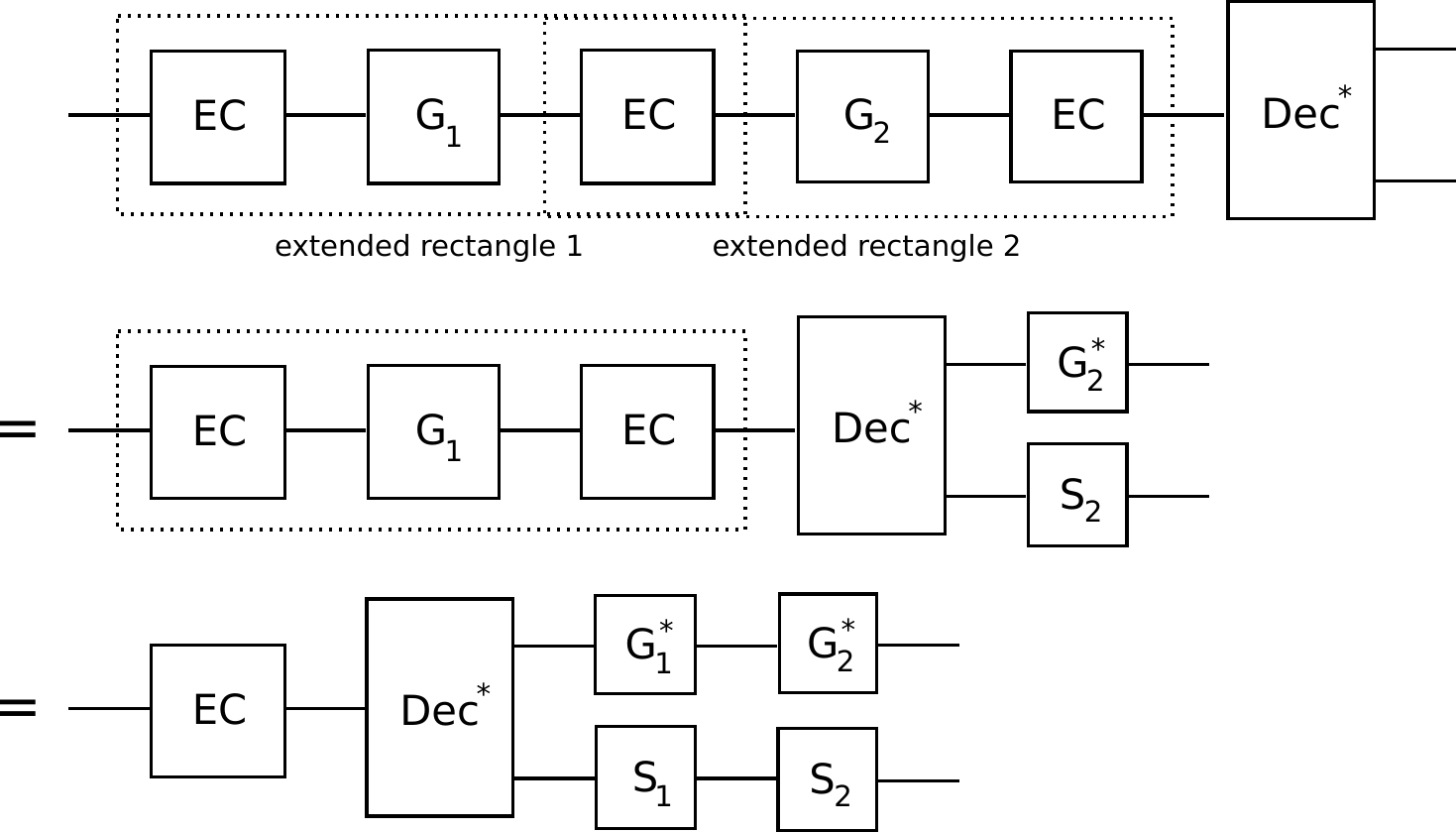}
        \caption{Transforming faulty but well-behaved extended rectangles using Definition \ref{defn:ExRecCorr}.}
        \label{fig:dancing}
\end{figure*} 

Finally, if the quantum circuit has classical input can use the transformation rule for a preparation (depending on the classical data) from Definition \ref{defn:ExRecCorr} to remove the ideal decoder in the initial step of the quantum circuit, or we keep the ideal decoder before the first error correction appearing in the quantum circuit. The argument just given is implicitly contained in~\cite[Lemma 4]{aliferis2006quantum} and the next lemma will make the conclusion more precise:    

\begin{lem}[Transformation of well-behaved implementations]
Let $\mathcal{C}\subset(\C^2)^{\otimes K}$ satisfying $\text{dim}(\mathcal{C})=2$ be the code space of a stabilizer code where a gadget is defined for every elementary operation as in Definition \ref{defn:Impl}, and let $\Gamma:\M^{\otimes n}_2\ra\M^{\otimes m}_2$ be a quantum circuit. If the quantum circuit $\Gamma_{\mathcal{C}}\circ \EC^{\otimes n}$ is well-behaved under a Pauli fault pattern $F$, then we have 
\[
(\Dec^*)^{\otimes m}\circ\lbr\Gamma_{\mathcal{C}}\circ \EC^{\otimes n}\rbr_F = \lb\Gamma\otimes S^F_\Gamma\rb\circ (\Dec^*)^{\otimes n}\circ\lbr\EC^{\otimes n}\rbr_F,
\] 
for some quantum channel $S^F_\Gamma:\M^{\otimes n(K-1)}_2\ra \M^{\otimes m(K-1)}_2$ acting on the syndrome space. Moreover, if $\Gamma:\C^{2^n}\ra\C^{2^m}$ is a quantum circuit with classical input and output, and $F$ is a Pauli fault pattern under which the quantum circuit $\Gamma_{\mathcal{C}}$ is well-behaved, then we have $\lbr\Gamma_{\mathcal{C}}\rbr_F = \Gamma.$
\label{lem:CircuitTransformExRecCorr}
\end{lem} 

\begin{proof}
We will only state the proof for the first part of the lemma, since the second part is an obvious modification. By Definition \ref{defn:Impl} the final part of the circuit $\Gamma_{\mathcal{C}}$ are error corrections, measurements or traces, depending on the final elementary operations in the circuit diagram of $\Gamma$. In the following, we denote by $\tilde{\Gamma}:\M^{n}_{2}\ra\M^{m'}_2$ the quantum circuit obtained by removing the final measurements and partial traces from $\Gamma$ leaving $m'\geq m$ qubits in the output. Without loss of generality we can assume qubits $1,\ldots ,m$ in the output of $\tilde{\Gamma}$ to correspond to the output of $\Gamma$. The extended rectangles corresponding to any measurement or partial trace are well-behaved by assumption, and replacing them as in Definition \ref{defn:ExRecCorr} shows that 
\[
(\Dec^*)^{\otimes m}\circ\lbr\Gamma_{\mathcal{C}}\circ \EC^{\otimes n}\rbr_F = \lb\bigotimes^{m}_{j=1}(\ident_2\otimes \ident_{S_j})\otimes \bigotimes^{m'}_{i=m+1}(G_i^*\otimes \Tr_{S_i} )\rb\circ(\Dec^*)^{\otimes m'}\circ\lbr\tilde{\Gamma}_{\mathcal{C}}\circ \EC^{\otimes n}\rbr_F,
\]  
where $G_i^*$ denote the ideal measurements or partial traces applied in the final step of $\Gamma$ and where the $\Tr_{S_i}$ acts on the syndrome system belonging to this code block. Using that every extended rectangle is well-behaved in the circuit $\lbr\tilde{\Gamma}_{\mathcal{C}}\circ \EC^{\otimes n}\rbr_F$ affected by fault pattern $F$, we can apply Definition \ref{defn:ExRecCorr} successively to transform each extended rectangle into the corresponding ideal operation. By doing so, the ideal decoder $\Dec^*$ moves towards the beginning of the circuit and in the final step we leave it directly after the initial error corrections. Collecting the quantum channels and partial traces acting on the syndrome space, and the quantum states on the syndrome space emerging from well-behaved extended rectangles corresponding to preparations into a single quantum channel $S^F_\Gamma:\M^{n(K-1)}_2\ra\M^{m(K-1)}_2$ finishes the proof. 
 
\end{proof}

Lemma \ref{lem:CircuitTransformExRecCorr} is only valid under the assumption that every extended rectangle in a potentially large circuit is well-behaved. Without any further assumptions this event might be very unlikely. In the next section we will restrict to quantum error correcting codes for which the extended rectangles are well-behaved with very high probability. Note that formally this is a property of both the code and the implementation of elementary gadgets (cf.~Definition \ref{defn:Impl}). In the following we will restrict to the concatenations~\cite{knill1996concatenated} of the $7$-qubit Steane code for which elementary gadgets have been constructed in~\cite{aliferis2006quantum}. Using these gadgets it is possible to prove the threshold theorem of~\cite{aliferis2006quantum} showing that fault-tolerant implementations of quantum circuits are possible.

\subsection{Concatenated quantum error correcting codes and the threshold theorem}
\label{sec:ConcCodesThresh}

A major result on fault-tolerant implementations of quantum circuits using quantum error correcting codes is the threshold theorem by Aharonov and Ben-Or~\cite{aharonov1996fault,aharonov1999fault,aharonov2008fault}. Here, we will focus our discussion on concatenated quantum error correcting codes~\cite{knill1996concatenated} constructed from the $7$-qubit Steane code and on the version of the threshold theorem stated in~\cite[Theorem 1]{aliferis2006quantum}. For convenience we have collected the construction and basic properties of this family of quantum error correcting codes in Appendix \ref{sec:Appendix}, but in the following we will not need to go through these details. We will only state the threshold theorem from~\cite{aliferis2006quantum} using the terminology introduced in the previous section. This will be sufficient to prove the results in the rest of our article. 

For any $l\in\N$ let $\mathcal{C}_l\subset (\C^2)^{\otimes 7^l}$ denote the $l$th level of the concatenated $7$-qubit Steane code. Note that each level defines a quantum error correcting code as introduced in the previous section. In particular, we use the following terminology throughout our article: 
\begin{itemize}
\item We denote by $\Enc^*_l$, $\Dec^*_l$, and $\EC^*_l$ the ideal operations introduced in \eqref{equ:EncIdeal}, \eqref{equ:DecIdeal}, and \eqref{equ:IdealEC} respectively for the $l$th level of the concatenated $7$-qubit Steane code. 
\item We refer to the elementary gadgets, error corrections, rectangles, and extended rectangles (see Definition \ref{defn:Impl} and Definition \ref{defn:ExRec}) at the $l$th level of the concatenated $7$-qubit Steane code as $l$-gadgets, $l$-error corrections, $l$-rectangles, and $l$-extended rectangles respectively. In formulas $l$ indices will indicate the level. 
\item All gadgets are constructed as explained in~\cite[Section 7]{aliferis2006quantum}.    
\end{itemize}  

In~\cite[p.10]{aliferis2006quantum} the notion of ``good'' extended rectangles is introduced, which implies the transformation rules of Definition \ref{defn:ExRecCorr}. As a consequence an extended rectangle is well-behaved whenever it is ``good''. Therefore, the following lemma follows directly from \cite[Lemma 2]{aliferis2006quantum}.

\begin{lem}[Threshold lemma]
For $l\in\N$ let $\mathcal{C}_l\subset (\C^2)^{\otimes 7^l}$ denote the $l$th level of the concatenated $7$-qubit Steane code. For every $l$ there exist gadgets implementing the elementary operations as in Definition \ref{defn:Impl} such that the following holds: There is a threshold $p_0\in\lb 0,1\rbr$ such that for any $0\leq p<p_0$, any $l\in\N$, and any $l$-extended rectangle $E^{(l)}_G$ we have 
\[
\text{P}\lb \lbr E^{(l)}_G\rbr_F\text{ is not well-behaved}\rb \leq p_0\lb \frac{p}{p_0}\rb^{2^l},
\] 
where $\text{P}$ denotes the probability distribution over Pauli fault patterns $F$ induced by the fault model $\mathcal{F}_{\PIID}(p)$.
\label{lem:Strongthreshold}
\end{lem}

Using the threshold theorem, long quantum computations can be protected from noise. The probability of any extended rectangle in a given quantum circuit to be not well-behaved can be upper bounded using the union bound. Combining this with Lemma \ref{lem:CircuitTransformExRecCorr} leads to the following theorem, which can be seen as yet another version of the threshold theorem~\cite[Theorem 1]{aliferis2006quantum}):

\begin{thm}[Threshold theorem II.]
For $l\in\N$ let $\mathcal{C}_l\subset (\C^2)^{\otimes 7^l}$ denote the $l$th level of the concatenated $7$-qubit Steane code with threshold $p_0\in\lb 0,1\rbr$ as in Lemma \ref{lem:Strongthreshold}. For any $0\leq p<p_0$, any $l\in\N$ the following statements hold:
\begin{enumerate}
\item For any quantum circuit $\Gamma:\M^{\otimes n}_2\ra\M^{\otimes m}_2$ we have
\[
\text{P}\lb \text{An extended rectangle in }\lbr \Gamma_{\mathcal{C}_l}\circ\EC_l^{\otimes n}\rbr_F\text{ is not well-behaved}\rb \leq C\lb \frac{p}{p_0}\rb^{2^l}|\Loc\lb\Gamma\rb|,
\]
where $C\in\R^+$ is a constant independent of $l\in\N$ and $\Gamma$, and $\text{P}$ denotes the probability distribution over Pauli fault patterns $F$ induced by the fault model $\mathcal{F}_{\PIID}(p)$.
\item For a quantum circuit $\Gamma:\C^{2^n}\ra\C^{2^m}$ with classical input and classical output we have
\[
\|\Gamma-\lbr \Gamma_{\mathcal{C}_l}\rbr_{\mathcal{F}_{\PIID}(p)}\|_{1\ra 1}\leq 2C\lb \frac{p}{p_0}\rb^{2^l}|\Loc\lb\Gamma\rb|.
\]
\end{enumerate}
\label{thm:Threshold2}
\end{thm}

It should be emphasized that in the previous discussion we did not have to define any notion of ``correct''. In~\cite[p.11]{aliferis2006quantum} correctness of rectangles under fault patterns is defined via the same transformation rules used in Definition \ref{defn:ExRecCorr} restricted to the rectangles contained in the corresponding extended rectangles (i.e., omitting the initial error correction). Since the proof of~\cite[Theorem 1]{aliferis2006quantum} and all the proofs in our article only use this notion together with well-behaved extended rectangles (or ``good'' ones in the case of~\cite{aliferis2006quantum}) it is sufficient to only use the stronger notion. An exception is our discussion of interfaces in Section \ref{sec:NoisyInterface} where we do define a notion of ``correctness'', which is related but different from the notion used in~\cite{aliferis2006quantum}.

\subsection{Capacities of classical-quantum and quantum channels}
\label{sec:CapQuChannel}

Arguably, the simplest quantum communication scenario is the transmission of classical information via a classical-quantum channel (cq-channel) $T:\mathcal{A}\ra \M_d$. Here, $\mathcal{A}=\lset 1,\ldots ,|\mathcal{A}|\rset$ is a classical input alphabet and for each $i\in\mathcal{A}$ the output $T(i)\in \D(\C^d)$ is a quantum state. We define the \emph{classical communication error} of a classical channel $K:\C^d\ra\C^d$ as
\begin{equation}
\epsilon_{cl}(K) := 1-\frac{1}{d}\sum^{d}_{i=1} \lb K\lb \ket{i}\rb\rb_i,
\label{equ:epsilonCL}
\end{equation}
Now, we can state the following definition:
\begin{defn}[Coding schemes for cq-channels]
For $n,m\in\N$ and $\epsilon\in\R^+$ an $(n,m,\epsilon)$-coding scheme for the cq-channel $T:\mathcal{A}\ra \M_d$ consists of a map $E:\lset 1,\ldots ,2^n\rset\ra \mathcal{A}^{m}$ and a channel $D:\M^{\otimes m}_d\ra \C^{2^n}$ with classical output such that the classical communication error
\[
\epsilon_{cl}\lb D\circ T^{\otimes m}\circ E\rb \leq \epsilon .
\]
\label{defn:CodingSchemeCQ}
\end{defn}

With the previous definition we define the capacity of a cq-channel as follows:

\begin{defn}[Capacity of a cq-channel]
The classical capacity of a cq-channel $T:\mathcal{A}\ra \M_d$ is defined as 
\[
C(T)=\sup\lset R\geq 0 \text{ achievable}\rset,
\]
where $R\geq 0$ is called achievable if for every $m\in\N$ there exists $n_m\in\N$ such that there are $(n_m,m,\epsilon_m)$-coding schemes with $\lim_{m\ra\infty} \epsilon_m=0$ and
\[
R\leq \liminf_{m\ra\infty} \frac{n_m}{m}.
\]
\label{defn:ClCapCQ}
\end{defn}

Given an ensemble of quantum states $\lset p_i,\rho_i\rset_{i}$, i.e., such that $(p_i)_i$ is a probability distribution and $\rho_i$ are quantum states, the Holevo quantity~\cite{holevo1973bounds} is given by
\begin{equation}
\chi\lb \lset p_i,\rho_i\rset_{i}\rb = S\lb\sum_{i} p_i \rho_i\rb - \sum_i p_i S(\rho_i).
\label{equ:HolevoChi}
\end{equation}
Here, $S(\rho)=-\Tr\lb \rho \log_2(\rho)\rb$ denotes the von-Neumann entropy. We now recall the following theorem: 

\begin{thm}[Holevo,Schumacher,Westmoreland~\cite{holevo1998capacity,schumacher1997sending}]
For any cq-channel $T:\mathcal{A}\ra \M_d$ we have
\[
C(T) = \sup_{p\in \mathcal{P}(\mathcal{A})}\chi\lb \lset p_i,T(i)\rset_{i}\rb .
\]
\label{thm:HSWTheorem}
\end{thm}

For quantum channels we will consider the \emph{classical capacity} and the \emph{quantum capacity} quantifying the optimal transmission rates of classical information or quantum information respectively via the quantum channel. Again, we will start by defining the coding schemes considered for these communication problems. In addition to the \emph{classical communication error} from \eqref{equ:epsilonCL} we need to define the \emph{quantum communication error} of a quantum channel $T:\M_d\ra\M_d$ by
\[
\epsilon_{q}(T) := 1-\min_{\ket{\psi}} \bra{\psi}T(\proj{\psi}{\psi})\ket{\psi}, 
\] 
where the minimum is over pure quantum states $\ket{\psi}\in \C^d$ with $\braket{\psi}{\psi}=1$.

\begin{defn}[Coding schemes]
Let $T:\M_{d_A}\ra\M_{d_B}$ denote a quantum channel. For $n,m\in\N$ and $\epsilon\in\R^+$ an $(n,m,\epsilon)$ coding scheme for 
\begin{itemize}
\item \emph{classical communication} consists of a cq-channel $E:\C^{2^n}\ra\M^{\otimes m}_{d_A}$ and a channel $D:\M^{\otimes m}_{d_B}\ra \C^{2^n}$ with classical output such that $\epsilon_{cl}\lb D\circ T^{\otimes m}\circ E\rb \leq \epsilon$.
\item \emph{quantum communication} consists of a quantum channel $E:\M^{\otimes n}_2\ra\M^{\otimes m}_{d_A}$ and a quantum channel $D:\M^{\otimes m}_{d_B}\ra \M^{\otimes n}_2$ such that $\epsilon_{q}\lb D\circ T^{\otimes m}\circ E\rb \leq \epsilon$.
\end{itemize}
\label{defn:codingSchemesQChannel}
\end{defn}

With this we can state the following definition.

\begin{defn}[Classical and quantum capacity]
Let $T:\M_{d_A}\ra\M_{d_B}$ denote a quantum channel. We call $R\geq 0$ an \emph{achievable rate} for classical (quantum) communication if for every $m\in\N$ there exists an $n_m\in\N$ and an $(n_m,m,\epsilon_m)$ coding scheme for classical (quantum) communication with $\epsilon_m\ra 0$ as $m\ra\infty$ and 
\[
\liminf_{m\ra\infty} \frac{n_m}{m}\geq R.
\] 
The \emph{classical capacity} of $T$ is given by
\[
C(T)=\sup\lset R\geq 0 \text{ achievable rate for classical communication}\rset,
\]
and the \emph{quantum capacity} of $T$ is given by 
\[
Q(T)=\sup\lset R\geq 0 \text{ achievable rate for quantum communication}\rset.
\]
\label{defn:CCapAndQCap}
\end{defn}

The classical and quantum capacity can be related to entropic quantities. Given a quantum channel $T:\M_{d_A}\ra\M_{d_B}$ we define
\begin{equation}
\chi(T) := \sup_{\lset p_i,\rho_i\rset_{i}}\chi\lb \lset p_i,T(\rho_i)\rset_{i}\rb ,
\label{equ:HolevoQChannel} 
\end{equation}
with the Holevo quantity from \eqref{equ:HolevoChi} and where the supremum is over all ensembles $\lset p_i,\rho_i\rset_{i}$ with quantum states $\rho_i\in \D\lb\C^{d_A} \rb$. The following theorem is a major achievement of quantum Shannon theory. 

\begin{thm}[Holevo, Schumacher, Westmoreland~\cite{holevo1998capacity,schumacher1997sending}]
For any quantum channel $T:\M_{d_A}\ra\M_{d_B}$ we have
\[
C(T)=\lim_{k\ra\infty} \frac{1}{k}\chi(T^{\otimes k}).
\]
\label{thm:HSWQChannel}
\end{thm}
For a bipartite quantum state $\rho_{AB}\in\D\lb\C^{d_A}\otimes\C^{d_B}\rb$ we define the \emph{coherent information} as
\begin{align*}
I_{\text{coh}}\lb A; B\rb_\rho := S(\rho_B) - S(\rho_{AB}),
\end{align*}  
where $\rho_B\in \D\lb \C^{d_B}\rb$ denotes the marginal $\rho_B = \Tr_A\lb\rho_{AB}\rb$. For a quantum channel $T:\M_{d_A}\ra \M_{d_B}$ we define 
\begin{equation}
I_{\text{coh}}\lb T\rb :=\max_{\rho_{AA'}} I_{\text{coh}}\lb A; B\rb_{(\ident_A\otimes T)(\rho_{AA'})},
\label{equ:cohInfoQChann}
\end{equation}
where the maximum is over quantum states $\rho_{AA'}\in \D\lb \C^{d_A}\otimes \C^{d_A}\rb$. The following theorem is a major achievement of quantum information theory building on earlier results by~\cite{schumacher1996sending,schumacher1996quantum,barnum1998information}:

\begin{thm}[Lloyd, Shor, Devetak~\cite{shor2002quantum,lloyd1997capacity,devetak2005private}]
For any quantum channel $T:\M_{d_A}\ra\M_{d_B}$ we have 
\[
Q(T) = \lim_{k\ra \infty} \frac{1}{k} I_{\text{coh}}\lb T^{\otimes k}\rb.
\]
\label{thm:LSD}
\end{thm}

\section{Fault-tolerant techniques for quantum communication}
\label{sec:FTTechniques}

\subsection{Interfaces for concatenated codes}
\label{sec:NoisyInterface}

To study capacities of quantum channels with encoding and decoding operations protected by quantum error correcting codes we need interfaces between the code space and the physical systems where the quantum channel acts (see Figure~\ref{fig:capacityReal}). For simplicity we will only consider interfaces between physical qubits and the code spaces composed from many qubits. It is straightforward to adapt our definitions and results to quantum error correcting codes using qudits. 

\begin{defn}[Interfaces]
Let $\mathcal{C}\subset (\C^2)^{\otimes K}$ be a stabilizer code with $\text{dim}\lb \mathcal{C}\rb=2$, such that $\proj{0}{0}\in\M^{\otimes (K-1)}_2$ denotes the pure state corresponding to the zero syndrome, and $\Enc^*:\M_2\otimes \M^{\otimes (K-1)}_2\ra \M^{\otimes K}_2$ denotes the ideal encoding operation and $\Dec^*:\M^{\otimes K}_{2}\ra\M_2\otimes \M^{\otimes (K-1)}_2$ denotes the ideal decoding operation (see Section \ref{sec:AnalyzingNoisyQuantumCircuits}). An \emph{interface} for $\mathcal{C}$ is given by an encoding quantum circuit $\Enc:\M_2\ra\M^{\otimes K}_2$ and a decoding quantum circuit $\Dec:\M^{\otimes K}_2\ra\M_2$ such that the following conditions hold:
\begin{enumerate}
\item The quantum circuit $\EC$ from Definition \ref{defn:Impl} is applied at the end of $\Enc$. 
\item We have 
\[
\Dec^*\circ\Enc = \ident_2\otimes \proj{0}{0}.
\]
\item We have 
\[
\Dec\circ \Enc^*(\cdot \otimes\proj{0}{0}) = \ident_2(\cdot).
\]
\end{enumerate} 
\label{defn:Interface}
\end{defn}
The quantum circuits $\Enc$ and $\Dec$ are related to the ideal quantum channels $\Enc^*$ and $\Dec^*$ as specified in the previous definition. However, $\Enc$ and $\Dec$ should be seen as quantum circuits that are implemented physically (and will be subject to noise), while $\Enc^*$ and $\Dec^*$ are ideal (mathematical) objects that only appear in the formal analysis of quantum circuits. Next, we define correctness for interfaces affected by noise.  

\begin{defn}[Correctness of interfaces]
Let $\mathcal{C}\subset (\C^2)^{\otimes K}$ be a stabilizer code with $\text{dim}\lb \mathcal{C}\rb=2$, and let $\Enc:\M_2\ra\M^{\otimes K}_2$ and $\Dec:\M^{\otimes K}_2\ra\M_2$ denote an interface for $\mathcal{C}$.  
\begin{itemize}
\item We say that $\Enc$ is \emph{correct} under a Pauli fault pattern $F$ if 
\begin{equation}
\Dec^*\circ\lbr\Enc\rbr_F = \ident_2\otimes \sigma_F,
\label{equ:CorrInterf1}
\end{equation}
for a quantum state $\sigma_F\in\M^{\otimes (K-1)}_2$ on the syndrome space depending on $F$.
\item We say that the quantum circuit $\Dec\circ \EC$ is \emph{correct} under a Pauli fault pattern $F$ if 
\begin{equation}
\lbr\Dec\circ \EC\rbr_F = (\ident_2\otimes \Tr_S)\circ \Dec^*\circ \lbr\EC\rbr_F,
\label{equ:CorrInterf2}
\end{equation}
where $\Tr_S:\M^{\otimes (K-1)}_2\ra\C$ traces out the syndrome space.
\end{itemize}

\label{defn:CorrectInterface}
\end{defn}

Note that in the above definition, correctness is defined differently for the encoder and the decoder, since we want to use the transformation rules from Definition \ref{defn:ExRecCorr} and Lemma \ref{lem:CircuitTransformExRecCorr} to analyze these interfaces. For example, we need to consider the whole circuit $\Dec\circ \EC$ including an initial error correction in \eqref{equ:CorrInterf2}, since otherwise we could not decompose the interface into extended rectangles and use the notion of well-behavedness as in Definition \ref{defn:ExRecCorr} and Lemma \ref{lem:CircuitTransformExRecCorr}. We do not have this problem for the encoder, since it ends in a final error correction due to Definition \ref{defn:Interface}.  

Interfaces should be robust against noise when they are used for quantum information processing. Clearly, the probability of a fault occuring in an interface between a quantum error correcting code and a single qubit is at least as large as the probability of a fault in a single qubit operation, which could happen at the end of $\Dec$ or at the beginning of $\Enc$. Fortunately, it is possible for concatenated quantum error correcting codes to construct interfaces that only fail with a probability slightly larger than this lower bound. This result has previously been established in~\cite{mazurek2014long} (see also~\cite{lodyga2015simple} for a similar discussion for various topological codes). Since the proof given in~\cite{mazurek2014long} seems to neglect certain details, we will here present a more rigorous version of the argument using the formalism stated in the previous section. It should be emphasized that the main ideas of the argument are the same as in~\cite{mazurek2014long}.   

\begin{thm}[Nice interface for concatenated $7$-qubit code] 
For $l\in\N$ we denote by $\mathcal{C}_l\subset (\C^2)^{\otimes 7^l}$ the $l$th level of the concatenated $7$-qubit Steane code with threshold $p_0\in\lb 0,1\rbr$ (see Lemma \ref{lem:Strongthreshold}). There exists a constant $c>0$, encoding quantum circuits $\Enc_l:\M_2\ra\M^{\otimes 7^l}_2$ and decoding quantum circuits $\Dec_l:\M^{\otimes 7^l}_2\ra \M_2$ for each $l\in\N$ forming an interface for the stabilizer code $\mathcal{C}_l$, such that for any $0\leq p\leq p_0/2$ we have 
\[
\text{P}\lb \lbr\Enc_l\rbr_F\text{ not correct}\rb \leq 2c p,
\]
and 
\[
\text{P}\lb \lbr\Dec_l\rbr_F\text{ not correct}\rb \leq 2c p,
\]
where $\text{P}$ denotes the probability distribution over Pauli fault patterns induced by the fault model $\mathcal{F}_{\PIID}(p)$.

\label{thm:NiceInterfaceConcatenated}
\end{thm}

The proof of the previous theorem will construct the desired interfaces successively by defining partial encoding (and decoding) quantum circuits between different levels of the code. Before presenting the proof, we will state some lemmas: 

\begin{lem}[Successive interfaces]
Let $\lb\mathcal{C}_l\rb_{l\in\N}$ be a family of stabilizer codes such that for every $l\in\N$ we have $\mathcal{C}_l\subset (\C^2)^{\otimes K_l}$, $\text{dim}\lb \mathcal{C}_l\rb=2$, and the state $\proj{0}{0}_l\in\M^{\otimes (K_l-1)}_2$ denotes the pure state corresponding to the zero syndrome of $\mathcal{C}_l$. Assume that for every $l\in\N$ there are quantum circuits $\Enc_{(l-1)\ra l}:\M^{\otimes K_{l-1}}_2\ra \M^{\otimes K_{l}}_2$ and $\Dec_{l\ra (l-1)}:\M^{\otimes K_{l}}_2\ra \M^{\otimes K_{l-1}}_2$, where $K_0:=1$, with the following properties: 
\begin{enumerate}
\item The pair $\Enc_{0\ra 1}$ and $\Dec_{1\ra 0}$ is an interface for $\mathcal{C}_1$.
\item For any $l\geq 2$ we have
\[
\Dec^*_l\circ \Enc_{(l-1)\ra l}\circ \Enc^*_{l-1}\lb \cdot\otimes \proj{0}{0}_{l-1}\rb = \ident_2(\cdot)\otimes \proj{0}{0}_{l}.
\]
\item For any $l\geq 2$ we have
\[
\Dec_{l\ra (l-1)} \circ\Enc^*_{l}\lb \cdot\otimes \proj{0}{0}_{l}\rb = \Enc^*_{l-1}(\cdot \otimes \proj{0}{0}_{l-1}).
\] 
\end{enumerate}
Then, for any $l\in\N$ the quantum circuits $\Enc_l:\M_2\ra\M^{\otimes K_l}_2$ and $\Dec_l:\M^{\otimes K_l}_2\ra \M_2$ given by 
\[
\Enc_l := \Enc_{(l-1)\ra l} \circ \cdots \circ \Enc_{1\ra 2} \circ \Enc_{0\ra 1}
\]
and 
\[
\Dec_l := \Dec_{1\ra 0}\circ \cdots \circ \Dec_{(l-1)\ra (l-2)} \circ \Dec_{l\ra (l-1)}
\]
are an interface for $\mathcal{C}_l$ according to Definition \ref{defn:Interface}.
\label{lem:SuccInterface}
\end{lem}

\begin{proof}
The proof proceeds by induction on $l\in\N$. For $l=1$ we have $\Enc_1=\Enc_{0\ra 1}$ and $\Dec_1=\Dec_{1\ra 0}$, which is an interface for $\mathcal{C}_1$ by assumption. By definition we have
\[
\Enc_{l+1} = \Enc_{l\ra (l+1)}\circ\Enc_l,
\] 
and 
\[
\Dec_{l+1} = \Dec_{l}\circ\Dec_{(l+1)\ra l}.
\]
Assuming that for $l\geq 1$ the pair of quantum circuits $\Enc_l$ and $\Dec_l$ is an interface for $\mathcal{C}_l$ we find that 
\begin{align*}
\Dec_{l+1}^*\circ\Enc_{l+1} &= \Dec_{l+1}^*\circ\Enc_{l\ra (l+1)}\circ\Enc_l \\
&= \Dec_{l+1}^*\circ\Enc_{l\ra (l+1)}\circ \Enc_l^*\circ \Dec_l^*\circ\Enc_l \\
&= \Dec_{l+1}^*\circ\Enc_{l\ra (l+1)}\circ \Enc_l^*(\cdot \otimes \proj{0}{0}_{l}) \\
&= \ident_2\otimes \proj{0}{0}_{l+1}
\end{align*}
where we used that $\Enc_l^*$ and $\Dec_l^*$ are inverse to each other in the second equality sign and the induction hypothesis together with the Property $2$ from Definition \ref{defn:Interface} in the third equality sign. In the last step, we used the Assumption $2$ on $\Enc_{l\ra (l+1)}$. Similarly, we find that 
\begin{align*}
\Dec_{l+1}\circ \Enc_{l+1}^*(\cdot \otimes\proj{0}{0}_{l+1}) &= \Dec_{l}\circ\Dec_{(l+1)\ra l}\circ \Enc_{l+1}^*(\cdot \otimes\proj{0}{0}_{l+1}) \\
&= \Dec_{l}\circ\Enc^*_{l}(\cdot \otimes \proj{0}{0}_{l}) \\
&= \ident_2(\cdot),
\end{align*}
where we used Assumption $3$ in the second equality sign and the inductive hypothesis in the last equality sign. This shows that the pair $\Enc_l$ and $\Dec_l$ form an interface for $\mathcal{C}_l$ according to Definition \ref{defn:Interface}.
\end{proof}

To construct an interface for the concatenated $7$-qubit Steane code, we will first construct interface circuits $\Enc_{0\ra 1}$ and $\Dec_{1\ra 0}$ for the first level, i.e., for the $7$-qubit Steane code. By implementing $\Enc_{0\ra 1}$ and $\Dec_{1\ra 0}$ in higher levels of the concatenated code, we will then obtain general interface circuits. The construction outlined below works in general whenever the quantum circuits $\Enc_{0\ra 1}$ and $\Dec_{1\ra 0}$ define an interface for the $7$-qubit Steane code, and only the size of these circuits will appear in the main results. For convenience we have included an explicit construction of interfaces $\Enc_{0\ra 1}$ and $\Dec_{1\ra 0}$ in Appendix \ref{Sec:InterfaceExpl} using a simple teleportation circuit.  

In the following, let $\Enc_{0\ra 1}$ and $\Dec_{1\ra 0}$ denote a fixed interface for the $7$-qubit Steane code. For $l\in\N$ we define the following quantum circuits as implementations as in Definition \ref{defn:Impl}:
\begin{equation}
\Enc_{(l-1)\ra l}:=\lb \Enc_{0\ra 1}\rb_{\text{C}_{l-1}} \hspace{1cm}\text{ and }\hspace{1cm}\Dec_{l\ra (l-1)}:=\lb \Dec_{1\ra 0}\rb_{\text{C}_{l-1}}.
\label{equ:EncDecPartial}
\end{equation}
Using \eqref{equ:EncRelation} from Appendix \ref{sec:CodeCon} it is easy to verify the conditions of Lemma \ref{lem:SuccInterface} for the quantum circuits from \eqref{equ:EncDecPartial}. Therefore, for any $l\in\N$, the quantum circuits
\begin{equation}
\Enc_l := \Enc_{(l-1)\ra l} \circ \cdots \circ\Enc_{1\ra 2} \circ \Enc_{0\ra 1}
\label{equ:EncDec1}
\end{equation}
and 
\begin{equation}
\Dec_l := \Dec_{1\ra 0}\circ \cdots \circ \Dec_{(l-1)\ra (l-2)} \circ \Dec_{l\ra (l-1)}
\label{equ:EncDec2}
\end{equation}
form an interface for the $l$-th level of the concatenated $7$-qubit Steane code. It remains to analyze how this interface behaves under the i.i.d.~Pauli noise model as introduced after Definition \ref{defn:NoiseModel}. It will be helpful to extend the notion of well-behavedness from Definition \ref{defn:ExRecCorr} to interfaces: 

\begin{defn}[Well-behaved interfaces]
For each $i\in\lset 0,\ldots ,l-1\rset$ we denote by $\Enc_{i\ra (i+1)}$ and $\Dec_{(i+1)\ra i}$ the quantum circuits from \eqref{equ:EncDecPartial} and by $\EC_{i}$ the error correction of the $i$th level of the concatenated $7$-qubit Steane code, where we set $\EC_0=\ident_2$.  
\begin{enumerate}
\item We will call the quantum circuit $\Enc_{i\ra (i+1)} \circ\EC_{i}$ \emph{well-behaved} under the Pauli fault pattern $F$ if all $i$-extended rectangles in $\lbr\Enc_{i\ra (i+1)} \circ\EC_{i}\rbr_F$ are well-behaved.  
\item We will call the quantum circuit $\Dec_{(i+1)\ra i} \circ\EC_{i+1}$ \emph{well-behaved} under the Pauli fault pattern $F$ if all $i$-extended rectangles in $\lbr\Dec_{(i+1)\ra i} \circ\EC_{i+1}\rbr_F$ are well-behaved.
\end{enumerate}
Moreover, for $l\in\N$ we will call the quantum circuit $\Enc_l$ or $\Dec_l\circ \EC_l$
\emph{well-behaved} under the Pauli fault pattern $F$ if the partial circuits $\Enc_{i\ra (i+1)} \circ\EC_{i}$ or $\Dec_{(i+1)\ra i} \circ\EC_{i+1}$ respectively are well-behaved under the restrictions of $F$ for every $i\in\lset 0,\ldots ,l-1\rset$. 
\label{defn:GoodFaultPatt}
\end{defn}

The following lemma gives an upper bound on the probability that interfaces are well-behaved. 

\begin{lem}[Probability of an interface to be well-behaved]
For $l\in\N$ consider $\Gamma_l\in\lset\Enc_l,\Dec_l\circ \EC_l \rset$ from \eqref{equ:EncDec1} and \eqref{equ:EncDec2}. Then, we have 
\[
\text{P}\lb \Gamma_l\text{ not well-behaved under }F\rb\leq cp_0\sum^{l-1}_{i=0}\lb \frac{p}{p_0}\rb^{2^i}
\]
where $p_0\in \lb 0,1\rbr$ denotes the threshold probability (see Lemma \ref{lem:Strongthreshold}), and where 
\[
c = \max\lb |\Loc\lb \Enc_{0\ra 1}\rb|,|\Loc\lb \Dec_{1\ra 0}\circ \EC_1\rb|\rb,
\] 
and $\text{P}$ denotes the probability distribution over Pauli fault patterns induced by the fault model $\mathcal{F}_{\PIID}(p)$ and restricted to $\Gamma_l$.
\label{lem:Goodness}
\end{lem} 

\begin{proof}
For $l\in\N$ we will state the proof in the case $\Gamma_l = \Enc_l$. The remaining case follows in the same way. Consider the fault model $\mathcal{F}_{\PIID}(p)$ restricted to $\Enc_l$. Let $N_i\in\N$ denote the number of $i$-extended rectangles in the quantum circuit $\Enc_{i\ra (i+1)} \circ\EC_{i}$ for $i\in\lset 0,\ldots, l-1\rset$, where $0$-extended rectangles are the elementary operations and $\EC_0=\ident_2$. By \eqref{equ:EncDecPartial} we have $N_i = N_0 = |\Loc\lb \Enc_{0\ra 1}\rb|$, and by the union bound and Lemma \ref{lem:Strongthreshold} we have 
\[
\text{P}\lb \text{All } i\text{-extended rectangles in } \lbr\Enc_{i\ra (i+1)} \circ\EC_{i}\rbr_F \text{ are well-behaved}\rb\leq p_0 N_0\lb \frac{p}{p_0}\rb^{2^i}
\]
for any $i\in\lset 0,\ldots, l-1\rset$. By comparing with Definition \ref{defn:GoodFaultPatt} we find that 
\[
\text{P}\lb \Enc_l~\text{not well-behaved}\rb\leq cp_0\sum^{l-1}_{i=0}\lb \frac{p}{p_0}\rb^{2^i},
\] 
where $c = |\Loc\lb \Enc_{0\ra 1}\rb|$.

\end{proof}

The following lemma analyzes how the successive interfaces defined above behave under the fault patterns introduced in Defintion \ref{defn:GoodFaultPatt}. 

\begin{lem}[Successive interfaces under noise]
For each $l\in\N$ let $\mathcal{C}_l\subset (\C^2)^{\otimes 7^l}$ denote the $l$th level of the concatenated $7$-qubit Steane code, and let $\Enc_{(l-1)\ra l}:\M^{\otimes 7^{l-1}}_2\ra \M^{\otimes 7^{l}}_2$ and $\Dec_{l\ra (l-1)}:\M^{\otimes 7^{l}}_2\ra \M^{\otimes 7^{l-1}}_2$ denote the quantum circuits from \eqref{equ:EncDecPartial}. Whenever the quantum circuit $\Enc_{(l-1)\ra l} \circ\EC_{l-1}$ or $\Dec_{l\ra (l-1)} \circ\EC_{l}$ is well-behaved under a Pauli fault pattern $F$ we have 
\[
\Dec^*_l\circ \lbr\Enc_{(l-1)\ra l} \circ\EC_{l-1}\rbr_F = (\ident_2\otimes S^{(l-1)\ra l}_F)\circ\Dec^*_{l-1}\circ\lbr\EC_{l-1}\rbr_F
\] 
or
\[
\Dec^*_{l-1}\circ \lbr\Dec_{l\ra (l-1)} \circ\EC_{l}\rbr_F = (\ident_2\otimes S^{l\ra (l-1)}_F)\circ\Dec^*_{l}\circ\lbr\EC_{l}\rbr_F,
\] 
respectively, for quantum channels $S^{(l-1)\ra l}_F:\M^{\otimes(K_{l-1}-1)}_2\ra\M^{\otimes (K_l-1)}_2$ and $S^{l\ra (l-1)}_F:\M^{\otimes(K_{l}-1)}_2\ra\M^{\otimes (K_{l-1}-1)}_2$ between the syndrome spaces of $\mathcal{C}_{l-1}$ and $\mathcal{C}_l$. 
 
\label{lem:InterfaceLemma}
\end{lem}

\begin{proof}
Recall that by construction of the concatenated $7$-qubit Steane code (see \eqref{equ:defnDecl}) we have 
\begin{equation}
\Dec^{*}_{l} = \lb\Dec^{*}_{1} \otimes \ident^{\otimes 7(7^{l-1}-1)}_2\rb\circ \lb\Dec^{*}_{l-1}\rb^{\otimes 7}.
\label{equ:DecoderLevelDecomp}
\end{equation} 
Let $F$ denote a Pauli fault pattern under which every $(l-1)$-extended rectangle in $\lbr\Enc_{(l-1)\ra l} \circ\EC_{l-1}\rbr_F$ is well-behaved. Treating $\Enc_{(l-1)\ra l} \circ\EC_{l-1}$ as the implementation of a quantum circuit at level $(l-1)$ of the concatenated $7$ qubit Steane code outputting $7$ registers, we have 
\[
\lb\Dec^*_{l-1}\rb^{\otimes 7}\circ \lbr\Enc_{(l-1)\ra l} \circ\EC_{l-1}\rbr_F = \lb\Enc_{0\ra 1}\otimes S_F\rb\circ \Dec^*_{l-1}\circ \lbr\EC_{l-1}\rbr_F,
\]
by Lemma \ref{lem:CircuitTransformExRecCorr} and for some quantum channel $S_F:\M^{\otimes (7^{l-1}-1)}_2\ra \M^{\otimes 7(7^{l-1}-1)}_2$. Using \eqref{equ:DecoderLevelDecomp} and that $\Enc_{0\ra 1}$ is an interface for $\mathcal{C}_1$ (see Definition \ref{defn:Interface}) we find that 
\begin{align*}
\Dec^*_{l}\circ \lbr\Enc_{(l-1)\ra l} \circ\EC_{l-1}\rbr_F &= \lb \Dec^*_{1}\circ\Enc_{0\ra 1}\otimes S_F\rb\circ \Dec^*_{l-1}\circ \lbr\EC_{l-1}\rbr_F \\
&=\lb \ident_2\otimes \tilde{S}_F\rb\circ \Dec^*_{l-1}\circ \lbr\EC_{l-1}\rbr_F,
\end{align*}
where $\tilde{S}_F = S_F\otimes \proj{0}{0}$ for the pure state $\proj{0}{0}\in\M^{\otimes 6}_2$ corresponding to the $0$ syndrome of $\mathcal{C}_1$. This verifies the first statement of the lemma.

For the second statement, let $F$ denote a Pauli fault pattern under which every $(l-1)$-extended rectangle in $\lbr\Dec_{l\ra (l-1)} \circ\EC_{l}\rbr_F$ is well-behaved. Since $\EC_l = \lb\EC_1\rb_{\mathcal{C}_{l-1}}$ the quantum circuit $\EC_{l}$ has error corrections $\EC^{\otimes 7}_{l-1}$ in its final step, and we can decompose $\EC_{l} = \EC^{\otimes 7}_{l-1}\circ \Gamma$ for some quantum circuit $\Gamma$. Treating $\Dec_{l\ra (l-1)} \circ\EC^{\otimes 7}_{l-1}$ as the implementation of a quantum circuit at level $l-1$ with $7$ input registers, we find
\begin{align*}
\Dec^*_{l-1}\circ\lbr\Dec_{l\ra (l-1)} \circ\EC_l\rbr_F &= \Dec^*_{l-1}\circ\lbr\Dec_{l\ra (l-1)} \circ\EC^{\otimes 7}_{l-1}\circ \Gamma\rbr_F \\
& = \lb\Dec_{1\ra 0}\otimes S_F\rb \circ \lb\Dec^*_{l-1}\rb^{\otimes 7}\circ \lbr \EC_{l}\rbr_F,
\end{align*}   
for some quantum channel $S_F:\M^{\otimes 7(7^{l-1}-1)}_2\ra \M^{\otimes (7^{l-1}-1)}_2$. Finally, we can use \eqref{equ:DecoderLevelDecomp} and that $\Dec_{1\ra 0}$ is an interface for $\mathcal{C}_1$ (see Definition \ref{defn:Interface}) to prove
\begin{align*}
\Dec^*_{l-1}\circ\lbr\Dec_{l\ra (l-1)} \circ\EC_l\rbr_F &= \lb\lb\Dec_{1\ra 0}\circ \Enc^*_{1}\circ \Dec^*_{1}\rb\otimes S_F\rb \circ \lb\Dec^*_{l-1}\rb^{\otimes 7}\circ \lbr \EC_{l}\rbr_F \\
&= \lb\ident_2\otimes S_F\rb \circ \Dec^*_{l}\circ \lbr \EC_{l}\rbr_F.
\end{align*}

\end{proof}

Now, we are in the position to prove Theorem \ref{thm:NiceInterfaceConcatenated}. 

\begin{proof}[Proof of Theorem \ref{thm:NiceInterfaceConcatenated}]

For each $l\in \N$ consider the interfaces given by $\Enc_l$ and $\Dec_l$ from \eqref{equ:EncDec1} and \eqref{equ:EncDec2} defined via the quantum circuits $\Enc_{i\ra (i+1)}$ and $\Dec_{(i+1)\ra i}$ from \eqref{equ:EncDecPartial}. We will now show that $\Enc_l$ is correct as in Definition \ref{defn:CorrectInterface} under a fault pattern $F$ whenever it is well-behaved under $F$ as in Definition \ref{defn:GoodFaultPatt}. Our proof will use induction on the level $l\in\N$. Clearly, we have 
\begin{equation}
\Dec^*_1\circ\lbr\Enc_{1}\rbr_F=\Dec^*_1\circ\lbr\Enc_{0\ra 1}\rbr_F = \ident_2\otimes \proj{0}{0}_1, 
\label{equ:EncStep1}
\end{equation}
for any fault pattern $F$ under which $\Enc_{1}$ is well-behaved, since then every elementary gate in $\Enc_{0\ra 1}$ is ideal by Definition \ref{defn:GoodFaultPatt}. Next, consider a fault pattern $F$ under which $\Enc_{l+1}$ is well-behaved for some $l\in\N$. Note that by Definition \ref{defn:GoodFaultPatt} also $\Enc_{l}$ (arising as a part of $\Enc_{l+1}$) is well-behaved under $F$ when restricted to that quantum circuit. Therefore, we can first apply Lemma \ref{lem:InterfaceLemma} and then the induction hypothesis to compute
\begin{align}
\Dec^*_{l+1}\circ\lbr\Enc_{l+1}\rbr_F &= \Dec^*_{l+1}\circ\lbr\Enc_{l\ra (l+1)}\circ\EC_{l} \circ\tilde{\Enc}_{l} \rbr_F \\
&= (\ident_2\otimes S^{l\ra (l+1)}_F)\circ\Dec^*_{l}\circ\lbr\Enc_{l}\rbr_F \\
& = \ident_2\otimes \sigma_F.
\end{align}
Here, $\tilde{\Enc}_{l}$ denotes the quantum circuit obtained from $\Enc_{l}$ by removing the final error correction $\EC_{l}$, and 
\[
S^{l\ra (l+1)}_F:\M^{\otimes(K_{l}-1)}_2\ra\M^{\otimes (K_{l+1}-1)}_2
\]  
denotes a quantum channel between the syndrome spaces of $\mathcal{C}_{l}$ and $\mathcal{C}_{l+1}$, and 
\[
\sigma_F = \prod^{l-1}_{i=1}S_F^{i\ra (i+1)}(\proj{0}{0}_1)
\]
is the final syndrome state. This shows that $\lbr\Enc_l\rbr_F$ is correct as in Definition \ref{defn:CorrectInterface} whenever $\Enc_l$ is well-behaved under $F$ according to Definition \ref{defn:GoodFaultPatt}. By Lemma \ref{lem:Goodness} we find that 
\[
\text{P}\lb \lbr\Enc_l\rbr_F\text{ not correct}\rb\leq \text{P}\lb \Enc_l\text{ not well-behaved under } F\rb\leq cp_0\sum^{l-1}_{i=0}\lb \frac{p}{p_0}\rb^{2^i} ,
\]
where 
\[
c = \max\lb |\Loc\lb \Enc_{0\ra 1}\rb|,|\Loc\lb \Dec_{1\ra 0}\circ \EC_1\rb|\rb.
\]
In the case where $0\leq p\leq p_0/2$ we upper bound the previous sum using a geometric series and obtain
\[
\text{P}\lb \lbr\Enc_l\rbr_F\text{ not correct}\rb\leq 2cp. 
\]

To deal with $\Dec_l$ we will again employ induction on the level $l\in\N$ to show that the quantum circuit $\Dec_l\circ\EC_l$ is correct under a Pauli fault pattern $F$ as in Definition \ref{defn:CorrectInterface} whenever it is well-behaved under $F$ as in Definition \ref{defn:GoodFaultPatt}. Let $F$ denote a Pauli fault pattern under which $\Dec_1\circ\EC_1$ is well-behaved and note that by Definition \ref{defn:GoodFaultPatt} every elementary gate is then ideal. Using that an error correction without any faults coincides with the ideal error correction from \eqref{equ:IdealEC} as a quantum channel, we find  
\begin{align*}
\lbr\Dec_1\circ \EC_1\rbr_F &= \Dec_{1\ra 0}\circ \EC^*_1 \\
&= \Dec_{1\ra 0}\circ \EC^*_1\circ\EC^*_1  \\
&= \Dec_{1\ra 0}\circ \Enc_1^*\circ (\ident_2\otimes \proj{0}{0}_1\Tr)\circ \Dec_1^* \circ \EC^*_1 \\
& = (\ident_2\otimes \Tr)\circ \Dec_1^* \circ \lbr\EC_1\rbr_F,
\end{align*}
where we used that the ideal error correction $\EC^*_1$ is a projection (see \eqref{equ:IdealEC}) and that $\Dec_{1\ra 0}$ is an interface (cf.~Definition \ref{defn:Interface}). Now, consider a fault pattern $F$ under which $\lbr \Dec_{l+1}\circ \EC_{l+1}\rbr_F$ is well-behaved for some $l\in\N$. By Definition \ref{defn:GoodFaultPatt}, $\Dec_{l}\circ \EC_{l}$ (arising as a part of the quantum circuit $\Dec_{l+1}\circ \EC_{l+1}$) is well-behaved under the restriction of $F$. Applying Lemma \ref{lem:InterfaceLemma} and the induction hypothesis we compute 
\begin{align*}
\lbr \Dec_{l+1}\circ \EC_{l+1}\rbr_F &= \lbr \Dec_{l}\circ \EC_{l}\circ \tilde{\Dec}_{(l+1)\ra l}\circ \EC_{l+1}\rbr_F \\
&= (\ident_2\otimes \Tr_S)\circ \Dec^*_{l}\circ \lbr \Dec_{(l+1)\ra l}\circ \EC_{l+1}\rbr_F\\
& = (\ident_2\otimes \Tr_S)\circ(\ident_2\otimes S^{(l+1)\ra l}_F)\circ\Dec^*_{l+1}\circ\lbr\EC_{l+1}\rbr_F\\
& = (\ident_2\otimes \Tr_S)\circ\Dec^*_{l+1}\circ\lbr\EC_{l+1}\rbr_F,
\end{align*}
where $\tilde{\Dec}_{(l+1)\ra l}$ in the first line denotes the quantum circuit obtained from $\Dec_{(l+1)\ra l}$ by removing the final error correction $\EC_{l}$. By induction it follows that $\lbr \Dec_l\circ \EC_l\rbr_F$ is correct as in Definition \ref{defn:CorrectInterface} whenever $\Dec_l\circ \EC_l$ is well-behaved under $F$. Again, we can apply Lemma \ref{lem:Goodness} to conclude
\[
\text{P}\lb \lbr\Dec_l\rbr_F\text{ not correct}\rb\leq \text{P}\lb \Dec_l\circ \EC_l\text{ not well-behaved under }F\rb\leq cp_0\sum^{l-1}_{i=0}\lb \frac{p}{p_0}\rb^{2^i} 
\]
where 
\[
c = \max\lb |\Loc\lb \Enc_{0\ra 1}\rb|,|\Loc\lb \Dec_{1\ra 0}\circ \EC_1\rb|\rb.
\]
Finally, whenever $0\leq p\leq p_0/2$ we can upper bound the previous sum using a geometric series as before and obtain
\[
\text{P}\lb \lbr\Enc_l\rbr_F\text{ not correct}\rb\leq 2cp. 
\]

\end{proof}

\subsection{Encoding and decoding of concatenated codes and effective communication channels}
\label{sec:EncDecEffComm}

To study quantum capacities in a fault-tolerant setting we will need to consider quantum circuits where some of the lines are replaced by quantum channels describing stronger noise than the usual noise model. These lines might describe wires connecting physically separate locations where quantum computers are located. Naturally, the noise affecting a qubit during transmission through such a wire might be much larger than the noise affecting the gates locally in the quantum computer. Here we will develop a framework to deal with this situation. We will start with a lemma combining the quantum interfaces introduced in the previous section with general quantum circuits. 

\begin{lem}[Noisy interfaces and quantum circuits]
Let $\Gamma^1:\C^{2^N}\ra\M^{\otimes m}_2$ be a quantum circuit with classical input, and $\Gamma^2:\M^{\otimes m}_2\ra\C^{2^M}$ be a quantum circuit with classical output. For each $l\in\N$ let $\mathcal{C}_l\subset (\C^2)^{\otimes 7^l}$ denote the $l$th level of the concatenated $7$-qubit Steane code with threshold $p_0\in\lb 0,1\rbr$ (see Lemma \ref{lem:Strongthreshold}). Moreover, we denote by $\Enc_l:\M_2\ra\M^{\otimes 7^l}_2$ and $\Dec_l:\M^{\otimes 7^l}_2\ra \M_2$ the interface circuits from \eqref{equ:EncDec1} and \eqref{equ:EncDec2}, and by $c>0$ the constant from Theorem \ref{thm:NiceInterfaceConcatenated}. Then, the following two statements hold for any $0\leq p\leq p_0/2$:
\begin{enumerate}
\item For any $l\in\N$ there exists a quantum channel $N_l:\M_2\otimes \M^{\otimes (7^l-1)}_2\ra \M_2$ acting on a data qubit and the syndrome space and only depending on $l$ and the interface circuit $\Dec_l$, and a quantum state $\sigma_S$ on the syndrome space such that 
\begin{align*}
&\Big\|\lbr \Dec^{\otimes m}_l\circ \Gamma^1_{\mathcal{C}_l}\rbr_{\mathcal{F}_\pi(p)} - \lb N^{\text{dec},l}_{p}\rb^{\otimes m}\circ \lb\Gamma^{1}\otimes \sigma_S\rb\Big\|_{1\ra 1}\leq 2C_1\lb \frac{p}{p_0}\rb^{2^{l}}|\Loc(\Gamma^1)| + 2C_2\lb \frac{p}{p_0}\rb^{2^{l-1}}m ,
\end{align*}
where $N^{\text{dec},l}_{p}:\M_2\otimes \M^{\otimes (7^l-1)}_2\ra\M_2$ is the quantum channel given by
\[
N^{\text{dec},l}_{p} = (1-2cp)\ident_2\otimes \Tr_{S_l} + 2cp N_l
\]
acting on a data qubit and the syndrome space.
\item For any $l\in\N$ there exists a quantum channel $N'_l:\M_2\ra\M_2$ only dependent on $l$ and the interface circuit $\Enc_l$ such that  
\begin{align*}
&\Big\|\lbr \Gamma^2_{\mathcal{C}_l} \circ\Enc^{\otimes m}_l\rbr_{\mathcal{F}_\PIID(p)} - \Gamma^2 \circ\lb N^{\text{enc},l}_{p}\rb^{\otimes m}\Big\|_{1\ra 1}\leq 2C\lb \frac{p}{p_0}\rb^{2^{l}}|\Loc(\Gamma^2)|,
\end{align*}
where $N^{\text{enc},l}_{p}:\M_2\ra\M_2$ is the quantum channel given by
\[
N^{\text{enc},l}_{p} = (1-2cp)\ident_2 + 2cp N'_l .
\]
\end{enumerate}
In the above, $C_1,C_2,C>0$ denote constants not depending on $m,l$ or the quantum circuits involved.
\label{lem:Main}
\end{lem}

The proof of Lemma \ref{lem:Main} will be based on Theorem \ref{thm:NiceInterfaceConcatenated} showing that the probability of an interface not being correct under the noise model $\mathcal{F}_\PIID(p)$ is upper bounded by an expression linear in $p$. However, a major difficulty arises from the fact, that the two notions ``well-behavedness'' (of the quantum circuits $\Gamma^1_{\mathcal{C}_l}$ and $\Gamma^2_{\mathcal{C}_l}$) and ``correctness'' (of the interfaces $\Enc_l$ and $\Dec_l$) refer to overlapping parts of the combined circuits $\Dec^{\otimes m}_l\circ \Gamma^1_{\mathcal{C}_l}$ and $\Gamma^2_{\mathcal{C}_l} \circ\Enc^{\otimes m}_l$. To be precise, the circuits in question overlap in joined error corrections. To obtain the i.i.d.~structure of Lemma \ref{lem:Main} we have to deal with this overlap, which will take the largest part of the following proof. It should be noted that doing so is slightly more difficult in the first part of the proof considering the interface $\Dec_l$.   
 
\begin{proof}\hfill 
\begin{enumerate}
\item The quantum circuit $\Gamma^1_{\mathcal{C}_l}$ ends in error corrections according to Definition \ref{defn:Impl}, and we can write
\[
\Dec^{\otimes m}_l\circ \Gamma^1_{\mathcal{C}_l} =\lb\Dec_l\circ \EC_l\rb^{\otimes m}\circ \tilde{\Gamma}^1_{\mathcal{C}_l},
\] 
for some quantum circuit $\tilde{\Gamma}^1_{\mathcal{C}_l}$. Every fault pattern $F$ affecting the circuit $\lb\Dec_l\circ \EC_l\rb^{\otimes m}\circ \tilde{\Gamma}^1_{\mathcal{C}_l}$ can be decomposed into $F = F_1 \oplus F_2 \oplus F_3 $ with fault pattern $F_1$ affecting $\Dec^{\otimes m}_l$, fault pattern $F_2$ affecting $\EC^{\otimes m}_l$ and $F_3$ affecting $\tilde{\Gamma}^1_{\mathcal{C}_l}$. Because $F_1$ and $F_2$ are local faults acting on tensor products of quantum circuits, we can decompose $F_1=F^1_1\oplus \cdots \oplus F^m_1$ and $F_2=F^1_2\oplus \cdots \oplus F^m_2$, where the fault pattern $F^k_i$ acts on the $k$th line of the $m$-fold tensor product $\lb\Dec_l\circ \EC_l\rb^{\otimes m}$. See Figure \ref{fig:faultPatternsDec} for a schematic of how the fault patterns are labeled. 

\begin{figure*}[hbt!]
        \center
        \includegraphics[scale=0.6]{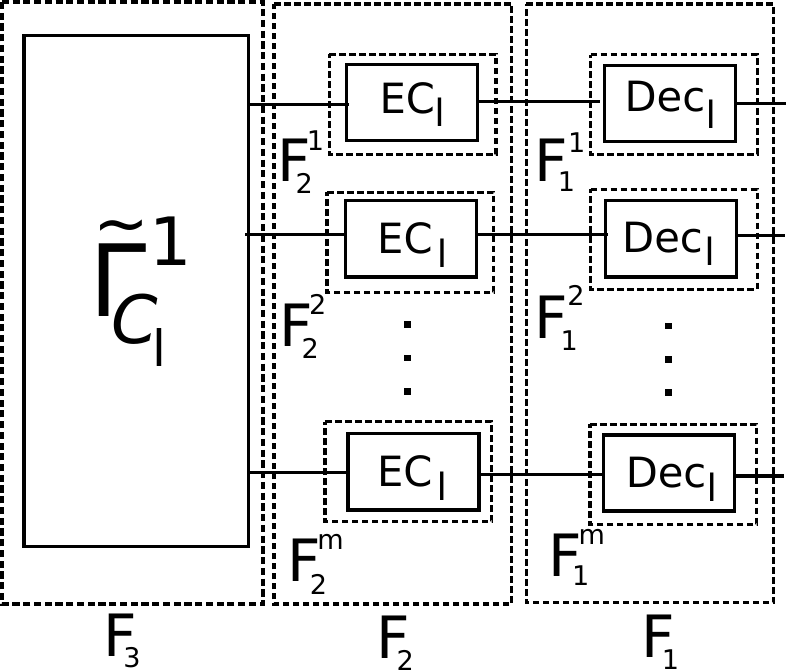}
        \caption{Partition of fault patterns in the first part of the Proof of Lemma \ref{lem:Main}.}
        \label{fig:faultPatternsDec}
\end{figure*}

Let $\mathcal{B}$ denote the set of all fault patterns $F_2\oplus F_3$ such that every $l$-extended rectangle in the quantum circuit
\[
\lbr\lb\EC_l\rb^{\otimes m}\circ \tilde{\Gamma}^1_{\mathcal{C}_l}\rbr_{F_2\oplus F_3}
\]
is well-behaved. By the threshold lemma (Lemma \ref{lem:Strongthreshold}) and the union bound we have 
\begin{equation}
\epsilon = P\lb F_2\oplus F_3\notin \mathcal{B}\rb \leq C_1\lb \frac{p}{p_0}\rb^{2^{l}}|\Loc(\Gamma^1)|,
\label{equ:ImpLem1}
\end{equation}
and by Lemma \ref{lem:CircuitTransformExRecCorr} (for $n=0$) we have 
\begin{equation}
\lb \Dec^*_l\rb^{\otimes m}\circ \lbr\EC^{\otimes m}_l\circ\tilde{\Gamma}^1_{\mathcal{C}_l}\rbr_{F_2\oplus F_3} = \Gamma^{1}\otimes \sigma_{F_2\oplus F_3}
\label{equ:ImpLem2}
\end{equation}
for any $F_2\oplus F_3\in \mathcal{B}$, and where $\sigma_{F_2\oplus F_3}$ denotes some quantum state on the syndrome space. For every $k\in\lset 1,\ldots ,m\rset$ we define the function
\[
\beta\lb F^k_1\oplus F^k_2\rb=\begin{cases} 0 \quad\text{ if }\Dec_l\circ \EC_l\text{ well-behaved under } F^k_1\oplus F^k_2 \\
1 \quad\text{ otherwise}\end{cases},
\]
where ``well-behaved'' is used as in Definition \ref{defn:GoodFaultPatt}. Since any well-behaved interface is correct as in Definition \ref{defn:CorrectInterface} (cf.~proof of Theorem \ref{thm:NiceInterfaceConcatenated}) we can transform the circuit $\lbr \Dec_l\circ \EC_l \rbr_{F^k_1\oplus F^k_2}$ as described in Definition \ref{defn:CorrectInterface} whenever $\beta\lb F^k_1\oplus F^k_2\rb=0$. If $\beta\lb F^k_1\oplus F^k_2\rb=1$, then we can insert an identity $\ident=\Enc^*_l\circ \Dec^*_l$ directly after the final error correction. This shows that  
\begin{equation}
\lbr\Dec_l\circ \EC_l \rbr_{F^k_1\oplus F^k_2} = I^{F^k_1}_{\beta\lb F^k_1\oplus F^k_2\rb} \circ \Dec_l^*\circ  \lbr\EC_l \rbr_{F^k_2},
\label{equ:ImpLem3}
\end{equation}
for quantum channels
\[
I^{F^k_1}_0 := \ident_2\otimes \Tr_S \quad\text{and}\quad I^{F^k_1}_1:=\lbr \Dec_l\rbr_{F^k_1}\circ \Enc^*_l.
\]
Using \eqref{equ:ImpLem3} we can rewrite 
\[
\lbr\Dec^{\otimes m}_l \rbr_{F_1}\circ\lbr\EC^{\otimes m}_l\rbr_{F_2} = \lb\bigotimes^m_{k=1}\lb I^{F^k_1}_{\beta(F^k_1\oplus F^k_2)}\circ \Dec^*_l\rb\rb\circ \lbr\EC^{\otimes m}_l\rbr_{F_2},
\]
where we decomposed $F_1=F^1_1\oplus \cdots \oplus F^m_1$ and $F_2=F^1_2\oplus \cdots \oplus F^m_2$ as explained above. Now, we compute 
\begin{align}
&\lbr \lb\Dec_l\circ \EC_l\rb^{\otimes m}\circ \tilde{\Gamma}^1_{\mathcal{C}_l}\rbr_{\mathcal{F}_\pi(p)}=\sum_{F_1\oplus F_2\oplus F_3} P(F_1\oplus F_2\oplus F_3)  \lbr\Dec^{\otimes m}_l \rbr_{F_1}\circ\lbr\EC^{\otimes m}_l\rbr_{F_2}\circ \lbr\tilde{\Gamma}^1_{\mathcal{C}_l}\rbr_{F_3} \nonumber\\
&=\sum_{F_1\oplus F_2\oplus F_3} P(F_1\oplus F_2\oplus F_3) \lb\bigotimes^m_{k=1}\lb I^{F^k_1}_{\beta(F^k_1\oplus F^k_2)}\circ \Dec^*_l\rb\rb\circ \lbr\EC^{\otimes m}_l\rbr_{F_2}\circ\lbr\tilde{\Gamma}^1_{\mathcal{C}_l}\rbr_{F_3}\nonumber \\
&=\sum_{\substack{F_1\oplus F_2\oplus F_3\\ F_2\oplus F_3\in\mathcal{B}}} P(F_1\oplus F_2\oplus F_3)  \lb\bigotimes^m_{k=1} \lb I^{F^k_1}_{\beta(F^k_1\oplus F^k_2)}\circ \Dec^*_l\rb\rb\circ \lbr\EC^{\otimes m}_l\rbr_{F_2}\circ\lbr\tilde{\Gamma}^1_{\mathcal{C}_l}\rbr_{F_3} + \epsilon E \nonumber\\
&=\sum_{\substack{F_1\oplus F_2\oplus F_3\\ F_2\oplus F_3\in \mathcal{B}}} P(F_1\oplus F_2\oplus F_3)  \lb\bigotimes^m_{k=1} I^{F^k_1}_{\beta(F^k_1\oplus F^k_2)}\rb\circ \lb\Gamma^{1}\otimes \sigma_{F_2\oplus F_3}\rb + \epsilon E, 
\label{equ:ImplLem4}
\end{align}
where we used \eqref{equ:ImpLem1} and \eqref{equ:ImpLem2} in the last two equalities. Here, $E$ denotes a quantum channel, and $P$ denotes the probability distribution on fault patterns according to the i.i.d.~Pauli fault model $\mathcal{F}_\PIID(p)$. 

Next, we introduce 
\[
\overline{\beta}(F^k_1) = \begin{cases} 0 \quad\text{ if there exists }F^k_2 \text{ such that }\Dec_l\circ \EC_l\text{ is well-behaved under }F^k_1\oplus F^k_2\\
1 \quad\text{ otherwise}\end{cases}.
\]
Note that $\overline{\beta}(F^k_1)=1$ implies $\beta(F^k_1\oplus F^k_2)=1$ for any fault pattern $F^k_2$. Therefore, the fault patterns $F^k_1\oplus F^k_2$ such that $\overline{\beta}(F^k_1)\neq \beta(F^k_1\oplus F^k_2)$ have to satisfy $\overline{\beta}(F^k_1)=0$ and $\beta(F^k_1\oplus F^k_2)=1$. This is only possible if the well-behavedness condition from Definition \ref{defn:GoodFaultPatt} fails only in the last step of the circuit containing the error correction affected by $F^k_2$, i.e., this part of the circuit has to contain an $(l-1)$-extended rectangle that is not well-behaved under the fault pattern $F^k_1\oplus F^k_2$. By this reasoning and the union bound we have 
\begin{align*}
\delta :=P&\lb \exists k\in \lset 1,\ldots ,m\rset~:~\overline{\beta}(F^k_1) \neq \beta(F^k_1\oplus F^k_2)\rb \\
&\leq \sum^{m}_{k=1} P\lb \overline{\beta}(F^k_1) \neq \beta(F^k_1\oplus F^k_2)\rb \\
& \leq \sum^{m}_{k=1} P\lb \text{An }(l-1)\text{-extended rectangle in} \lbr\Dec^k_{l\ra (l-1)} \circ\EC^k_l\rbr_{F^k_1\oplus F^k_2} \text{ is not well-behaved}\rb \\
& \leq m C_2\lb \frac{p}{p_0}\rb^{2^{l-1}} ,
\end{align*}
where we used the threshold lemma (Lemma \ref{lem:Strongthreshold}) in the two last inequalities and introduced the number $\delta$. Note that $C_2>0$ depends only on the number of $(l-1)$-extended rectangles in the quantum circuit $\Dec_{l\ra (l-1)} \circ\EC_l$, but not on $l$. 

Using \eqref{equ:ImplLem4} and twice the triangle inequality we have  
\begin{align*}
&\|\lbr \lb\Dec_l\circ \EC_l\rb^{\otimes m}\circ \tilde{\Gamma}^1_{\mathcal{C}_l}\rbr_{\mathcal{F}_\pi(p)} - \sum_{\substack{F_1\oplus F_2\oplus F_3\\ F_2\oplus F_3\in \mathcal{B}}} P(F_1\oplus F_2\oplus F_3) \bigotimes^m_{k=1} \lb I^{F^k_1}_{\overline{\beta}(F^k_1)}\rb\circ \lb\Gamma^{1}\otimes \sigma_{F_2\oplus F_3}\rb\|_{1\ra 1}\\
&\leq \| \sum_{\substack{F_1\oplus F_2\oplus F_3\\ F_2\oplus F_3\in \mathcal{B}}} P(F_1\oplus F_2\oplus F_3) \lbr \bigotimes^m_{k=1} \lb I^{F^k_1}_{\beta(F^k_1\oplus F^k_2)}\rb - \bigotimes^m_{k=1} \lb I^{F^k_1}_{\overline{\beta}(F^k_1)}\rb\rbr\circ \lb\Gamma^{1}\otimes \sigma_{F_2\oplus F_3}\rb \|_{1\ra 1} + \epsilon \\
&\leq \sum_{\substack{F_1\oplus F_2\oplus F_3\\ F_2\oplus F_3\in \mathcal{B}}} P(F_1\oplus F_2\oplus F_3) \Big{\|} \lbr \bigotimes^m_{k=1} \lb I^{F^k_1}_{\beta(F^k_1\oplus F^k_2)}\rb - \bigotimes^m_{k=1} \lb I^{F^k_1}_{\overline{\beta}(F^k_1)}\rb\rbr \Big{\|}_{1\ra 1} + \epsilon . 
\end{align*}
The $1\ra 1$-norm in the final expression of the previous computation is either $0$ when $\beta(F^k_1\oplus F^k_2)$ and $\overline{\beta}(F^k_1)$ coincide for every $k\in\lset 1,\ldots, m\rset$, or it can be upper bounded by $2$ since its argument is the difference of two quantum channels. Therefore, we find
\begin{align*}
&\|\lbr \lb\Dec_l\circ \EC_l\rb^{\otimes m}\circ \tilde{\Gamma}^1_{\mathcal{C}_l}\rbr_{\mathcal{F}_\pi(p)} - \sum_{\substack{F_1\oplus F_2\oplus F_3\\ F_2\oplus F_3\in \mathcal{B}}} P(F_1\oplus F_2\oplus F_3) \bigotimes^m_{k=1} \lb I^{F^k_1}_{\overline{\beta}(F^k_1)}\rb\circ \lb\Gamma^{1}\otimes \sigma_{F_2\oplus F_3}\rb\|_{1\ra 1}\\
& \leq 2P\lb \exists k\in \lset 1,\ldots ,m\rset~:~\overline{\beta}(F^k_1) \neq \beta(F^k_1\oplus F^k_2) \text{ and }F_2\oplus F_3\in \mathcal{B}\rb + \epsilon \\
&\leq 2P\lb \exists k\in \lset 1,\ldots ,m\rset~:~\overline{\beta}(F^k_1) \neq \beta(F^k_1\oplus F^k_2)\rb + \epsilon \\
&= 2\delta + \epsilon.
\end{align*}
where we used elementary probability theory and $\delta$ from above. Since faults $F_1,F_2$ and $F_3$ affecting the three parts of the circuit act independently, we can write 
\begin{align*}
\sum_{\substack{F_1\oplus F_2\oplus F_3\\ F_2\oplus F_3\in \mathcal{B}}}& P(F_1\oplus F_2\oplus F_3) \bigotimes^m_{k=1} \lb I^{F^k_1}_{\overline{\beta}(F^k_1)}\rb\circ \lb\Gamma^{1}\otimes \sigma_{F_2\oplus F_3}\rb \\
&=\sum_{\substack{F_1\oplus F_2\oplus F_3\\ F_2\oplus F_3\in \mathcal{B}}} P_1(F_1)P_2(F_2)P_3(F_3) \bigotimes^m_{k=1} \lb I^{F^k_1}_{\overline{\beta}(F^k_1)}\rb\circ \lb\Gamma^{1}\otimes \sigma_{F_2\oplus F_3}\rb \\
&= (1-\epsilon)\bigotimes^m_{k=1} \lb \sum_{F^k_1} P_1(F^k_1) I^{F^k_1}_{\overline{\beta}(F^k_1)}\rb\circ \lb\Gamma^{1}\otimes \sigma_S\rb \\
& =(1-\epsilon) \lb \tilde{N}^{\text{dec},l}_{q_1}\rb^{\otimes m}\circ \lb\Gamma^{1}\otimes \sigma_S\rb, 
\end{align*}
where $P_1,P_2$ and $P_3$ denote the probability distributions on fault patterns according to the i.i.d.~Pauli fault model $\mathcal{F}_\PIID(p)$ restricted to the three partial circuits. Furthermore, $\sigma_S$ denotes a quantum state on the syndrome space given by 
\[
\sigma_S = \frac{1}{1-\epsilon}\sum_{F_2\oplus F_3\in \mathcal{B}}P_2(F_2)P_3(F_3)\sigma_{F_2\oplus F_3},
\]
and we introduced a quantum channel $\tilde{N}^{\text{dec},l}_{q_1}:\M_2\otimes \M^{\otimes (7^l-1)}_2\ra \M_2$ as
\[
\tilde{N}^{\text{dec},l}_{q_1} = (1-q_1)\ident_2\otimes \Tr_S + q_1 \tilde{N}_l
\]
for 
\[
q_1:= P\lb \overline{\beta}(F^k_1) = 1\rb
\]
and the quantum channel $\tilde{N}_l:\M_2\otimes \M^{\otimes (7^l-1)}_2\ra \M_2$ defined as 
\[
\tilde{N}_l = \frac{1}{q_1}\sum_{F_1\text{ s.th. }\overline{\beta}(F_1)=1} P_1(F_1)\lbr \Dec_l\rbr_{F_1}\circ \Enc^*_l .
\]
Finally, by the aforementioned properties of $\beta$ and $\overline{\beta}$ and the same reasoning as in the proof of Theorem \ref{thm:NiceInterfaceConcatenated} (and with the same constant $c>0$) we have
\[
q_1=P\lb \overline{\beta}(F^k_1) = 1\rb \leq P\lb \beta(F^k_1\oplus F^k_2) = 1\rb \leq 2cp,
\] 
where $p$ is the local noise parameter of the i.i.d.~Pauli noise model $\mathcal{F}_\PIID(p)$. Finally, we have
\[
\tilde{N}^{\text{dec},l}_{q_1} = (1-2cp)\ident_2\otimes \Tr_S + 2cp N_l =: N^{\text{dec},l}_{p} 
\]
for some quantum channel $N_l:\M_2\otimes \M^{\otimes (7^l-1)}_2\ra \M_2$ finishing the proof.

\item The quantum circuit $\Enc_l$ from \eqref{equ:EncDec1} ends in an error correction $\EC_l$ of the code $\mathcal{C}_l$ (cf.~Definition \ref{defn:Interface}) and we can write 
\[
\Gamma^2_{\mathcal{C}_l}\circ \Enc^{\otimes m}_l =\Gamma^2_{\mathcal{C}_l}\circ \EC^{\otimes m}_l\circ \overline{\Enc}^{\otimes m}_l ,
\] 
where $\overline{\Enc}_l$ denotes the quantum circuit obtained from $\Enc_l$ by removing the final error correction. In the following, we will write $F=F_1\oplus F_2\oplus F_3$ to denote a fault pattern $F$ affecting the quantum circuit $\Gamma^2_{\mathcal{C}_l}\circ \Enc^{\otimes m}_l$ that is composed of fault patterns $F_1$, $F_2$ and $F_3$ affecting the quantum circuits $\Gamma^1_{\mathcal{C}_l}$, $\EC^{\otimes m}_l$ and $\overline{\Enc}^{\otimes m}_l$ respectively. Let $\mathcal{A}$ denote the set of fault patterns $F=F_1\oplus F_2$ such that every $l$-extended rectangle in the quantum circuit
\[
\lbr\Gamma^2_{\mathcal{C}_l}\circ \EC^{\otimes m}_l\rbr_{F_1\oplus F_2}
\]
is well-behaved. By Lemma \ref{lem:CircuitTransformExRecCorr} we have   
\[
\lbr\Gamma^2_{\mathcal{C}_l}\circ \EC^{\otimes m}_l\rbr_{F_1\oplus F_2} = \lb\Gamma^2\otimes \text{Tr}_S\rb\circ (\Dec^*_l)^{\otimes m}\circ \lbr\EC^{\otimes m}_l\rbr_{F_2} .
\]
for any $F_1\oplus F_2\in \mathcal{A}$, and the threshold lemma (Lemma \ref{lem:Strongthreshold}) shows that 
\[
\epsilon := P(F_1\oplus F_2\notin \mathcal{A}) \leq C\lb \frac{p}{p_0}\rb^{2^{l}}|\Loc(\Gamma^2)|,
\] 
where $P$ denotes the probability distribution of fault patterns according to the i.i.d.~Pauli fault model $\mathcal{F}_\PIID(p)$. Now, we can compute
\begin{align*}
&\lbr \Gamma^2_{\mathcal{C}_l}\circ \Enc^{\otimes m}_l\rbr_{\mathcal{F}_\PIID(p)} \\
&=\sum_{F\text{ fault pattern }}P(F)\lbr \Gamma^2_{\mathcal{C}_l}\circ \EC^{\otimes m}_l\circ \overline{\Enc}^{\otimes m}_l\rbr_{F} \\
&=\sum_{\substack{F=F_1\oplus F_2\oplus F_3\\ F_1\oplus F_2\in \mathcal{A}}}P_1(F_1)P_2(F_2)P_3(F_3)\lbr \Gamma^2_{\mathcal{C}_l}\circ \text{EC}^{\otimes m}_l\circ \overline{\Enc}^{\otimes m}_l\rbr_{F} + \epsilon E \\
&= \lb\Gamma^2\otimes \Tr_S\rb\circ (\Dec^*_l)^{\otimes m}\circ\lb\sum_{F_2\oplus F_3}P_2(F_2)P_3(F_3)\lb\sum_{\substack{F_1\\ F_1\oplus F_2\in\mathcal{A}}}P_1(F_1)\rb \lbr \Enc^{\otimes m}_l\rbr_{F_2\oplus F_3}\rb + \epsilon E,
\end{align*}
for some quantum channel $E$, and where $P_1,P_2$ and $P_3$ denote the probability distributions on fault patterns according to the i.i.d.~Pauli fault model $\mathcal{F}_\PIID(p)$ restricted to the three partial circuits. By normalization we have 
\[
\sum_{F_2\oplus F_3}P_2(F_2)P_3(F_3)\lb\sum_{\substack{F_1\\ F_1\oplus F_2\in\mathcal{A}}}P_1(F_1)\rb = 1-\epsilon.
\] 
By definition of the i.i.d.~Pauli fault model $\mathcal{F}_\PIID(p)$ we have
\[
\lbr \Enc^{\otimes m}_l\rbr_{\mathcal{F}_\PIID(p)} = \sum_{F_2\oplus F_3}P_2(F_2)P_3(F_3)\lbr \Enc^{\otimes m}_l\rbr_{F_2\oplus F_3}.
\]
Using the triangle inequality and that $\lb\Gamma^2\otimes \Tr_S\rb\circ (\Dec^*_l)^{\otimes m}$ is a quantum channel, we find that 
\begin{align*}
\|&\lbr \Gamma^2_{\mathcal{C}_l}\circ \Enc^{\otimes m}_l\rbr_{\mathcal{F}_\PIID(p)} - \lb\Gamma^2\otimes \Tr_S\rb\circ (\Dec^*_l)^{\otimes m}\circ \lbr \Enc^{\otimes m}_l\rbr_{\mathcal{F}_\PIID(p)}\|_{1\ra 1} \\
&\leq \epsilon + \sum_{F_2\oplus F_3}P_2(F_2)P_3(F_3)\lb 1-\sum_{\substack{F_1\\ F_1\oplus F_2\in\mathcal{A}}}P_1(F_1)\rb =2\epsilon. 
\end{align*}
Finally, we note that by Definition \ref{defn:CorrectInterface} we have 
\[
\sum_{F\text{ s.th. }\lbr \Enc_l\rbr_F\text{ correct}}P(F) \Dec^*_l\circ \lbr \Enc_l\rbr_{F} = (1-q_1)\ident_2\otimes \sigma,
\]
where 
\[
q_2 :=P(\lbr \Enc_l\rbr_F\text{ not correct})\leq 2cp,
\]
using Theorem \ref{thm:NiceInterfaceConcatenated} (and with the same constant $c>0$). Moreover, the quantum state $\sigma$ on the syndrome space is given by 
\[
\sigma = \frac{1}{1-q_2}\sum_{F\text{ s.th. }\lbr \Enc_l\rbr_F\text{ correct}}P(F) \sigma_F ,
\]
for quantum states $\sigma_F$ depending on the specific fault patterns. Dividing the fault patterns $F$ into two sets, one where $\lbr \Enc_l\rbr_F$ is correct and one where it is not, leads to
\[
\lb\ident_2\otimes \Tr_S\rb\circ\Dec^*_l\circ \lbr \Enc_l\rbr_{\mathcal{F}_\PIID(p)} = (1-q_2)\ident_2 + q_2 \tilde{N}'_l,
\]
where $\tilde{N}'_l:\M_2\ra\M_2$ is given by 
\[
\tilde{N}'_l= \frac{1}{q_2}\sum_{F\text{ s.th. }\lbr \Enc_l\rbr_F\text{ not correct}}P(F) \lb\ident_2\otimes \Tr_S\rb\circ \Dec^*_l\circ \lbr \Enc_l\rbr_{F}.
\]
Finally, we can rewrite
\[
(1-q_2)\ident_2 + q_2 \tilde{N}'_l = (1-2cp)\ident_2 + 2cp N'_l =: N^{\text{enc},l}_{p}  
\]
for some quantum channel $N'_l:\M_2\ra\M_2$ finishing the proof.
\end{enumerate}

\end{proof}

The previous lemma shows that using a quantum error correcting code to protect a general coding scheme for information transmission via a physical channel leads to a modified effective channel between the data subspaces. To make this precise, we state the following theorem which is a direct consequence of Lemma \ref{lem:Main}. 

\begin{thm}[Effective communication channel]
Let $T:\M^{\otimes j_1}_2\ra\M^{\otimes j_2}_2$ denote a quantum channel. Furthermore, let $\Gamma^1:\C^{2^N}\ra\M^{\otimes m}_2$ be a quantum circuit with classical input, and $\Gamma^2:\M^{\otimes m}_2 \ra\C^{2^M}$ a quantum circuit with classical output. For each $l\in\N$ let $\mathcal{C}_l\subset (\C^2)^{\otimes 7^l}$ denote the $l$th level of the concatenated $7$-qubit Steane code with threshold $p_0\in\lb 0,1\rbr$ (see Lemma \ref{lem:Strongthreshold}). Moreover, we denote by $\Enc_l:\M_2\ra\M^{\otimes 7^l}_2$ and $\Dec_l:\M^{\otimes 7^l}_2\ra \M_2$ the interface circuits from \eqref{equ:EncDec1} and \eqref{equ:EncDec2} and by $c>0$ the constant from Theorem \ref{thm:NiceInterfaceConcatenated}. For any $0\leq p\leq \min(p_0/2,(2(j_1+j_2)c)^{-1})$ and any $l\in\N$ we have
\begin{align*}
\Big\| &\lbr\Gamma^2_{\mathcal{C}_l}\circ \lb\Enc^{\otimes j_2}_l\circ T\circ\Dec^{\otimes j_1}_{l}\rb^{\otimes m}\circ \Gamma^1_{\mathcal{C}_l}\rbr_{\mathcal{F}_{\PIID(p)}} - \Gamma^2\circ T_{p,l}^{\otimes m}\circ (\Gamma^1\otimes \sigma_S)\Big\|_{1\ra 1}\\
&\quad\quad\leq C\lb \frac{p}{p_0}\rb^{2^l}\lb |\Loc(\Gamma^1)|+ |\Loc(\Gamma^2)| +j_1m\rb,
\end{align*}
where $\sigma_S$ denotes some quantum state on the syndrome space of the last $j_1m$ lines depending on $l\in\N$ the quantum circuit $\Gamma^2$ and the interface circuit $\Dec_l$. The effective quantum channel $T_{p,l}:\M^{\otimes j_1}_2\otimes \M^{\otimes j_1(7^l-1)}_2\ra \M^{\otimes j_2}_2$ is of the form 
\[
T_{p,l} = (1-2(j_1+j_2)cp)T\otimes \Tr_S + 2(j_1+j_2)cp N_l,
\]
where the partial trace is acting on the syndrome space, and the quantum channel $N_l:\M^{\otimes j_1}_2\otimes \M^{\otimes j_1(7^l-1)}_2\ra\M^{\otimes j_2}_2$ depends on $T$, the level $l$ and the interface circuits $\Enc_l$ and $\Dec_l$. In the above, $C>0$ denotes a constant not depending on $m,j_1,j_2,l$ or any of the occurring quantum circuits or quantum channels.
\label{thm:EffCommChan}
\end{thm}
It should be noted that the right hand side of the inequality in Theorem \ref{thm:EffCommChan} has to be small for the encoded quantum circuits $\Gamma^1_{\mathcal{C}_l}$ and $\Gamma^2_{\mathcal{C}_l}$ to be well-behaved with high probability. If these circuits implement a coding scheme for transmitting information (classical or quantum) over the quantum channel $T$, then to function correctly under noise, they also have to be a coding scheme for the effective quantum channel $T_{p,l}$ taking as input the data qubits and a syndrome state that might be entangled over multiple copies of the quantum channel. Note also that in general the number of locations $|\Loc(\Gamma^1)|$ and $|\Loc(\Gamma^2)|$ of the quantum circuits may scale exponentially in the number of channel uses $m$. Even in this case, the right hand side of the inequality in Theorem \ref{thm:EffCommChan} can be made arbitrarily small by choosing the level $l$ large enough. 

The entanglement in the syndrome state $\sigma_S$ in Theorem \ref{thm:EffCommChan} and how it might affect the effective quantum channel has to be studied more carefully in the future. It is certainly possible for $\sigma_S$ to be highly entangled between multiple communication lines and still correspond to a correctable syndrome. However, we have not actually shown this to be the case in practice. Difficulties arise from the fact that the structure of $\sigma_S$ depends on the quantum circuits in question, and that high levels of the concatenated $7$-qubit Steane code and the corresponding interface circuits are quite complicated. In the following, we have therefore adopted the approach of finding coding schemes for the worst-case scenario, where $\sigma_S$ is highly entangled. These coding schemes will likely be applicable for noise models beyond i.i.d.~Pauli noise $\mathcal{F}_\PIID(p)$ including correlations between multiple communication lines. 

Note that Theorem \ref{thm:EffCommChan} is formulated for quantum channels $T:\M^{\otimes j_1}_2\ra\M^{\otimes j_2}_2$ between quantum systems composed of qubits. This is only for notational convenience since we consider interfaces between logical qubits and physical qubits. A general quantum channel can always be embedded into a quantum channel between systems composed of qubits, and then Theorem \ref{thm:EffCommChan} applies. We will use this fact in the next section to obtain more general results.

\section{Capacities under arbitrarily varying perturbations}\label{sec:CapAVP}

Before studying fault-tolerant capacities we want to discuss a different kind of capacity which we refer to capacities under arbitrarily varying perturbations. This will turn out to capture the effective communication problems, i.e., for channel models of the same form as the effective communication channel in Theorem \ref{thm:EffCommChan}, which emerge from applying interface circuits in the scenarios of classical and quantum communication over quantum channels. 

\subsection{Arbitrarily varying perturbations}\label{sec:CapAVPFirst}

Given a quantum channel $T:\M_{d_A}\ra \M_{d_B}$, we define a quantum channel $T_{p,N}:\M_{d_A}\otimes \M_{d_E}\ra \M_{d_B}$ by 
\[
T_{p,N} = (1-p) T\otimes \text{Tr}_E + p N, 
\]
for any $d_E\in \N$ and any quantum channel $N:\M_{d_A}\otimes \M_{d_E}\ra \M_{d_B}$. We will now consider communication problems for the quantum channels
\begin{equation}\label{equ:AVPChannelModel}
\M^{\otimes m}_{d_A}\ni X\mapsto T^{\otimes m}_{p,N}\lb X\otimes \sigma_E\rb \in \M^{\otimes m}_{d_B}, 
\end{equation}
where $\sigma_E\in \D\lb(\C^{d_E})^{\otimes m}\rb$ is some environment state. We call this channel model \emph{channels under arbitrarily varying perturbation} (AVP) motivated by similar channel models~\cite{boche2018fully}. In the following, we will define capacities under AVP and we will start by defining the corresponding coding schemes. It should be emphasized that in our definitions, the environment states $\sigma_E$ will have unconstrained dimension (but they might belong to a certain subset of quantum states). Therefore, previous results such as~\cite{boche2018fully}, where this dimension is finite, are not directly applicable to these definitions.

\begin{defn}[Coding schemes for classical communication under AVP]\label{defn:CodingSchemesEffClassCap}
Let $T:\M_{d_A}\ra\M_{d_B}$ denote a quantum channel and fix $p\in\lbr 0,1\rbr$. For each $m,d_E\in\N$ consider a set $\mathcal{S}_m(d_E)\subseteq \D\lb(\C^{d_E})^{\otimes m}\rb$ of environment states. For $n,m\in\N$ and $\epsilon\in\R^+$, an $(n,m,\epsilon)$ coding scheme for classical communication over $T$ under AVP of strength $p$ consists of a classical-quantum channel $E:\C^{2^n}\ra\M^{\otimes m}_{d_A}$ and a quantum-classical channel $D:\M^{\otimes m}_{2}\ra \C^{2^n}$ such that 
\[
\sup_{d_E,\sigma_E, N_m}\epsilon_{cl}\lb D\circ T_{p,N_{m}}^{\otimes m}\circ\lb E \otimes \sigma_E\rb\rb \leq \epsilon ,
\]
where the supremum goes over the environment dimension $d_E\in\N$, states $\sigma_E\in \mathcal{S}_m(d_E)$ and quantum channels $N_m:\M_{d_A}\otimes \M_{d_{E}}\ra \M_{d_B}$. By $\epsilon_{cl}$ we denote the classical communication error from \eqref{equ:epsilonCL}.
\end{defn}

We have a similar definition in the case of quantum communication:

\begin{defn}[Coding scheme for quantum communication under AVP]\label{defn:CodingSchemesEffQuantCap}
Let $T:\M_{d_A}\ra\M_{d_B}$ denote a quantum channel and fix $p\in\lbr 0,1\rbr$. For each $m,d_E\in\N$ consider a set $\mathcal{S}_m(d_E)\subseteq \D\lb(\C^{d_E})^{\otimes m}\rb$ of quantum states. For $n,m\in\N$ and $\epsilon\in\R^+$, an $(n,m,\epsilon)$ coding scheme for quantum communication over $T$ under AVP of strength $p$ consists of a quantum channel $E:\C^{2^n}\ra\M^{\otimes m}_{d_A}$ and a quantum channel $D:\M^{\otimes m}_{2}\ra \C^{2^n}$ such that 
\[
\sup_{d_E, \sigma_E, N_m}\| \ident^{\otimes n}_2 - D\circ T_{p,N_{m}}^{\otimes m}\circ\lb E \otimes \sigma_E\rb\|_{\diamond}\leq \epsilon ,
\]
where the supremum goes over the dimension $d_E\in\N$, states $\sigma_E\in \mathcal{S}_m(d_E)$ and quantum channels $N_m:\M_{d_A}\otimes \M_{d_{E}}\ra \M_{d_B}$.

\end{defn}

Having defined coding schemes for quantum channels under AVP, we can define the corresponding capacities as usual:

\begin{defn}[Classical and quantum capacity under AVP]\label{defn:EffCCapAndQCap}
Let $T:\M_{d_A}\ra\M_{d_B}$ denote a quantum channel and fix $p\in\lbr 0,1\rbr$. For each $m,d_E\in\N$ consider a set $\mathcal{S}_m(d_E)\subseteq \D\lb(\C^{d_E})^{\otimes m}\rb$ of quantum states. We call $R\geq 0$ an achievable rate for classical (or quantum) communication under AVP of strength $p$ if for every $m\in\N$ there exists an $n_m\in\N$ and an $(n_m,m,\epsilon_m)$ coding scheme under AVP of strength $p$ for classical (or quantum) communication over $T$ with $\epsilon_m\ra 0$ as $m\ra\infty$ and 
\[
\liminf_{m\ra\infty} \frac{n_m}{m}\geq R.
\] 
The \emph{classical capacity under arbitrarily varying perturbations} of $T$ is given by
\[
C^{\text{AVP}}_{\mathcal{S}_m(d_E)}(p,T)=\sup\lset R\geq 0 \text{ achievable rate for classical communication under AVP of strength $p$}\rset,
\]
and the \emph{quantum capacity under arbitrarily varying perturbations} of $T$ is given by 
\[
Q^{\text{AVP}}_{\mathcal{S}_m(d_E)}(p,T)=\sup\lset R\geq 0 \text{ achievable rate for quantum communication under AVP of strength $p$}\rset.
\]
We will use a special notation for two special cases of these capacities:
\begin{itemize}
\item We write $C^{\text{AVP}}_{\text{SEP}}(p,T)$ and $Q^{\text{AVP}}_{\text{SEP}}(p,T)$, when the set $\mathcal{S}_m(d_E)$ of environment states consists of all fully separable states in $\D\lb(\C^{d_E})^{\otimes m}\rb$ for each $m\in\N$. 
\item We write $C^{\text{AVP}}_{\text{ALL}}(p,T)$ and $Q^{\text{AVP}}_{\text{ALL}}(p,T)$, when the set $\mathcal{S}_m(d_E)$ of environment states is the entire set $\D\lb(\C^{d_E})^{\otimes m}\rb$ for each $m\in\N$.
\end{itemize}
\end{defn}

\subsection{Separable environment states}\label{sec:CapAVPSep}

In the case of separable environment states, the capacities under AVP coincide with certain capacities for arbitrarily varying quantum channels (we refer to~\cite{ahlswede2013quantum} for the precise definition of this channel model). Specifically, we have the following lemma:

\begin{lem}\label{lem:AVPvsAVC}
For any quantum channel $T:\M_{d_A}\ra \M_{d_B}$ we have
\[
C^{\text{AVP}}_{\text{SEP}}(p,T) = C\lb\lset (1-p)T + pS\rset_{S}\rb \quad\text{ and }\quad Q^{\text{AVP}}_{\text{SEP}}(p,T) = Q\lb\lset (1-p)T + pS\rset_{S}\rb,
\]
where $\lset (1-p)T + pS\rset_{S}$ denotes the arbitrarily varying quantum channel consisting of quantum channels $(1-p)T + pS$ for arbitrary quantum channels $S:\M_{d_A}\ra \M_{d_B}$.
\end{lem}
\begin{proof}
To see that the respective left-hand sides are larger than the right-hand sides in the identies of the lemma, consider a fixed dimension $d_E\in\N$ a quantum channel $N:\M_{d_A}\otimes \M_{d_E}\ra \M_{d_B}$ and the quantum channel 
\begin{equation}\label{equ:channelUnderPertInProof}
\M^{\otimes m}_{d_A}\ni X\mapsto T^{\otimes m}_{p,N}\lb X\otimes \sigma_E\rb \in \M^{\otimes m}_{d_B}, 
\end{equation}
for any fixed fully-separable quantum state $\sigma_E\in \D\lb(\C^{d_E})^{\otimes m}\rb$. By assumption, we have
\[
\sigma_E = \sum^k_{i=1} q_i \tau^{(1)}_i\otimes \cdots \otimes \tau^{(m)}_i ,
\] 
for a probability distribution $\lset q_i\rset^k_{i=1}$ and quantum states $\tau^{(1)}_i,\ldots ,\tau^{(m)}_i \in \D\lb \C^{d_E}\rb$. Inserting this state into \eqref{equ:channelUnderPertInProof} shows that
\[
T^{\otimes m}_{p,N}\lb X\otimes \sigma_S\rb = \sum^k_{i=1} q_i \bigotimes^m_{j=1} \lb (1-p)T + pN_{i,j}\rb,
\]
where $N_{i,j}:\M_{d_A}\ra \M_{d_B}$ is the quantum channel given by $N_{i,j}(X)=N\lb X\otimes \tau^{(j)}_i\rb$. By convexity of the error measures used in Defintion \ref{defn:CodingSchemesEffClassCap} and Definition \ref{defn:CodingSchemesEffQuantCap}, we conclude that any communication rate that is achievable for the arbitrarily varying channel $\lset (1-p)T + pS\rset_{S}$ parametrized by quantum channels $S:\M_{d_A}\ra \M_{d_B}$ is also achievable under AVP in the case of fully separable environment states.  

To see the other directions, consider any $m$-tuple of quantum channels $S_1,\ldots ,S_m:\M_{d_A}\ra \M_{d_B}$. It is easy to see that  
\[
\bigotimes^m_{j=1} \lb (1-p)T + pS_{j}\rb\lb X\rb = T^{\otimes m}_{p,N}\lb X\otimes \sigma_E\rb ,
\]
by choosing $N:\M_{d_A}\otimes \M_{d_E}\ra \M_{d_B}$ for $d_E=m$ as 
\[
N(Y) = \sum^m_{i=1} N_i\lb (\one_{d_A}\otimes \bra{i})Y(\one_{d_A}\otimes \ket{i})\rb,
\]
for any $Y\in \M_{d_A}\otimes \M_{d_E}$, and $\sigma_E = \bigotimes^m_{j=1} \proj{j}{j}$. This shows that the arbitrarily varying quantum channel $\lset (1-p)T + pS\rset_{S}$ can be simulated by quantum channels of the form \eqref{equ:channelUnderPertInProof}. We conclude that communication rates under AVP, in the case where the environment states are fully separable, are also achievable for the arbitrarily varying quantum channel $\lset (1-p)T + pS\rset_{S}$.
\end{proof}

Now, we have the following lower bound:

\begin{thm}[Lower bounds on capacities under AVP with separable environment states]\label{thm:EffCapSepEnvLowerB}
For any quantum channel $T:\M_{d_A}\ra \M_{d_B}$ with classical capacity $C(T)>0$ there exists a $p_0\in \lbr 0,1\rbr$ such that  
\[
Q^{\text{AVP}}_{\text{SEP}}(p,T) \geq Q(T) - 2p\log(d_B)-(1+p)h_2\lb \frac{p}{1+p}\rb, 
\]
for any $p\leq p_0$.
\end{thm}
\begin{proof}
By Lemma \ref{lem:AVPvsAVC} and the quantum version of the Ahlswede dichotomy~\cite[Theorem 5]{ahlswede2013quantum} we have 
\[
Q^{\text{AVP}}_{\text{SEP}}(p,T) = Q\lb\lset (1-p)T + pS\rset_{S}\rb \geq \lim_{k\ra\infty} \frac{1}{k}\max_{\phi_{A'}} \inf_{S} I_{\text{coh}}\lb \phi_{A'},\lb(1-p)T+pS\rb^{\otimes k}\rb,
\]
with optimizations over quantum states $\phi_{A'}\in \mathcal{D}\lb (\C^{d_A})^{\otimes k}\rb$ and quantum channels $S:\M_{d_A}\ra\M_{d_B}$, whenever $C^{\text{AVP}}_{\text{SEP}}(p,T)>0$. Using Theorem \ref{thm:LowerBoundEffCap} and the fact that $C(T)>0$, there exists $p_0\in \lbr 0,1\rbr$ such that $C^{\text{AVP}}_{\text{SEP}}(p,T)\geq C^{\text{AVP}}_{\text{ALL}}(p,T)>0$. Finally, by the continuity bound~\cite[Proposition 3A]{shirokov2017tight} (see also~\cite{leung2009continuity}) we find that
\[
Q^{\text{AVP}}_{\text{SEP}}(p,T)\geq \lim_{k\ra\infty} \frac{1}{k}\max_{\phi_{A'}} \inf_{S} I_{\text{coh}}\lb \phi_{A'},\lb(1-p)T+pS\rb^{\otimes k}\rb \geq Q(T) - 2p\log(d_B)-(1+p)h_2\lb \frac{p}{1+p}\rb,
\]
for any $p\leq p_0$. This finishes the proof.
\end{proof}

We are certain that a similar bound as in the previous theorem should also hold for the classical capacity under AVP with separable environment states. However, the analogue of the Ahlswede dichotomy for the classical capacity does not seem to appear in the literature.

\subsection{Unrestricted environment states}\label{sec:CapAVPAll}

To deal with arbitrary environment states, we will use the following postselection technique:

\begin{lem}[Simple postselection technique]
Consider $m,d_E\in \N$ and a quantum state $\sigma_E\in \D\lb(\C^{d_{E}})^{\otimes m}\rb$. For $p\in\lbr 0,1\rbr$, and quantum channels $T:\M_{d_A}\ra \M_{d_B}$ and $N:\M_{d_A}\otimes \M_{d_E}\ra \M_{d_B}$ we define the quantum channel $T_{p,N}:\M_{d_A}\otimes \M_{d_E}\ra \M_{d_B}$ as
\[
T_{p,N} = (1-p)T\otimes \Tr_S + p N
\]
We denote by $T_{p,N}^{\otimes m}\lb \cdot\otimes \sigma_E\rb:\M^{\otimes m}_{d_A}\ra \M^{\otimes m}_{d_B}$ the quantum channel given by 
\[
\M^{\otimes m}_{d_A}\ni X\mapsto T_{p,N}^{\otimes m}\lb X\otimes \sigma_E\rb\in \M^{\otimes m}_{d_B}, 
\]
where each $T_{p,N}$ acts partially on the state $\sigma_E$. Then, for any $\delta>0$ we have
\[
T_{p,N}^{\otimes m}\lb \cdot\otimes \sigma_E\rb\leq (d_Ad_B)^{m(p+\delta)}\tilde{T}^{\otimes m}_p + \exp(-m\frac{\delta^2}{3p})S
\] 
for a quantum channel $S:\M^{\otimes m}_{d_A}\ra\M^{\otimes m}_{d_B}$ and where 
\[
\tilde{T}_p = (1-p)T + p \frac{\one_{d_B}}{d_B}\Tr.
\]
Here, we write $S_1\leq S_2$ for linear maps $S_1$ and $S_2$ when $S_2-S_1$ is completely positive.
\label{lem:PostSelect}
\end{lem}

\begin{proof}
Consider a quantum state $\sigma_E\in \D\lb(\C^{d_{E}})^{\otimes m}\rb$, $m\in\N$, and some $\delta>0$. Setting $S_0 = T\otimes \Tr_E$ and $S_1 = N$ we find
\begin{align*}
&T_{p,N}^{\otimes m}\lb \cdot \otimes \sigma_E\rb \\
&= \sum^m_{k=0} (1-p)^{m-k}p^k\sum_{i_1+\cdots +i_{m}=k}\bigotimes^m_{s=1} S_{i_s}\lb \cdot\otimes \sigma_E\rb\\
&=\sum^{\lfloor m(p+\delta)\rfloor}_{k=0} (1-p)^{m-k}p^k\sum_{i_1+i_2+\cdots i_{m}=k} \bigotimes^m_{s=1} S_{i_s}\lb \cdot\otimes \sigma_E\rb + \text{P}\lb \frac{1}{m}\sum^{m}_{i=1} X_i >p+\delta\rb S,
\end{align*} 
where $S:\M^{\otimes m}_{d_A}\ra \M^{\otimes m}_{d_B}$ denotes a quantum channel collecting the second part of the sum. Furthermore, we introduced independent and identically distributed $\lset 0,1\rset$-valued random variables $X_i$ with $P(X_1=1)=q$. Note that each of the product channels in the remaining sum can be upper bounded as
\[
\bigotimes^m_{s=1} S_{i_s}\lb \cdot\otimes \sigma_E\rb \leq (d_Ad_B)^{\lfloor m(p+\delta)\rfloor} \bigotimes^m_{s=1} \tilde{S}_{i_s}, 
\]
with $\tilde{S}_0 = T$ and $\tilde{S}_1 = \frac{\one_{d_B}}{d_B}\Tr$, where we used that $i_1+\cdots +i_m\leq \lfloor m(p+\delta)\rfloor$ and that every quantum channel $N':\M^{\otimes k}_{d_A}\ra\M^{\otimes k}_{d_B}$ satisfies
\[
N'\leq (d_Ad_B)^{k}\lb\frac{\one_{d_B}}{d_B}\Tr\rb^{\otimes k}.
\]
By the Chernoff bound we have 
\[
\text{P}\lb \frac{1}{m}\sum^{m}_{i=1} X_i >p+\delta\rb \leq \text{exp}\lb -m\frac{\delta^2}{3p}\rb,
\]
and combining the previous equations, we find that 
\begin{align*}
T_p^{\otimes m}\lb \cdot \otimes \sigma_E\rb &\leq (d_Ad_B)^{\lfloor m(p+\delta)\rfloor}\sum^{\lfloor m(p+\delta)\rfloor}_{k=0} (1-p)^{m-k}p^k\sum_{i_1+i_2+\cdots i_{m}=k} \bigotimes^m_{s=1} \tilde{S}_{i_s}+ \text{exp}\lb -m\frac{\delta^2}{3p}\rb S \\
&\leq (d_Ad_B)^{\lfloor m(p+\delta)\rfloor}\sum^{m}_{k=0} (1-p)^{m-k}p^k\sum_{i_1+i_2+\cdots i_{m}=k} \bigotimes^m_{s=1} \tilde{S}_{i_s}+ \text{exp}\lb -m\frac{\delta^2}{3p}\rb S \\
&= (d_Ad_B)^{\lfloor m(p+\delta)\rfloor}\tilde{T}^{\otimes m}_p+ \text{exp}\lb -m\frac{\delta^2}{3p}\rb S,
\end{align*} 
where we added a completely positive term in the second inequality. Since $\lfloor m(p+\delta)\rfloor\leq m(p+\delta)$ the proof is finished. 

\end{proof}

We will show the following theorem:

\begin{thm}[Lower bounds on capacities under AVP]\label{thm:LowerBoundEffCap}
Let $T:\M_{d_A}\ra\M_{d_B}$ denote a quantum channel, and fix $p\in\lbr 0,1\rb$. In the classical case, we have 
\[
C^{\text{AVP}}_{\text{ALL}}(p,T) \geq \frac{1}{k}\chi\lb T^{\otimes k}\rb - (2kp\log_2(d_Ad_B))^{1/2}|\log\lb\frac{p}{d_B}\rb| - 3p\log_2(d_Ad^2_B) - (1+p) h_2\lb \frac{p}{1+p}\rb,
\]
for any $k\in\N$. In the quantum case, we have
\begin{align*}
&Q^{\text{AVP}}_{\text{ALL}}(p,T) \\
&\geq \frac{1}{k}I_{\text{coh}}\lb T^{\otimes k}\rb -  \eta(p^k) - 3(kp\log_2(d_Ad_B))^{1/2}|\log\lb\min\lb \frac{p}{d_A},\frac{p^2}{d_B}\rb\rb| - 2p\log_2(d_Ad^2_B) - (1+p) h_2\lb \frac{p}{1+p}\rb , 
\end{align*}
for any $k\in\N$ such that 
\[
p^k \leq (2+\sqrt{3})/4 \approx 0.067.
\]
In the above bounds, $h_2$ denotes the binary entropy and $\eta:\lbr 0,1\rbr\ra \R^+$ is given by
\[
\eta(\epsilon) = (2\epsilon + 4\sqrt{\epsilon(1-\epsilon)})\log_2(d_A)+ h_2(2\epsilon) + h_2(2\sqrt{\epsilon(1-\epsilon)}).
\]
\end{thm}

Note that the lower bounds in the previous theorem cannot be regularized, i.e., for fixed $p>0$ the limit $k\ra\infty$ leads to a vanishing lower bound. However, they are still strong enough to show the following point-wise continuity statement. 

\begin{thm}[Point-wise continuity of capacities under AVP]\label{thm:continuityAVP}
For every quantum channel $T:\M_{d_A}\ra\M_{d_B}$ and every $\epsilon>0$, there exists a threshold $p(\epsilon,T)>0$ such that 
\[
C^{\text{AVP}}_{\text{ALL}}(p,T) \geq  C(T) - \epsilon,
\]
and 
\[
Q^{\text{AVP}}_{\text{ALL}}(p,T) \geq  Q(T) - \epsilon,
\]
for all $0\leq p\leq p(\epsilon,T)$. In particular, we have 
\[
\lim_{p\searrow 0} C^{\text{AVP}}_{\text{ALL}}(p,T) = C(T) \quad\text{ and }\quad\lim_{p\searrow 0} Q^{\text{AVP}}_{\text{ALL}}(p,T) = Q(T),
\]
for all quantum channels $T:\M_{d_A}\ra\M_{d_B}$.
\end{thm}

\begin{proof}
We will only state the proof for the classical capacity $C(T)$ since it is the same for the quantum capacity $Q(T)$. Let $T:\M_{d_A}\ra\M_{d_B}$ denote a fixed quantum channel. For every $\epsilon>0$ there exists a $k_\epsilon\in\N$ such that $\frac{1}{k_\epsilon}\chi\lb T^{\otimes k_\epsilon}\rb \geq C(T) -\epsilon/2$. Using Theorem \ref{thm:LowerBoundEffCap}, we then find $p(\epsilon,T)\in \lbr 0,1\rbr$ such that 
\[
C^{\text{AVP}}_{\text{ALL}}(p,T) \geq \frac{1}{k_{\epsilon}}\chi\lb T^{\otimes k_{\epsilon}}\rb - \frac{\epsilon}{2} \geq C(T)-\epsilon, 
\] 
for all $0\leq p\leq p(\epsilon,T)$. This finishes the proof.
\end{proof}

The main difficulty in the proof of Theorem \ref{thm:LowerBoundEffCap} comes from the fact that the quantum channel under AVP (see~\eqref{equ:AVPChannelModel}) partially acts on the environment state $\sigma_E$, which, in general, is entangled over the applications of the channel. Therefore, we are in a setting not covered by standard i.i.d.~quantum information theory. We will apply the following strategy:
\begin{enumerate}
\item Find a coding scheme for classical or quantum communication for the quantum channel $\tilde{T}_p$ from Lemma \ref{lem:PostSelect} at a fixed blocklength $k\in\N$.
\item Use Lemma \ref{lem:PostSelect} to show that the coding scheme from 1. is a coding scheme under AVP of strength $p$ for the original quantum channel $T$.
\item Apply a continuity inequality for the quantities $\chi$ and $I_{\text{coh}}$ to relate the resulting capacity bound involving $\tilde{T}_p$ to a similar bound involving the original channel $T$. 
\end{enumerate} 
Note that step 1. in the previous strategy is straightforward using standard techniques from quantum Shannon theory (i.e., random code constructions). To execute step 2. we need to know precise error bounds in the coding theorems used for step 1., because these errors have to vanish quickly enough to compensate the exponentially growing factor arising from Lemma \ref{lem:PostSelect}. See Appendix \ref{app:HSWError} and Appendix \ref{app:LSDError} for a review of the explicit error bounds we are using in our proof. Step 3. is again straightforward.

\begin{proof}[Proof of Theorem \ref{thm:LowerBoundEffCap}]
For $p\in \lbr 0,1\rb$, we will focus on the lower bound for the quantum capacity $Q^{\text{AVP}}_{\text{ALL}}(p,T)$ stated in the theorem. The lower bound on the effective classical capacity $C^{\text{AVP}}_{\text{ALL}}(p,T)$ follows along the same lines avoiding some additional technicalities as explained at the end of this proof. 

Consider the quantum channel $\tilde{T}_p:\M_{d_A}\ra\M_{d_B}$ given by
\[
\tilde{T}_{p}=(1-p)T+p\frac{\one_{d_B}}{d_B}\Tr.
\] 
Fix $k\in\N$ such that the condition from the theorem is satisfied. Applying Corollary \ref{cor:DecouplWithErrorChanFidelity} for the quantum channel $\tilde{T}^{\otimes k}_p$, some pure state $\ket{\phi_{AA'}}\in \C^{d^k_A}\otimes \C^{d^k_A}$, some $R> 0$ and each $m\in\N$ shows the existence of encoders $E_m:\M^{\otimes (Rkm-1)}_2\ra \M^{\otimes km}_{d_A}$ and decoders $D_m:\M^{\otimes km}_{d_B}\ra \M^{\otimes (Rkm-1)}_2$ such that 
\[
F\lb D_m\circ \tilde{T}_p^{\otimes km}\circ E_m\rb \geq 1-\epsilon_m
\]
with 
\begin{equation}
\epsilon_m = 8\sqrt{3}\exp\lb - \frac{m\delta^2}{\log\lb \mu_{\min}\lb \phi_{A'}\rb\rb^2}\rb + 2\cdot2^{-\frac{mk}{2}\lb \frac{1}{k}I_{\text{coh}}\lb \phi_{A'},\tilde{T}^{\otimes k}_p\rb - R - 3\frac{\delta}{k}\rb}
\label{equ:FTQCError1n} 
\end{equation}
and
\[
\mu_{\min}(\phi_{A'}) = \min\lset \lambda>0 ~:~\lambda\in \text{spec}\lb \phi_{A'}\rb\cup \text{spec}\lb \tilde{T}^{\otimes k}_p(\phi_{A'})\rb\cup\text{spec}\lb \lb\tilde{T}^{\otimes k}_p\rb^c(\phi_{A'})\rb\rset .
\]
Here, we use the minimum fidelity (cf.~\cite{kretschmann2004tema})
\[
F(S)=\min\lset \bra{\psi}S(\proj{\psi}{\psi})\ket{\psi} ~:~\ket{\psi}\in\C^{d}, \braket{\psi}{\psi}=1\rset,
\]
to quantify the distance between a quantum channel $S:\M_d\ra\M_d$ and the identity channel. Fix any sequence of quantum channels $N_m:\M_{d_A}\otimes \M_{d_{E,m}}\ra \M_{d_B}$ for any sequence of dimensions $d_{E,m}\in \N$. We will now use the coding scheme for the quantum channel $\tilde{T}_{p}$ to send quantum information over the quantum channel $T$ under AVP (cf., Definition \ref{defn:CodingSchemesEffQuantCap}) with the quantum channels $T_{p,N_m}:\M_{d_A}\otimes\M_{d_{E,m}}\ra \M_{d_B}$ given by 
\[
T_{p,N_m} = (1-p)T\otimes \Tr_E + p N_m .
\]
For any syndrome state $\sigma_{E,m}\in \D\lb(\C^{d_{E,m}})^{\otimes m}\rb$ we may use the Fuchs-van-de-Graaf inequalities to compute 
\begin{align}
\| \ident^{\otimes (Rkm-1)}_2 -   D_m\circ T_{p,N_m}^{\otimes km}\circ (E_m\otimes \sigma_{E,m})\|_{1\ra 1} &\leq 2\lb 1-F\lb D_m\circ T_{p,N_m}^{\otimes km}\circ (E_m\otimes \sigma_{E,m})\rb\rb^{1/2} \nonumber\\
&\hspace{-6cm}\leq  2\lb 2^{km\log_2(d_Ad_B)(p+\tilde{\delta})}\lb 1- F\lb D_m\circ \tilde{T}^{\otimes km}_p\circ E_m\rb\rb +\exp\lb-mk\frac{\delta^2}{3p}\rb  \rb^{1/2},
\label{equ:FTQC2n}
\end{align}
where we used Lemma \ref{lem:PostSelect} in the final line, and where $\tilde{\delta}>0$ may be chosen as small as we need it to be. With \eqref{equ:FTQCError1n} we have
\begin{equation}
2^{km\log_2(d_Ad_B)(p+\tilde{\delta})}\lb 1- F\lb D_m\circ \tilde{T}^{\otimes km}_p\circ E_m\rb\rb \ra 0
\label{equ:FTQCErrorConv0n}
\end{equation}
as $m\ra\infty$ when 
\begin{equation}
2^{-m\lb \frac{\delta^2\log_2(e)}{\log\lb \mu_{\min}\lb \phi_{A'}\rb\rb^2} - k\log_2(d_Ad_B)(p+\tilde{\delta})\rb} \ra 0
\label{equ:FTQCErrorConv1n}
\end{equation}
and 
\begin{equation}
2^{-\frac{mk}{2}\lb\frac{1}{k}I_{\text{coh}}\lb \phi_{A'},\tilde{T}^{\otimes k}_p\rb - R - 3\frac{\delta}{k} - 2\log_2(d_Ad_B)(p+\tilde{\delta}) \rb } \ra 0
\label{equ:FTQCErrorConv2n}
\end{equation}
as $m\ra \infty$. To guarantee \eqref{equ:FTQCErrorConv1n} we choose 
\[
\delta = (kp\log_2(d_Ad_B))^{1/2}|\log\lb\mu_{\min}\lb \phi_{A'}\rb\rb|,
\]
and $\tilde{\delta}>0$ sufficiently small. Now, \eqref{equ:FTQCErrorConv2n} is satisfied for $\tilde{\delta}>0$ sufficiently small whenever
\begin{equation}
R<\frac{1}{k}I_{\text{coh}}\lb \phi_{A'},\tilde{T}^{\otimes k}_p\rb - \frac{3}{k^{1/2}}(p\log_2(d_Ad_B))^{1/2}|\log\lb\mu_{\min}\lb \phi_{A'}\rb\rb| - 2p\log_2(d_Ad_B).
\label{equ:RateBound1n}
\end{equation}
We have constructed a sequence of $(mkR-1,mk,\epsilon_m)$ coding schemes under AVP of strength $p$ as in Definition \ref{defn:CodingSchemesEffQuantCap} for any $R$ satisfying \eqref{equ:RateBound1n} and some sequence of communication errors $(\epsilon_m)_{m\in\N}$ such that $\epsilon_m\ra 0$ as $m\ra\infty$. Using \cite[Lemma 7.1]{kretschmann2004tema} it is easy to find for any $m\in\N$ a $(mR,m,\epsilon'_m)$ coding schemes under AVP of strength $p$ such that $\epsilon'_m\ra 0$ as $m\ra \infty$. This shows that any rate $R$ satisfying \eqref{equ:RateBound1n} is an achievable rate under AVP of strength $p$ as in Definition \ref{defn:CodingSchemesEffQuantCap}. We conclude that
\begin{equation}
Q^{\text{AVP}}_{\text{ALL}}(p,T)\geq \frac{1}{k}I_{\text{coh}}\lb \phi_{A'},\tilde{T}^{\otimes k}_p\rb - \frac{3}{k^{1/2}}(p\log_2(d_Ad_B))^{1/2}|\log\lb\mu_{\min}\lb \phi_{A'}\rb\rb| - 2p\log_2(d_Ad_B),
\label{equ:FTCapBoundInProofn}
\end{equation}
for any pure state $\ket{\phi_{AA'}}\in \C^{d^k_A}\otimes \C^{d^k_A}$. To obtain an expression in terms of the channel $T$, we fix a pure state $\ket{\psi_{AA'}}\in \C^{d^k_A}\otimes \C^{d^k_A}$ such that 
\[
\frac{1}{k}I_{\text{coh}}\lb \psi_{A'},T^{\otimes k}\rb = \max_{\ket{\phi_{AA'}}} \frac{1}{k}I_{\text{coh}}\lb \phi_{A'},T^{\otimes k}\rb =:\frac{1}{k}I_{\text{coh}}\lb T^{\otimes k}\rb.
\]
For $\epsilon=p^{k}$, we consider a purification $\ket{\psi^{\epsilon}_{AA'}}\in \C^{d^k_A}\otimes \C^{d^k_A}$ of the state
\[
\psi^{\epsilon}_{A'} = (1-2\epsilon)\psi_{A'} + \epsilon \lb\frac{\one_{d_A}}{d_A}\rb^{\otimes k} + \epsilon \proj{\tau_{A'}}{\tau_{A'}},
\]
for some fixed pure state $\ket{\tau_{A'}}\in \C^{d_A^k}$, such that 
\[
\|\psi_{AA'}-\psi^{\epsilon}_{AA'}\|_1 \leq 4\sqrt{\epsilon(1-\epsilon)},
\]
which is possible by Uhlmann's theorem~\cite[Theorem 3.22]{watrous2018theory} and the Fuchs-van-de-Graaf inequality~\cite[Theorem 3.33]{watrous2018theory}. Using the definition of the coherent information and the Fannes-Audenaert inequality (see~\cite[Theorem 5.26]{watrous2018theory}), we find that
\begin{align}\label{equ:cohInfoBoundn}
|\frac{1}{k}I_{\text{coh}}\lb \psi_{A'},\tilde{T}^{\otimes k}_p\rb -\frac{1}{k}I_{\text{coh}}\lb \psi^{\epsilon}_{A'},\tilde{T}^{\otimes k}_p\rb| &\leq (2\epsilon + 4\sqrt{\epsilon(1-\epsilon)})\log_2(d_A)+ h_2(2\epsilon) + h_2(2\sqrt{\epsilon(1-\epsilon)}) \nonumber\\
&=:\eta(\epsilon),
\end{align}
where we have used that $2\sqrt{\epsilon(1-\epsilon)}\leq 1/2$ by the assumption on $k$ and $p$. Furthermore, we note that 
\begin{equation}\label{equ:muMinEstn}
\mu_{\min}\lb \psi^{\epsilon}_{A'}\rb \geq \lb\min\lb \frac{p}{d_A},\frac{p^2}{d_B}\rb\rb^k,
\end{equation}
since 
\begin{align*}
\min&\lset \lambda>0 ~:~\lambda\in\text{spec}\lb \psi^{\epsilon}_{A'}\rb\rset \geq \frac{\epsilon}{d_A^{k}} =  \lb\frac{p}{d_A} \rb^k,\\
\min&\lset \lambda>0 ~:~\lambda\in\text{spec}\lb \tilde{T}^{\otimes k}_p\lb\psi^{\epsilon}_{A'}\rb\rb\rset \geq \lb\frac{p}{d_B} \rb^k,\\
\min&\lset \lambda>0 ~:~\lambda\in\text{spec}\lb \lb\tilde{T}^{\otimes k}_p\rb^{c}\lb\psi^{\epsilon}_{A'}\rb\rb\rset \geq \epsilon\lb\frac{p}{d_B} \rb^k = \lb\frac{p^2}{d_B} \rb^k .
\end{align*}
For the first two estimates we used the special form of $\psi^{\epsilon}_{A'}$ and $\tilde{T}_p$, and for the final estimate the fact that the non-zero spectrum of $\lb\tilde{T}^{\otimes k}_p\rb^{c}\lb\proj{\tau_{A'}}{\tau_{A'}}\rb$ and $\lb\tilde{T}^{\otimes k}_p\rb\lb\proj{\tau_{A'}}{\tau_{A'}}\rb$ coincides (see~\cite[Theorem 3]{king2005properties}). After inserting the state $\psi^{\epsilon}_{A'}$ into \eqref{equ:FTCapBoundInProofn} and using both \eqref{equ:cohInfoBoundn} and \eqref{equ:muMinEstn} we find that 
\begin{equation}\label{equ:FTCapBoundInProof2n}
Q^{\text{AVP}}_{\text{ALL}}(p,T)\geq \frac{1}{k}I_{\text{coh}}\lb \psi_{A'},\tilde{T}^{\otimes k}_p\rb - \eta(\epsilon) - 3(kp\log_2(d_Ad_B))^{1/2}|\log\lb\min\lb \frac{p}{d_A},\frac{p^2}{d_B}\rb\rb| - 2p\log_2(d_Ad_B) .
\end{equation}
By the continuity bound~\cite[Proposition 3A]{shirokov2017tight} (see also~\cite{leung2009continuity}) we find that 
\[
\frac{1}{k}\big{|}I_{\text{coh}}\lb \psi_{A'},\tilde{T}^{\otimes k}_p\rb - I_{\text{coh}}\lb T^{\otimes k}\rb\big{|} \leq 2p\log_2(d_B) + (1+p) h_2\lb \frac{p}{1+p}\rb.
\]
Combining this bound with \eqref{equ:FTCapBoundInProof2n} leads to 
\begin{align*}
&Q^{\text{AVP}}_{\text{ALL}}(p,T)\\
&\geq \frac{1}{k}I_{\text{coh}}\lb T^{\otimes k}\rb -  \eta(\epsilon) - 3(kp\log_2(d_Ad_B))^{1/2}|\log\lb\min\lb \frac{p}{d_A},\frac{p^2}{d_B}\rb\rb| - 2p\log_2(d_Ad^2_B) - (1+p) h_2\lb \frac{p}{1+p}\rb
\end{align*}
This finishes the proof.

The lower bound on the classical capacity $C^{\text{AVP}}_{\text{ALL}}(p,T)$ under AVP follows along the same lines as the proof above. To construct a coding scheme for the quantum channel $\tilde{T}_p$ it is convenient to use the standard techniques outlined in Appendix \ref{app:HSWError}. In particular, the error bound from Theorem \ref{thm:HSWErrorBound} replaces the error bound from \eqref{equ:FTQCError1n} in the previous proof. There are two simplifications compared to the quantum case: First, the classical communication error $\epsilon_{cl}$ already has the required monotonicity property under the partial order $\leq$ appearing in Lemma \ref{lem:PostSelect} and we do not need to relate it to another distance measure. Second, both values $\lambda_{\min}$ and $\mu_{\min}$ appearing in Theorem \ref{thm:HSWErrorBound} can be lower bounded by $(p/d_B)^k$ for the special channel $\tilde{T}^{\otimes k}_p$. Therefore, the analogue of \eqref{equ:muMinEstn} can be obtained much easier without introducing another state like $\psi^{\epsilon}_{A'}$ in the previous proof. Finally, we applied the continuity bound from~\cite[Proposition 3A]{shirokov2017tight} (see also~\cite{leung2009continuity}) to make the final estimate. 

\end{proof}

\section{Fault-tolerant capacities}
\label{sec:FTCapacities}

Definition \ref{defn:ClCapCQ} of the classical capacity $C(T)$ of a cq-channel $T$ assumes that the decoder can be applied without faults. This assumption might not be realistic in practice, since the decoder necessarily performs a measurement of a potentially large quantum state. This reasoning applies also to the classical capacity $C(T)$ and the quantum capacity $Q(T)$ of a quantum channel $T$ from Definition \ref{defn:CCapAndQCap} considering coding schemes even more involved. In this section, we will introduce fault-tolerant versions of the aforementioned capacities. Since our circuit model (including the noise model) is based on qubits, we will state definitions of fault-tolerant coding schemes for cq-channels of the form $T:\mathcal{A}\ra \M^{\otimes j}_2$ and for quantum channels of the form $T:\M^{\otimes j_1}_2\ra\M^{\otimes j_2}_2$, i.e., with input and output quantum systems composed from qubits. However, these definitions also apply for general quantum channels by simply embedding them into a multi-qubit quantum channel. Our results are therefore stated for general channels.

\subsection{Fault-tolerant capacity of a classical-quantum channel}
\label{sec:FTCapCQ}

To define the fault-tolerant capacity of a cq-channel, we will first define fault-tolerant coding schemes taking into account the faults occuring in quantum circuits executed by the receiver. 

\begin{defn}[Fault-tolerant coding scheme for cq-channels]\label{defn:FTCodingSchemeCQChannel}
For $p\in\lbr 0,1\rbr$ let $\mathcal{F}_{\PIID}(p)$ denote the i.i.d.~Pauli noise model from Definition \ref{defn:NoiseModel}. For $n,m\in\N$ and $\epsilon\in\R^+$, an $(n,m,\epsilon)$-fault tolerant coding scheme for classical communication over the cq-channel $T:\mathcal{A}\ra \M^{\otimes j}_2$ under the noise model $\mathcal{F}_{\PIID}(p)$ consists of a (classical) map $E:\lset 1,\ldots ,2^n\rset\ra \mathcal{A}^{m}$ and a quantum circuit with classical output $\Delta:\lb\M^{\otimes j}_2\rb^{\otimes m}\ra \C^{2^n}$ such that the classical communication error (see \eqref{equ:epsilonCL}) satisfies
\[
\epsilon_{cl}\lb \lbr \Delta\rbr_{\mathcal{F}_{\PIID}(p)}\circ T^{\otimes m}\circ E\rb \leq \epsilon .
\]
\end{defn} 

Now, we can define the fault-tolerant capacity as follows:

\begin{defn}[Fault-tolerant capacity of a cq-channel]
For a cq-channel $T:\mathcal{A}\ra\M_d$, and the i.i.d.~Pauli noise model $\mathcal{F}_{\PIID}(p)$ with $p\in\lbr 0,1\rbr$ (see Definition \ref{defn:NoiseModel}) we define the fault-tolerant classical capacity as
\[
C_{\mathcal{F}_{\PIID}(p)}(T) = \sup\lset R\geq 0 \text{ fault-tolerantly achievable rate for }T\rset.
\] 
Here, $R\geq 0$ is called an \emph{fault-tolerantly achievable rate} if for every $m\in\N$ there exists an $n_m\in\N$ and an $(n_m,m,\epsilon_m)$-fault-tolerant coding scheme for classical communication over the cq-channel $T$ under the noise model $\mathcal{F}_{\PIID}(p)$ such that $\epsilon_m\ra 0$ as $m\ra\infty$ and
\[
R \leq \liminf_{m\ra \infty}\frac{n_m}{m}.
\]
\label{defn:FTCQCap}
\end{defn} 

Although Definition \ref{defn:FTCodingSchemeCQChannel} allows for arbitrary decoding circuits $\Delta$, we will consider decoding circuits of the form $\Delta = \Gamma_\mathcal{C}^D\circ \Enc_{\mathcal{C}}^{\otimes jm}$. Here, $\mathcal{C}$ is the concatenated $7$-qubit Steane code at some concatentation level and $\Enc_{\mathcal{C}}$ the interface circuit constructed in Section \ref{sec:NoisyInterface}. The effective channel for this code and interface (cf.~Theorem \ref{thm:EffCommChan}) will have a tensor product structure allowing the use of coding schemes for compound channels~\cite{bjelakovic2009classical,datta2009classical,mosonyi2015coding}) to construct the decoding circuit $\Gamma^D$. Using more advanced quantum error correcting codes~\cite{gottesman2014fault,fawzi2018constant}  might lead to more complicated effective channels, but possibly to higher information transmission rates. We leave this for further investigation. We will show the following:

\begin{thm}[Lower bound on the fault-tolerant capacity of a cq-channel]
Let $p_0$ denote the threshold of the concatenated $7$-qubit Steane code (see Lemma \ref{lem:Strongthreshold}) and $c>0$ the constant from Theorem \ref{thm:NiceInterfaceConcatenated}. For any cq-channel $T:\mathcal{A}\ra\M_d$ and any $p\leq \min(p_0/2,(2c\lceil\log_2(d)\rceil)^{-1})$ we have 
\[
C_{\mathcal{F}_{\PIID}(p)}(T) \geq C(T) - 2cp\lceil\log_2(d)\rceil^2 - 2\lb 1+2cp\lceil\log_2(d)\rceil\rb h_2\lb\frac{2cp\lceil\log_2(d)\rceil}{1+2cp\lceil\log_2(d)\rceil}\rb, 
\]
where $h_2$ denotes the binary entropy.
\label{thm:LowerBoundFTCQ}
\end{thm}

The previous theorem directly implies the following threshold-type result. 

\begin{thm}[Threshold theorem for the fault-tolerant capacity of a cq-channel]
For every $\epsilon>0$ and every $d\in\N$ there exists a threshold $p(\epsilon,d)>0$ such that 
\[
C_{\mathcal{F}_{\PIID}(p)}(T) \geq C(T) - \epsilon,
\]
for all cq-channels $T:\mathcal{A}\ra\M_d$ and all $0\leq p\leq p(\epsilon,d)$. In particular, we have 
\[
\lim_{p\searrow 0} C_{\mathcal{F}_{\PIID}(p)}(T) = C(T).
\]
\label{thm:ThresholdCQ}
\end{thm}

We will now prove the lower bound on the fault-tolerant capacity of a cq-channel stated above.

\begin{proof}[Proof of Theorem \ref{thm:LowerBoundFTCQ}]
Without loss of generality we may assume that $T:\mathcal{A}\ra\M^{\otimes j}_2$ with $j=\lceil\log_2(d)\rceil$. Our fault-tolerant coding scheme will use the concatenated $7$-qubit Steane code (see Appendix \ref{sec:Appendix}) as a quantum circuit code. For each $l\in\N$ let $\mathcal{C}_l\subset (\C^2)^{\otimes 7^l}$ denote the $l$th level of the concatenated $7$-qubit Steane code with threshold $p_0\in\lb 0,1\rbr$ (see Lemma \ref{lem:Strongthreshold}). Moreover, we denote by $\Enc_l:\M_2\ra\M^{\otimes 7^l}_2$ the interface circuit from \eqref{equ:EncDec1}, and recall Theorem \ref{thm:NiceInterfaceConcatenated} introducing a constant $c>0$.

We will start by constructing a fault-tolerant coding scheme as in Definition \ref{defn:FTCodingSchemeCQChannel}. For any cq-channel $N:\mathcal{A}\ra\M^{\otimes j}_2$ and $p\in \lbr 0,(2cj)^{-1}\rbr$ we denote by $T_{p,N}:\mathcal{A}\ra\M^{\otimes j}_2$ the cq-channel
\[
T_{p,N} = (1-2cpj)T + 2cpjN.
\]
Next, we fix a probability distribution $q\in \mathcal{P}\lb \mathcal{A}\rb$ and a rate
\begin{equation}
R <  \inf_{N} \chi\lb\lset q_i,T_{p,N}(i)\rset\rb,
\label{equ:CQRate}
\end{equation}
with infimum running over cq-channels $N:\mathcal{A}\ra\M^{\otimes j}_2$. Applying \cite[Theorem IV.18.]{mosonyi2015coding} we obtain a sequence $n_m\in\N$ satisfying 
\[
R\leq \liminf_{m\ra\infty} \frac{n_m}{m},
\]
and maps $E_m:\lset 1,\ldots ,2^{n_m}\rset\ra \mathcal{A}^{m}$ and quantum-classical channels $D_m:\M^{\otimes jm}_2\ra \C^{\otimes 2^{n_m}}$ such that 
\begin{equation}
\sup_{N} \epsilon_{cl}\lb D_m\circ T_{p,N}^{\otimes m}\circ E_m\rb \ra 0
\label{equ:cqLowerBError0}
\end{equation}
as $m\ra \infty$ and with $\epsilon_{cl}$ as in \eqref{equ:epsilonCL}. For each $m\in\N$ let $\Gamma^{D,m}:\M^{\otimes jm}_2\ra \C^{\otimes 2^{n_m}}$ denote a quantum circuit satisfying 
\begin{equation}
\|\Gamma^{D,m} - D_m\|_{1\ra 1} \leq \frac{1}{m},
\label{equ:cqLowerBError1}
\end{equation}
and choose $l_m\in\N$ such that 
\begin{equation}
2c\lb \frac{p}{p_0}\rb^{2^{l_m}}|\Loc(\Gamma^{D,m})|\leq \frac{1}{m}.
\label{equ:cqLowerBError2}
\end{equation}
By the previous bound and the second case of Lemma \ref{lem:Main} (using Bernoulli's inequality) we find that
\[
\| \lbr\Gamma_{\mathcal{C}_{l_m}}^{D,m}\circ \Enc_{l_m}^{\otimes jm}\rbr_{\mathcal{F}_{\PIID}(p)} - \Gamma^{D,m}\circ (N^{l_m}_{p})^{\otimes m}\|_{1\ra 1}\leq \frac{1}{m}, 
\]
where $N^{l_m}_{p}:\M^{\otimes j}_2\ra\M^{\otimes j}_2$ denotes a quantum channel of the form 
\[
N^{l_m}_p = (1-2cpj)\ident_2 + 2cpj N_{l_m},
\]
for a quantum channel $N_{l_m}:\M^{\otimes j}_2\ra\M^{\otimes j}_2$ depending on $l_m\in\N$. Using the particular form of the coding error $\epsilon_{cl}$ from \eqref{equ:epsilonCL} and the estimate \eqref{equ:cqLowerBError1} leads to   
\begin{align}
\epsilon_{cl}\lb \lbr \Gamma_{\mathcal{C}_{l_m}}^{D,m}\circ\Enc^{\otimes jm}_{l_m}\rbr_{\mathcal{F}_{\PIID}(p)} \circ T^{\otimes m}\circ E_m \rb &\leq \epsilon_{cl}\lb \Gamma^{D,m}\circ \lb T_{p,N_{l_m}}\rb^{\otimes m}\circ E_m \rb + \frac{1}{m} \nonumber\\
&\leq \epsilon_{cl}\lb D_m\circ \lb T_{p,N_{l_m}}\rb^{\otimes m}\circ E_m \rb + \frac{2}{m}\nonumber\\
&\longrightarrow 0 \quad \text{ as }\quad m\ra \infty,
\label{equ:WithEffectiveChannel}
\end{align}
where we used \eqref{equ:cqLowerBError0} in the final line. We have shown that any rate $R$ chosen as in \eqref{equ:CQRate} is fault-tolerantly achievable in the sense of Definition \ref{defn:FTCQCap}. We conclude that 
\begin{equation}
C_{\mathcal{F}_{\PIID}(p)}(T) \geq \inf_{N} \chi\lb\lset q_i,T_{p,N}(i)\rset\rb,
\label{equ:CQCapBoundCompound}
\end{equation}
for any probability distribution $q\in \mathcal{P}\lb \mathcal{A}\rb$. Finally, we use the continuity bound from \cite[Proposition 5]{shirokov2017tight} to estimate
\begin{align*}
|\chi\lb\lset q_i,T_{p,N}(i)\rset\rb - \chi\lb\lset q_i,T(i)\rset\rb| &\leq \epsilon_0j + 2(1+\epsilon_0)h_2\lb\frac{\epsilon_0}{1+\epsilon_0}\rb \\
&\leq 2cpj^2 + 2(1+2cpj)h_2\lb\frac{2cpj}{1+2cpj}\rb,
\end{align*}
where we used that the function $x\mapsto (1+x)h_2\lb \frac{x}{1+x}\rb$ is monotonically increasing and that
\[
\epsilon_0 := \frac{1}{2}\sum_i \| q_iT_{p,N}(i) - q_iT(i)\|_1 = cpj\sum_i q_i\| T(i)-N(i)\|_1\leq 2cpj.
\]
Combining this estimate with \eqref{equ:CQCapBoundCompound} we have
\[
C_{\mathcal{F}_{\PIID}(p)}(T) \geq \chi\lb\lset q_i,T(i)\rset\rb - 2cpj^2 - 2(1+2cpj)h_2\lb\frac{2cpj}{1+2cpj}\rb
\]
for any probability distribution $q\in \mathcal{P}\lb \mathcal{A}\rb$. Maximizing over $q$ and using the Holevo-Schumacher-Westmoreland theorem (see Theorem \ref{thm:HSWTheorem}) finishes the proof. 
\end{proof}

The previous proof used a coding scheme for a so-called compound channel~\cite{bjelakovic2009classical,datta2009classical,mosonyi2015coding} in~\eqref{equ:cqLowerBError0}. The same proof would also work for a coding scheme that is constructed for the specific sequence of tensor powers $\lb T_{p,N_{l_m}}\rb^{\otimes m}$ appearing in the first and second line of \eqref{equ:WithEffectiveChannel}. However, due to the dependence of the local channel on $m$ through the concatenation level $l_m$ of the concatenated code, this is not the tensor power of a fixed qubit channel. Standard techniques from i.i.d.~quantum information theory might therefore not apply in this setting, and similar constructions as for compound channels might be needed here in general.

\subsection{Fault-tolerant capacities of quantum channels}
\label{sec:FTCapQC}

Next, we consider fault-tolerant capacities of quantum channels. We will focus on the classical capacity and the quantum capacity introduced in Section \ref{sec:CapQuChannel}. We will begin by introducing fault-tolerant coding schemes for transmitting classical information over quantum channels.  

\begin{defn}[Fault-tolerant coding schemes for classical communication]\label{defn:FTcodingSchemesClassCap}
For $p\in\lbr 0,1\rbr$ let $\mathcal{F}_{\PIID}(p)$ denote the i.i.d~Pauli noise model from Definition \ref{defn:NoiseModel} and consider a quantum channel $T:\M^{\otimes j_1}_2\ra\M^{\otimes j_2}_2$. For $n,m\in\N$ and $\epsilon\in\R^+$, an $(n,m,\epsilon)$ fault-tolerant coding scheme for classical communication over $T$ under the noise model $\mathcal{F}_{\PIID}(p)$ consists of a quantum circuit $E:\C^{2^n}\ra\M^{\otimes j_1m}_{2}$ with classical input and a quantum circuit $\Delta:\M^{\otimes j_2m}_{2}\ra \C^{2^n}$ with classical output such that 
\[
\epsilon_{cl}\lb \lbr \Delta\circ T^{\otimes m}\circ E\rbr_{\mathcal{F}_{\PIID}(p)}\rb \leq \epsilon ,
\]
where $\epsilon_{cl}$ denotes the classical communication error from \eqref{equ:epsilonCL}.
\end{defn}

To define fault-tolerant coding schemes for quantum communication, it will be important to choose a suitable way to measure the communication error. Note that the usual equivalences between different error measures commonly used to define the quantum capacity (see for instance the discussion in~\cite{kretschmann2004tema}) might not hold in a setting where noise is affecting the coding operations. Our definition is motivated by possible applications of fault-tolerant capacities in quantum computing, where a quantum computation is performed on two separate quantum computers connected by a communication line. Intuitively, the fault-tolerant quantum capacity should measure the optimal rates at which quantum information can be send within any quantum computation such that the overall quantum computation, i.e., the combination of the two parts of the executed circuit and the communication scheme, can be implemented fault-tolerantly. For technical reasons, which will become clear in the discussion after the definition, we will need the following consequence of the Solovay-Kitaev theorem and the approximation of multi-qubit unitary operations by universal gate sets (see the discussion in~\cite{Nielsen2007}): For each $m\in\N$ and $\delta>0$, we denote by $N(m,\delta)\in \N$ the smallest number such that for any classical-quantum channel $R:\C^{2^k}\ra \M^{\otimes m}_2$ there exists a quantum circuit $\Gamma_R$ with $|\text{Loc}\lb \Gamma_R\rb|\leq N(m,\delta)$ and $\|\Gamma_R - R\|_{1\ra 1} \leq \delta$, and such that the analogous approximation holds for any quantum-classical channel $R':\M^{\otimes m}_2\ra \C^{2^k}$ as well. Note that $N(m,\delta)$ is monotonically decreasing in $\delta$. Now, we can state our definition:

\begin{defn}[Fault-tolerant coding scheme for quantum communication]\label{defn:FTcodingSchemesQuantCap}
For $p\in\lbr 0,1\rbr$ let $\mathcal{F}_{\PIID}(p)$ denote the i.i.d.~Pauli noise model from Definition \ref{defn:NoiseModel} and consider a quantum channel $T:\M^{\otimes j_1}_2\ra\M^{\otimes j_2}_2$. For $n,m\in\N$ and $\epsilon >0$ an $(n,m,\epsilon)$-fault-tolerant coding scheme for quantum communication over $T$ under the noise model $\mathcal{F}_{\PIID}(p)$ consists of a quantum error correcting code $\mathcal{C}$ (depending on $n$ and $m$) of dimension $\text{dim}\lb \mathcal{C}\rb$ on $K_{\mathcal{C}}$ qubits and quantum circuits
\[
\Gamma^E(\mathcal{C}):\M^{\otimes nK_{\mathcal{C}}}_2\ra \M^{\otimes mj_1}_2 \text{ and } \Gamma^D(\mathcal{C}):\M^{\otimes mj_2}_2\ra \M^{\otimes nK_{\mathcal{C}}}_2,
\]
such that 
\[
\| \Gamma^2\circ \Gamma^1 - \lbr \Gamma_{\mathcal{C}}^2\circ \Gamma^D(\mathcal{C})\circ T^{\otimes m}\circ \Gamma^E(\mathcal{C})\circ \Gamma_{\mathcal{C}}^1\rbr_{\mathcal{F}_{\PIID}(p)}\|_{1\ra 1}\leq \epsilon .
\]
for any $k_1,k_2\in\N$ and all quantum circuits $\Gamma^1:\C^{2^{k_1}}\ra \M^{\otimes n}_2$ and $\Gamma^2:\M^{\otimes n}_2\ra \C^{2^{k_2}}$ satisfying
\[
\max\lb|\text{Loc}\lb \Gamma^1\rb|,|\text{Loc}\lb \Gamma^2\rb|\rb\leq N(n,\delta),
\]
for some $\delta\leq \epsilon$.
\end{defn}

We should emphasize that the previous definition can be stated in the same way with quantum error correcting codes $\mathcal{C}$ of dimensions $\text{dim}\lb\mathcal{C}\rb>2$ (including the case where multiple qubits are encoded in the same code). We have chosen to restrict to dimension $\text{dim}\lb\mathcal{C}\rb=2$ because we have only introduced implementations of circuits in this case (see Definition \ref{defn:Impl}). While a more general definition is possible it would be more involved and beyond the scope of this article. Note, however, that Definition \ref{defn:FTcodingSchemesClassCap} of fault-tolerant coding schemes for classical information is as general as possible. 

Having defined fault-tolerant coding schemes, we can define fault-tolerant capacities as usual:

\begin{defn}[Fault-tolerant classical and quantum capacity]\label{defn:FTCCapAndQCap}
For $p\in\lbr 0,1\rbr$ let $\mathcal{F}_{\PIID}(p)$ denote the i.i.d.~Pauli noise model from Definition \ref{defn:NoiseModel} and consider a quantum channel $T:\M^{\otimes j_1}_2\ra\M^{\otimes j_2}_2$. We call $R\geq 0$ an \emph{fault-tolerantly achievable rate} for classical (or quantum) communication if for every $m\in\N$ there exists an $n_m\in\N$ and an $(n_m,m,\epsilon_m)$ fault-tolerant coding scheme for classical (or quantum) communication over $T$ under the noise model $\mathcal{F}_{\PIID}(p)$ with $\epsilon_m\ra 0$ as $m\ra\infty$ and 
\[
\liminf_{m\ra\infty} \frac{n_m}{m}\geq R.
\] 
The \emph{fault-tolerant classical capacity} of $T$ is given by
\[
C_{\mathcal{F}_{\PIID}(p)}(T)=\sup\lset R\geq 0 \text{ fault-tolerantly achievable rate for classical communication}\rset,
\]
and the \emph{fault-tolerant quantum capacity} of $T$ is given by 
\[
Q_{\mathcal{F}_{\PIID}(p)}(T)=\sup\lset R\geq 0 \text{ fault-tolerantly achievable rate for quantum communication}\rset.
\]
\end{defn}

It should also be emphasized that the noiseless capacity $Q(T)$ is recovered from the previous definition when $p=0$. Indeed, it is easy to see how fault-tolerant coding schemes under the noise model $\mathcal{F}_{\PIID}(0)$ can be constructed from coding scheme for quantum communication over the channel $T$ (see Definition \ref{defn:codingSchemesQChannel}) by using a trivial quantum error correcting code and quantum circuits approximating the coding operations. This shows that $Q(T)\leq Q_{\mathcal{F}_{\PIID}(0)}(T)$. To see the remaining inequality, we choose for every pure state $\ket{\psi}\in(\C^2)^{\otimes n}$ a quantum circuit $\Gamma^{1,\psi}:\C\ra \M^{\otimes n}_2$ preparing $\proj{\psi}{\psi}$ up to an error $\epsilon$ in $1\ra 1$-norm and a quantum circuit $\Gamma^{2,\psi}:\M^{\otimes n}_2\ra \C^{2}$ measuring the POVM $\lset\proj{\psi}{\psi},\one^{\otimes n}_{2} - \proj{\psi}{\psi}\rset$ up to an error $\epsilon$ in $1\ra 1$-norm, and such that each of these circuits has less than $N(n,\epsilon)$ many locations (which is possible by the Solovay-Kitaev theorem). Assume that there exists a $(n,m,\epsilon)$-coding scheme consisting of an error correcting code $\mathcal{C}$ and quantum circuits $\Gamma^E(\mathcal{C})$ and $\Gamma^D(\mathcal{C})$ as in Definition \ref{defn:FTcodingSchemesQuantCap} such that
\[
\| \Gamma^{2,\psi}\circ \Gamma^{1,\psi} - \Gamma_{\mathcal{C}}^{2,\psi}\circ \Gamma^D(\mathcal{C})\circ T^{\otimes m}\circ \Gamma^E(\mathcal{C})\circ \Gamma_{\mathcal{C}}^{1,\psi}\|_{1\ra 1}\leq \epsilon . 
\]
Since there are no faults affecting the circuits, we may apply the transformation rules in Definition \ref{defn:Impl} to conclude that 
\[
\|\Gamma^{2,\psi}\circ \Gamma^{1,\psi} - \Gamma^{2,\psi}\circ \lb \ident^{\otimes n}_2\otimes \text{Tr}_S\rb\circ \text{Dec}^{*} \circ \Gamma^D(\mathcal{C})\circ T^{\otimes m}\circ \Gamma^E(\mathcal{C})\circ \text{Enc}^{*}\circ \lb \Gamma^{1,\psi}\otimes \proj{0}{0}\rb\|_{1\ra 1}\leq \epsilon ,
\]
for any pure state $\ket{\psi}\in(\C^2)^{\otimes n}$, and where $\proj{0}{0}$ denotes the syndrome state corresponding to the zero syndrome. Using the minimum fidelity 
\[
F(S)=\min\lset \bra{\psi}S(\proj{\psi}{\psi})\ket{\psi} ~:~\ket{\psi}\in\C^{d_1}, \braket{\psi}{\psi}=1\rset,
\]
of a quantum channel $S:\M_{d_1}\ra \M_{d_2}$, we conclude that 
\[
F\lb \lb\ident^{\otimes n}_2\otimes \text{Tr}_S\rb\circ \text{Dec}^{*} \circ \Gamma^D(\mathcal{C})\circ T^{\otimes m}\circ \Gamma^E(\mathcal{C})\circ \text{Enc}^{*}\lb \cdot \otimes \proj{0}{0}\rb\rb \geq 1-5\epsilon.
\]
By the equivalences of error criteria from~\cite{kretschmann2004tema}, we conclude that the pair of quantum channels 
\[
D = \lb\ident^{\otimes n}_2\otimes \text{Tr}_S\rb\circ \text{Dec}^{*} \circ \Gamma^D(\mathcal{C}),
\]
and 
\[
E= \Gamma^E(\mathcal{C})\circ \text{Enc}^{*}\lb \cdot \otimes \proj{0}{0}\rb,
\]
defines an $(n,m,5\epsilon)$-coding scheme. After going through the definitions of the capacities, we find that $Q(T)\geq Q_{\mathcal{F}_{\PIID}(0)}(T)$. 

The previous definitions of fault-tolerant capacities and their coding schemes allow for arbitrary quantum error correcting codes to be used by the sender and the receiver. As in Section \ref{sec:FTCapCQ}, we will restrict to concatenated codes with the same concatenation level and the interface circuits constructed in Section \ref{sec:NoisyInterface}. As before, we should emphasize that more advanced quantum error correcting codes~\cite{gottesman2014fault,fawzi2018constant}  might lead to more complicated effective channels, but possibly to higher information transmission rates. We leave this for further investigation. By using concatenated codes, we can use the effective channel from Theorem \ref{thm:EffCommChan} to show that any coding scheme for communication over $T$ under AVP (see Definition \ref{defn:CodingSchemesEffClassCap}) can be used to construct a fault-tolerant coding scheme. This leads to the following lower bounds in the same spirit as Theorem \ref{thm:LowerBoundFTCQ}.

\begin{thm}[Lower bounds on fault-tolerant capacities]\label{thm:FTCCQCLowerBound}
Let $p_0$ denote the threshold of the concatenated $7$-qubit Steane code (see Lemma \ref{lem:Strongthreshold}) and $c>0$ the constant from Theorem \ref{thm:NiceInterfaceConcatenated}. For any quantum channel $T:\M_{d_1}\ra\M_{d_2}$ and any 
\[
p\leq \min(p_0/2,(2(j_1+j_2)c)^{-1}),
\]
with $j_1=\lceil\log_2(d_1)\rceil$ and $j_2=\lceil\log_2(d_2)\rceil$, we have 
\[
C_{\mathcal{F}_{\PIID}(p)}(T) \geq C^{\text{AVP}}_{\text{ALL}}(2(j_1+j_2)cp,T) 
\]
and
\[
Q_{\mathcal{F}_{\PIID}(p)}(T) \geq Q^{\text{AVP}}_{\text{ALL}}(2(j_1+j_2)cp,T) . 
\]
An explicit lower bound can be obtained from Theorem \ref{thm:LowerBoundEffCap}.
\end{thm}

Together with Theorem \ref{thm:continuityAVP}, the previous theorem implies the following threshold-type result.

\begin{thm}[Threshold theorem for fault-tolerant capacities]\label{thm:ThresholdQChan}
For every quantum channel $T:\M_{d_1}\ra\M_{d_2}$ and every $\epsilon>0$, there exists a threshold $p(\epsilon,T)>0$ such that 
\[
C_{\mathcal{F}_{\PIID}(p)}(T) \geq  C(T) - \epsilon,
\]
and 
\[
Q_{\mathcal{F}_{\PIID}(p)}(T) \geq  Q(T) - \epsilon,
\]
for all $0\leq p\leq p(\epsilon,T)$. In particular, we have 
\[
\lim_{p\searrow 0} C_{\mathcal{F}_{\PIID}(p)}(T) = C(T) \quad\text{ and }\quad\lim_{p\searrow 0} Q_{\mathcal{F}_{\PIID}(p)}(T) = Q(T),
\]
for all quantum channels $T:\M_{d_1}\ra\M_{d_2}$.
\end{thm}

It should be emphasized that the threshold in Theorem \ref{thm:ThresholdQChan} does depend on the quantum channel $T$ and it is not uniform as in the case of cq-channels (cf.~Theorem~\ref{thm:ThresholdCQ}). This issue is closely connected to the fact that our bounds in Theorem \ref{thm:LowerBoundEffCap} do not regularize (i.e., the limit $k\ra\infty$ leads to a vanishing lower bound for any $p>0$). This issue might be resolved by either using a different fault-tolerant coding scheme, or by showing that the syndrome states $\sigma_S$ occuring in Theorem \ref{thm:EffCommChan} are fully separable. In the latter case, we could lower bound the fault-tolerant capacities by the capacities under AVP with separable environment states $C^{\text{SEP}}_{\text{ALL}}(2(j_1+j_2)cp,T)$ and $C^{\text{SEP}}_{\text{ALL}}(2(j_1+j_2)cp,T)$ (see Definition \ref{defn:EffCCapAndQCap}). Using Theorem \ref{thm:EffCapSepEnvLowerB} (and a potential analogue of this theorem for the classical capacity) the lower bounds on the fault-tolerant capacities could then be improved and a threshold not depending on the particular channel (as in Theorem~\ref{thm:ThresholdCQ}) could be obtained. 

By restricting our statements to classes of quantum channels where the regularization in the capacity formulas is not required, it is also possible to obtain uniform thresholds using our techniques. For example, there is a threshold $p(\epsilon,d_A,d_B)$ only depending on $\epsilon$ and the dimensions $d_A,d_B$ such that $Q_{\mathcal{F}_{\PIID}(p)}(T) \geq  Q(T) - \epsilon$ holds for every degrabable quantum channel $T:\M_{d_A}\ra\M_{d_B}$ (see~\cite{cubitt2008structure} for the definition and properties of this class of quantum channels) and for every $0\leq p\leq p(\epsilon,d_A,d_B)$.

\begin{proof}[Proof of Theorem \ref{thm:FTCCQCLowerBound}]
We will focus on the lower bound for the fault-tolerant quantum capacity $Q_{\mathcal{F}_{\PIID}(p)}(T)$ stated in the theorem. The lower bound on the fault-tolerant classical capacity $C_{\mathcal{F}_{\PIID}(p)}(T)$ follows along the same lines. 

Without loss of generality we may assume that $T:\M^{\otimes j_1}_{2}\ra\M^{\otimes j_2}_2$ with 
\[
j_1=\lceil\log_2(d_1)\rceil \quad\text{ and }\quad j_2=\lceil\log_2(d_1)\rceil,
\]
and we fix any rate $R< Q^{\text{AVP}}_{\text{ALL}}(2(j_1+j_2)cp,T)$. Our fault-tolerant coding scheme (as in Definition \ref{defn:FTcodingSchemesQuantCap}) achieving this rate will use the concatenated $7$-qubit Steane code (see Appendix \ref{sec:Appendix}) as a quantum circuit code. For each $l\in\N$ let $\mathcal{C}_l\subset (\C^2)^{\otimes 7^l}$ denote the $l$th level of the concatenated $7$-qubit Steane code with threshold $p_0\in\lb 0,1\rbr$ (see Lemma \ref{lem:Strongthreshold}). Moreover, we denote by $\Enc_l:\M_2\ra\M^{\otimes 7^l}_2$ and $\Dec_l:\M^{\otimes 7^l}_2\ra \M_2$ the interface circuits from \eqref{equ:EncDec1} and \eqref{equ:EncDec2}, and recall Theorem \ref{thm:NiceInterfaceConcatenated} introducing the constant $c>0$. Consider a sequence of $(n_m,m,\epsilon_m)$-coding schemes for quantum communication over $T$ under AVP of strength $2(j_1+j_2)cp$ given by quantum channels $E_m:\M^{\otimes n_m}_2\ra \M^{\otimes j_1m}_2$ and $D_m:\M^{\otimes j_2m}_2\ra \M^{\otimes n_m}_2$ for each $m\in \N$ (as in Definition \ref{defn:CodingSchemesEffQuantCap}) such that $R\leq \liminf_{m\ra \infty} n_m/m$ and $\lim_{m\ra\infty}\epsilon_m = 0$. 

To construct a fault-tolerant coding scheme, we consider quantum circuits $\Gamma^{E,m}:\M^{\otimes n_m}_2\ra \M^{\otimes j_1m}_2$ and $\Gamma^{D,m}:\M^{\otimes j_2m}_2\ra \M^{\otimes n_m}_2$ such that 
\begin{equation}
\|\Gamma^{E,m}-E_m\|_\diamond \leq \frac{1}{m},
\label{eq:EapproxProof}
\end{equation}
and 
\begin{equation}
\|\Gamma^{D,m}-D_m\|_\diamond \leq \frac{1}{m},
\label{eq:DapproxProof}
\end{equation}
for every $m\in\N$. Recall the numbers $N(k,\delta)\in\N$ defined in the text before Definition \ref{defn:FTcodingSchemesQuantCap} specifying the number of locations needed in quantum circuits to approximate any classical-quantum and quantum-classical channels with output and input, respectively, consisting of $k$ qubits up to an error $\delta>0$ in $1\ra 1$-norm. After choosing $l_m\in\N$ such that
\begin{equation}
\lb \frac{p}{p_0}\rb^{2^{l_m}}\lb |\text{Loc}\lb \Gamma^{E,m}\rb| + |\text{Loc}\lb \Gamma^{D,m}\rb| + 2N(n_m,\epsilon_m) +j_1m\rb \leq \frac{1}{m},
\label{equ:ErrBoundFTCapThm}
\end{equation}
for every $m\in\N$, we define quantum circuits
\[
\Gamma^{E,m}\lb \mathcal{C}_{l_m}\rb = \lb\Dec^{\otimes j_1}_{l_m}\rb^{\otimes m}\circ \Gamma^{E,m}_{\mathcal{C}_{l_m}},
\]
and 
\[
\Gamma^{D,m}\lb \mathcal{C}_{l_m}\rb = \Gamma_{\mathcal{C}_{l_m}}^{D,m}\circ \lb\Enc^{\otimes j_2}_{l_m}\rb^{\otimes m}.
\]
For any pair of quantum circuits $\Gamma^{1,m}:\C^{2^{k_1}}\ra \M^{n_m}_2$ and $\Gamma^{2,m}:\M^{n_m}_2\ra \C^{2^{k_2}}$ satisfying 
\[
\max\lb |\text{Loc}\lb \Gamma^{1,m}\rb|,|\text{Loc}\lb \Gamma^{2,m}\rb| \rb\leq N(n_m,\epsilon_m),
\]
for any $m\in\N$, we can compute
\begin{align}
&\| \Gamma^{2,m}\circ \Gamma^{1,m} - \lbr \Gamma^{2,m}_{\mathcal{C}_{l_m}}\circ\Gamma^{D,m}\lb \mathcal{C}_{l_m}\rb\circ  T^{\otimes m}\circ\Gamma^{E,m}\lb \mathcal{C}_{l_m}\rb\circ \Gamma^{1,m}_{\mathcal{C}_{l_m}}\rbr_{\mathcal{F}_\PIID(p)}\|_{1\ra 1}\nonumber \\
&\leq \frac{C}{m} + \| \Gamma^{2,m}\circ \Gamma^{1,m} -  \Gamma^{2,m}\circ \Gamma^{D,m}\circ T_{p,l_m}^{\otimes m}\circ\Gamma^{E,m} \circ (\Gamma^{1,m}\otimes \sigma_{S,m})\|_{1\ra 1}\nonumber \\
&\leq \frac{C+2}{m} + \| \Gamma^{2,m}\circ \Gamma^{1,m} -  \Gamma^{2,m}\circ D_m\circ T_{p,l_m}^{\otimes m}\circ E_m\circ (\Gamma^{1,m}\otimes \sigma_{S,m})\|_{1\ra 1}\nonumber\\
&\leq \frac{C+2}{m} + \| \ident^{\otimes Rkm}_2 -   D_m\circ T_{p,l_m}^{\otimes m}\circ (E_m\otimes \sigma_{S,m})\|_{\diamond} \leq \frac{C+2}{m} + \epsilon_m .
\label{equ:FTQC1}
\end{align}
Here, we used Theorem \ref{thm:EffCommChan} together with \eqref{equ:ErrBoundFTCapThm} in the first inequality, \eqref{eq:EapproxProof} and \eqref{eq:DapproxProof} together with the triangle inequality in the second inequality, and finally monotonicity of the $1\ra 1$-norm under quantum channels in the third inequality. Here, $\sigma_{S,m}$ is some quantum state on the syndrome space depending on $\Gamma^{1,m}$ and $\Gamma^{E,m}$, and $T_{p,l_m}\lb\cdot\otimes \sigma_{S,m}\rb$ is the effective quantum channel introduced in Theorem \ref{thm:EffCommChan}, which is a special case of the channel $T$ under AVP (see~\eqref{equ:AVPChannelModel}) and it is therefore corrected up to an error $\epsilon_m$ by the coding scheme defined by $E_m$ and $D_m$. This implies the final inequality. We have shown that each $R< Q^{\text{AVP}}_{\text{ALL}}(2(j_1+j_2)cp,T)$ is a fault-tolerantly achievable rate for quantum communication over the quantum channel $T$ under the noise model $\mathcal{F}_{\PIID}(p)$ and we conclude that $Q_{\mathcal{F}_{\PIID}(p)}(T) \geq Q^{\text{AVP}}_{\text{ALL}}(2(j_1+j_2)cp,T)$.

The lower bound on the fault-tolerant classical capacity $C_{\mathcal{F}_{\PIID}(p)}(T)$ follows along the same lines as the previous proof by using a coding scheme for classical communication over $T$ under AVP of strength $2(j_1+j_2)cp$. Again, we can approximate the coding maps $E_m$ and $D_m$ by quantum circuits and use the concatenated $7$-qubit Steane code at a sufficiently high level to protect these circuits from noise. 
\end{proof}

\subsection{Specific coding schemes from asymptotically good codes}
\label{sec:GoodCodes}

In this section, we will show how to construct fault-tolerant coding schemes for certain quantum channels from asymptotically good codes. For our purposes it will be sufficient to consider such codes between systems where the dimension is a power of two. 

\begin{defn}[Asymptotically good codes]
Let $d=2^j$ be a power of two. An asymptotically good code of rate $R>0$ and goodness $\alpha\in\lb 0,1\rb$ is given by a sequence $\lb (E_m,D_m)\rb_{m\in\N}$ of encoding operations $E_{m}:\M^{\otimes n_m}_2\ra\M^{\otimes m}_d$ and decoding operations $D_m:\M^{\otimes m}_d\ra \M^{\otimes n_m}_2$ such that there is a sequence $(t_m)_{m\in\N}\in\N^{\N}$ satisfying the following:
\begin{enumerate}
\item $\liminf_{m\ra\infty} \frac{n_m}{m}>R$.
\item $\liminf_{m\ra\infty} \frac{t_m}{m}>\alpha$. 
\item For any quantum channel $N:\M^{\otimes m}_d\ra\M^{\otimes m}_d$ acting non-trivially on only $t_m$ qudits, we have 
\[
D_{m}\circ N\circ E_m = \ident^{\otimes n_m}_2 .
\]
\end{enumerate} 
\label{defn:AsymptGoodCode}
\end{defn}

Asymptotically good codes were first constructed for $d=2$ by Calderbank and Shor in~\cite{calderbank1996good}. For 
\begin{equation}
\alpha_0 := \min\lset \alpha\in\lbr 0,1\rbr ~:~h_2(2\alpha)=\frac{1}{2}\rset 
\label{equ:alpha0}
\end{equation}
and any goodness $\alpha\in \lb 0,\alpha_0\rb$ these codes achieve a rate
\[
R(\alpha) := 1-2h_2(2\alpha),
\] 
where $h_2$ denotes the binary entropy $h_2(p)=-(1-p)\log_2(1-p) - p\log_2(p)$. By now, many families of asymptotically good codes are known~\cite{ashikhmin2001asymptotically,chen2001asymptotically,matsumoto2002improvement,li2009family}. We will show the following theorem:

\begin{thm}[Fault-tolerant coding schemes from asymptotically good codes]
Let $d=2^j$ be a power of $2$ and let $p_0$ denote the threshold of the concatenated $7$-qubit Steane code (see Lemma \ref{lem:Strongthreshold}). For a sequence $(n_m)_{m\in\N}$ consider quantum operations $E_{m}:\M^{\otimes n_m}_2\ra\M^{\otimes m}_d$ and $D_{m}:\M^{\otimes m}_d\ra \M^{\otimes n_m}_2$ defining an asymptotically good code of rate $R>0$ and goodness $\alpha\in\lb 0,1\rb$. Moreover, consider $p,q\in\lbr 0,1\rbr$ such that
\begin{equation}
p<\min(p_0/2,c^{-1})\quad \text{ and }\quad x:=4jcp+q<\alpha,
\label{equ:alphaCond}
\end{equation}
where $c>0$ is the constant from Theorem \ref{thm:NiceInterfaceConcatenated}. For any $\delta>0$ with $x+\delta<\alpha$ there is an $m_0\in\N$ such that for any $m\geq m_0$ there exists an $(\epsilon'_m,n_m,m)$ fault-tolerant coding scheme for quantum communication under the noise model $\mathcal{F}_{\PIID(p)}$ via the quantum channel
\[
I_q = (1-q)\ident_d + q T,
\]
with any fixed quantum channel $T:\M_{d}\ra\M_{d}$. Here, we have 
\[
\epsilon'_m = \epsilon_m + 3\exp\lb -m\frac{\delta^2}{3x}\rb, 
\] 
for any sequence $\epsilon_m>0$ with $\lim_{m\ra \infty}\epsilon_m = 0$. In particular, the rate $R$ is fault-tolerantly achievable and we have
\[
Q_{\mathcal{F}_{\PIID(p)}}(I_q)\geq R .
\]
\label{thm:FTCDepol}
\end{thm} 

\begin{proof}
For each $l\in\N$ let $\mathcal{C}_l\subset (\C^2)^{\otimes 7^l}$ denote the $l$th level of the concatenated $7$-qubit Steane code with threshold $p_0\in\lb 0,1\rbr$ (see Lemma \ref{lem:Strongthreshold}). Moreover, we denote by $\Enc_l:\M_2\ra\M^{\otimes 7^l}_2$ and $\Dec_l:\M^{\otimes 7^l}_2\ra \M_2$ the interface circuits from \eqref{equ:EncDec1} and \eqref{equ:EncDec2}, and recall Theorem \ref{thm:NiceInterfaceConcatenated} introducing a constant $c>0$. For $0\leq p<\min(p_0/2,c^{-1})$ and $q\in\lbr 0,1\rbr$ we choose some $\delta>0$ such that $x+\delta <\alpha$, where $x<\alpha$ was defined in \eqref{equ:alphaCond}. Finally, note again that $d=2^j$ throughout the proof.

Consider any sequence $(\epsilon_m)_{m\in\N}$ such that $\epsilon_m>0$ for all $m\in\N$ and $\lim_{m\ra\infty}\epsilon_m=0$. To construct a fault-tolerant coding scheme, we consider quantum circuits $\Gamma^{E_m}:\M^{\otimes n_m}_2\ra\M^{\otimes m}_d$ and $\Gamma^{D_m}:\M^{\otimes m}_d\ra \M^{\otimes n_m}_2$ such that 
\begin{equation}\label{equ:ApproxCodes}
\max\lb \|\Gamma^{E_m} - E_m\|_{\diamond}, \|\Gamma^{D_m} - D_m\|_{\diamond}\rb\leq \epsilon_m .
\end{equation}
For each $m\in\N$ we choose the code $\mathcal{C}_{l_m}$ with $l_m\in\N$ large enough such that 
\begin{equation}\label{equ:LevelBound}
C\lb \frac{p}{p_0}\rb^{2^{l_m}}\lb |\Loc(\Gamma^{D_m})|+ |\Loc(\Gamma^{E_m})| + 2N(n_m,\epsilon_m) + m\rb \leq \exp\lb -m\frac{\delta^2}{3x}\rb,
\end{equation}
where $C>0$ is the constant from Theorem \ref{thm:EffCommChan}. Next, we define quantum circuits
\[
\Gamma^{E,m}\lb \mathcal{C}_{l_m}\rb = \lb\Dec^{\otimes j_1}_{l_m}\rb^{\otimes m}\circ \Gamma^{E_m}_{\mathcal{C}_{l_m}},
\]
and 
\[
\Gamma^{D,m}\lb \mathcal{C}_{l_m}\rb = \Gamma_{\mathcal{C}_{l_m}}^{D_m}\circ \lb\Enc^{\otimes j_2}_{l_m}\rb^{\otimes m},
\] 
for each $m\in\N$. Now, consider a pair of quantum circuits $\Gamma^1:\C^{2^{k_1}}\ra \M^{\otimes n_m}_2$ and $\Gamma^2:\M^{\otimes n_m}_2\ra \C^{2^{k_2}}$ for some $k_1,k_2\in\N$ such that 
\[
\max\lb |\Loc(\Gamma^1)|,|\Loc(\Gamma^2)|\rb\leq N(n_m,\epsilon_m).
\]
By Theorem \ref{thm:EffCommChan} and \eqref{equ:LevelBound} we have     
\begin{align}
\Big{\|} &\lbr\Gamma^2_{\mathcal{C}_{l_m}}\circ \Gamma^{D,m}\lb{\mathcal{C}_{l_m}}\rb\circ  I^{\otimes m}_q\circ \Gamma^{E,m}\lb{\mathcal{C}_{l_m}}\rb\circ \Gamma^1_{\mathcal{C}_{l_m}}\rbr_{\mathcal{F}_{\PIID(p)}}-\Gamma^2\circ \Gamma^{D_m}\circ I_{x,l_m}^{\otimes m}\circ (\Gamma^{E_m}\circ \Gamma^1\otimes \sigma_S)\Big{\|}_{1\ra 1} \nonumber\\
&\quad\quad\quad\quad\quad\quad\leq \exp\lb -m\frac{\delta^2}{3x}\rb,
\label{equ:DepolCapPrEqu1}
\end{align}
for every $m\in\N$, where the quantum channel $I_{x,l_m}:\M_d\otimes \M^{\otimes (7^{l_m}-1)}_2\ra\M_d$ is of the form 
\[
I_{x,l_m} = (1-x)\ident_d\otimes \Tr_S + x N_{l_m}
\]
for some quantum channel $N_{l_m}:\M_2\otimes \M^{\otimes (7^{l_m}-1)}_2\ra \M_2$ acting on a data qubit and the syndrome space, and $x<\alpha$ as in \eqref{equ:alphaCond}. For any $m\in\N$ and any syndrome state $\sigma_S\in \lb\M^{\otimes (7^{l_m}-1)}_2\rb^{\otimes m}$ we denote by $I_{x,l_m}^{\otimes m}\lb \cdot \otimes \sigma_S\rb:\M_d\ra \M_d$ the quantum channel given by 
\[
\M_d\ni X\mapsto I_{x,l_m}^{\otimes m}\lb X \otimes \sigma_S\rb\in \M_d, 
\]
where each $I_{x,l_m}$ acts partially on some part of the state $\sigma_S$. Now, we compute 
\begin{align}
I_{x,l_m}^{\otimes m}\lb \cdot \otimes \sigma_S\rb \nonumber &=\sum_{i_1,\ldots ,i_{m}\in\lset 0,1\rset} (1-x)^{m-\sum^{m}_{j=1}i_j}x^{\sum^{m}_{j=1}i_j}\lb\bigotimes_{j|i_j=0} \ident_2\rb\otimes \tilde{N}_{l_m}^{(i_1\ldots ,i_{m})} \nonumber\\
&=\sum^{\lfloor(x+\delta) m\rfloor}_{s=0} (1-x)^{m-s}x^s\sum_{i_1+i_2+\cdots i_{m}=s} \lb\bigotimes_{j|i_j=0} \ident_2\rb\otimes\tilde{N}_{l_m}^{(i_1\ldots ,i_{m})} \nonumber\\
&\quad\quad\quad\quad\quad\quad+ \text{P}\lb \frac{1}{m}\sum^{m}_{i=1} X_i >x+\delta\rb N,
\label{equ:DepolCapPrEqu2}
\end{align} 
where $\tilde{N}_{l_m}^{(i_1\ldots ,i_{m})}:\M^{\otimes \sum_{j}i_j}_2\ra \M^{\otimes \sum_{j}i_j}_2$ denotes a quantum channel acting on the tensor factors corresponding to the $j$ for which $i_j=1$ (constructed from tensor powers of $N_{l_m}$ acting partially on the state $\sigma_S$), and $N:\M^{\otimes m}_2\ra \M^{\otimes m}_2$ denotes a quantum channel collecting the second part of the sum. Furthermore, we introduced independent and identically distributed $\lset 0,1\rset$-valued random variables $X_i$ with $P(X_1=1)=x$. Finally, we can use that $\Gamma^{E_m}:\M^{\otimes n_m}_2\ra\M^{\otimes m}_2$ and $\Gamma^{D_m}:\M^{m}_2\ra\M^{\otimes n_m}_2$ approximate the coding operations of an asymptotically good code with goodness $\alpha$. Let $m_0\in\N$ be large enough such that $\lfloor(x+\delta) m\rfloor\leq \alpha m\leq t_m$ for any $m\geq m_0$, where $(t_m)_{m\in\N}$ denotes the sequence as in Definition \ref{defn:AsymptGoodCode} for the asymptotically good code given by $(E_m,D_m)$. Using first the triangle inequality with \eqref{equ:ApproxCodes} and then property 3.~from Definition \ref{defn:AsymptGoodCode} together with \eqref{equ:DepolCapPrEqu2} we find
\begin{align}
\|\ident^{m}_2 - \Gamma^{D_m}\circ I_{x,l_m}^{\otimes m}\circ (\Gamma^{E_m}\otimes \sigma_S)\|_\diamond &\leq \|\ident^{m}_2 - D_m\circ I_{x,l_m}^{\otimes m}\circ (E_m\otimes \sigma_S)\|_\diamond +2\epsilon_m\nonumber\\
&\leq 2\text{P}\lb \frac{1}{m}\sum^{m}_{i=1} X_i >x+\delta\rb + 2\epsilon_m \nonumber\\
&\leq 2\text{exp}\lb -m\frac{\delta^2}{3x}\rb + 2\epsilon_m,
\label{equ:DepolCapPrEqu3} 
\end{align}
where the final estimate is the Chernoff bound. Combining \eqref{equ:DepolCapPrEqu3} with \eqref{equ:DepolCapPrEqu1} using the triangle inequality, we find 
\begin{align*}
\| \Gamma^2\circ \Gamma^1 - \lbr\Gamma^2_{\mathcal{C}_{l_m}}\circ \Gamma^{D,m}\lb{\mathcal{C}_{l_m}}\rb\circ  I^{\otimes m}_q\circ \Gamma^{E,m}\lb{\mathcal{C}_{l_m}}\rb\circ \Gamma^1_{\mathcal{C}_{l_m}}\rbr_{\mathcal{F}_{\PIID(p)}}\|_{1\ra 1} \\
\leq 3\text{exp}\lb -m\frac{\delta^2}{3x}\rb + 2\epsilon_m .
\end{align*}
Therefore, we find that the codes $\mathcal{C}_{l_m}$ and the pairs $(\Gamma^{E,m}\lb\mathcal{C}_{l_m}\rb,\Gamma^{D,m}\lb\mathcal{C}_{l_m}\rb)$ define a sequence of $(n_m,m,\epsilon'_m)$ fault-tolerant coding schemes as in Definition \ref{defn:FTcodingSchemesQuantCap}. 
\end{proof}

Using the good codes constructed by Calderbank and Shor in~\cite{calderbank1996good} we obtain the following corollary:

\begin{cor}[Lower bound from good codes]
Let $p_0$ denote the threshold of the concatenated $7$-qubit Steane code (see Lemma \ref{lem:Strongthreshold}), $c>0$ the constant from Theorem \ref{thm:NiceInterfaceConcatenated}, and $\alpha_0$ the constant from \eqref{equ:alpha0}. For $p,q\in\lbr 0,1\rbr$ such that $p\leq \min(p_0/2,c^{-1})$ and $4cp+q\leq \alpha_0$ we have
\[
Q_{\mathcal{F}_{\PIID(p)}}\lb (1-q)\ident_2 + q T\rb\geq 1-2h_2\lb 8cp+2q\rb ,
\]
for any quantum channel $T:\M_2\ra\M_2$.

\end{cor}

\section{Conclusion and open problems}
\label{sec:Conclusion}

By combining techniques from fault-tolerant quantum computation and quantum Shannon theory, we have initiated the study of fault-tolerant quantum Shannon theory. We introduced fault-tolerant capacities for classical and quantum communication via quantum channels, and for classical communication via classical-quantum channels. These capacities take into account that the encoding and decoding operations in the usual definitions of capacities are inherently affected by noise. We proved threshold-type theorems for the fault-tolerant capacities showing that rates $\epsilon$-close to the usual capacities can be obtain for non-vanishing gate error probabilities below some threshold value depending on $\epsilon$ and the communication channel. In the case of classical-quantum channels $T:\mathcal{A}\ra\M_d$ the threshold only depends on $\epsilon$ and the output dimension $d$. We leave open the question whether such ``uniform'' threshold theorems also hold for the classical and quantum capacity of a quantum channel.

Although we have focused on capacities and optimal achievable communication rates our results also apply for specific codes. As an example we considered fault-tolerant quantum communication schemes based on asymptotically good codes and via quantum channels of a specific form. Similar to the threshold theorem from~\cite{aliferis2006quantum}, protecting a specific coding scheme against Pauli i.i.d.~noise (of strength below the threshold) requires only a polylogarithmic overhead in the size of the quantum circuit implementing it. It will then yield a fault-tolerant coding scheme if the ideal coding scheme corrects the errors introduced by the effective quantum channel induced by the interfaces (cf.~Theorem \ref{thm:EffCommChan}).    

In future research it would be interesting to extend our results to other communication scenarios, such as private communication, quantum communication assisted by classical communication, or communication scenarios between multiple parties. Finally, it would also be interesting to study the effects of different circuit noise models on the corresponding fault-tolerant capacities.

\section{Acknowledgements}

We thank Gorjan Alagic and Hector Bombin for their contributions during early stages of this project. Moreover, we thank Paula Belzig, Omar Fawzi and Milan Mosonyi for helpful comments on an earlier version of this article. MC acknowledges financial support from the European Research Council (ERC Grant Agreement No. 81876), VILLUM FONDEN via the QMATH Centre of Excellence (Grant No.10059) and the  QuantERA ERA-NET Cofund in Quantum Technologies implemented within the European Union's Horizon 2020 Programme (QuantAlgo project) via the Innovation Fund Denmark. AMH acknowledges funding from the European Union's Horizon 2020 research and innovation programme under the Marie Sk\l odowska-Curie Action TIPTOP (grant no. 843414).

\appendix

\section{Concatenated quantum error correcting codes}
\label{sec:Appendix}

In this appendix we review basic facts about concatenated quantum codes~\cite{knill1996concatenated}.

\subsection{7-qubit Steane code}
\label{sec:7qubitSteane}

Let $V:\C^2\ra(\C^2)^{\otimes 7}$ denote the encoding isometry for the $7$-qubit Steane code~\cite{steane1996multiple} given by 
\begin{align*}
\ket{\overline{0}}=V\ket{0} = \frac{1}{8}( \ket{0000000}&+\ket{0001111}+ \ket{0110011}+ \ket{1010101}\\
&+ \ket{0111100}+ \ket{1011010}+ \ket{1100110}+ \ket{1101001}) 
\end{align*}
and 
\begin{align*}
\ket{\overline{1}}=V\ket{1} = \frac{1}{8}( \ket{1111111}&+\ket{1110000}+ \ket{1001100}+ \ket{0101010}\\
&+ \ket{1000011}+ \ket{0100101}+ \ket{0011001}+ \ket{0010110}).
\end{align*}
This code is a stabilizer code with the following stabilizer generators:
\begin{align*}
g_1 &= \id\otimes \id\otimes \id\otimes X\otimes X\otimes X\otimes X \\
g_2 &= \id\otimes X\otimes X\otimes \id\otimes \id\otimes X\otimes X \\
g_3 &= X\otimes \id\otimes X\otimes \id\otimes X\otimes \id\otimes X \\
g_4 &= \id\otimes \id\otimes \id\otimes Z\otimes Z\otimes Z\otimes Z \\
g_5 &= \id\otimes Z\otimes Z\otimes \id\otimes \id\otimes Z\otimes Z \\
g_6 &= Z\otimes \id\otimes Z\otimes \id\otimes Z\otimes \id\otimes Z .
\end{align*}
As a stabilizer code, the codewords $\ket{\overline{0}}$ and $\ket{\overline{1}}$ arise as the common eigenvectors for the eigenvalue $+1$ of the commuting set of Hermitian involutions $g_1,g_2,\ldots ,g_6$. As explained in Section \ref{sec:AnalyzingNoisyQuantumCircuits}, we can define a subspace $W_s\subset (\C^{2})^{\otimes 7}$ for each syndrome $s\in\F^6_2$ as the space of common eigenvectors for the eigenvalues $(-1)^{s_i}$ with respect to each $g_i$. Then, we find that 
\begin{equation}
\lb\C^2\rb^{\otimes 7} = \bigoplus_{s\in\F^6_2} W_s ,
\label{equ:compl7Qubit}
\end{equation}
and we can introduce the error basis
\[
\bigcup_{s\in\F^6_2} \lset E_s\ket{\overline{0}}, E_s\ket{\overline{1}}\rset
\]
of $\lb\C^2\rb^{\otimes 7}$ with Pauli operators $E_s$ associated to the syndrome $s\in\F^6_2$ and such that
\begin{equation}
W_s = \text{span}\lset E_s\ket{\overline{0}},E_s\ket{\overline{1}}\rset. 
\label{equ:basisWs}
\end{equation}
Finally, we define the unitary map $D:(\C^2)^{\otimes 7}\ra \C^2\otimes (\C^2)^{\otimes 6}$ by
\[
D\lb E_s\ket{\overline{i}}\rb = \ket{i}\otimes \ket{s}, 
\]
extended linearily giving rise to the ideal decoder $\Dec^*:\M^{\otimes 7}_{2}\ra\M_2\otimes \M^{\otimes 6}_2$ via 
\[
\Dec^* = \text{Ad}_{D},
\]  
and its inverse, the ideal encoder, $\Enc^*:\M_{2}\otimes \M^{\otimes 6}_2\ra\M^{\otimes 7}_2$ via
\[
\Enc^* = \text{Ad}_{D^\dagger}.
\]
For more details on quantum error correcting codes and the stabilizer formalism see~\cite{gottesman1997stabilizer}.

\subsection{Code concatenation}
\label{sec:CodeCon}

The $7$-qubit Steane code can be used to construct a concatenated code~\cite{knill1996concatenated} achieving higher protection against noise. To define this concatenated code we recursively define the encoding isometry $V^{(l)}:\C^2\ra\C^{7^l}$ encoding a single (logical) qubit into $7^l$ (physical) qubits via
\[
V^{(1)} = V \quad\text{ and }\quad V^{(l)} = V^{\otimes 7^{l-1}}\circ V^{(l-1)} \text{ for any }l\geq 2.
\] 
Note that in this way we have 
\[
V^{(l)} = V^{\otimes 7^{l-1}}\circ V^{\otimes 7^{l-2}}\circ \cdots\circ V^{\otimes 7}\circ V,
\]
and regrouping of this equation leads to the identity
\begin{equation}
V^{\otimes 7^{l-1}}\circ V^{(l-1)} = (V^{(l-1)})^{\otimes 7}\circ V.
\label{equ:Videntity}
\end{equation}
Recall that the syndrome space of the $7$-qubit Steane code consist of $6$ qubits. By the previous discussion, we see that the syndrome space of the $k$th level of the concatenated $7$-qubit Steane code consists of 
\[
6\sum^{l-1}_{i=0} 7^i = 7^{l}-1
\]
qubits. This is not surprising since we encode a single qubit into $7^l$ qubits using a code satisfying \eqref{equ:compl7Qubit} and \eqref{equ:basisWs}. Next, we recursively define the ideal operations 
\[
\Enc^*_{l}:\M_2\otimes \M^{\otimes (7^{l}-1)}_2\ra \M^{\otimes 7^l}_2 \quad\text{ and }\quad \Dec^{*}_{l}:\M^{\otimes 7^l}_2\ra\M_2\otimes \M^{\otimes (7^{l}-1)}_2,
\]
for every level $l\in\N$ by 
\begin{equation}
\Enc^{*}_{1}=\Enc^*, \quad\text{ and }\quad \Enc^{*}_{l} = \lbr\Enc^*_{1}\rbr^{\otimes 7^{l-1}}\circ \lbr \Enc^*_{l-1}\otimes \ident^{\otimes (7^{l-1}6)}_2\rbr \text{ for any }l\geq 2,
\label{equ:defnEncl}
\end{equation}
and 
\begin{equation}
\Dec^{*}_{1}=\Dec^*, \quad\text{ and }\quad \Dec^{*}_{l} = \lbr\Dec^{*}_{1} \otimes \ident^{\otimes 7(7^{l-1}-1)}_2\rbr\circ \lbr\Dec^{*}_{l-1}\rbr^{\otimes 7} \text{ for any }l\geq 2,
\label{equ:defnDecl}
\end{equation}
where we reordered the tensor factors such that ideal operations executed later in the circuit do not act on the syndrome output of operations executed earlier in the circuit. Again we can expand the previous recursions and verify by regrouping that 
\begin{equation}
\lbr\Enc^*_{1}\rbr^{\otimes 7^{l-1}}\circ \lbr \Enc^*_{l-1}\otimes \ident^{\otimes 7^{l-1}6}_2\rbr = \lbr\Enc^{*}_{l-1}\rbr^{\otimes 7}\circ \lbr\Enc^{*}_{1} \otimes \ident^{\otimes 7(7^{l-1}-1)}_2\rbr ,
\label{equ:EncRelation}
\end{equation}
and
\begin{equation}
\lbr\Dec^{*}_{1} \otimes \ident^{\otimes 7(7^{l-1}-1)}_2\rbr\circ \lbr\Dec^{*}_{l-1}\rbr^{\otimes 7} = \lbr \Dec^*_{l-1}\otimes \ident^{\otimes 7^{l-1}6}_2\rbr \circ \lbr\Dec^*_{1}\rbr^{\otimes 7^{l-1}},
\label{equ:DecRelation}
\end{equation}
where we again ordered the tensor factors appropriately. 

\subsection{Explicit interface from level $0$ to level $1$}
\label{Sec:InterfaceExpl}

In this section, we will explicitly construct an interface for the first level of the $7$-qubit Steane code using a simple teleportation circuit. We should emphasize that this construction is certainly well-known and it works also for more general quantum error correcting codes. We state it here for convenience, so that the constructions in Section \ref{sec:NoisyInterface} can be made explicit. 

Let $\omega_{(0,1)}\in\lb\M_2\otimes \M^{\otimes 7}_2\rb^+$ denote a maximally entangled state between a physical qubit and the $7$-qubit Steane code (i.e., the first level of the concatenated code). Specifically, we define $\omega_{(0,1)} := \proj{\Omega_{(0,1)}}{\Omega_{(0,1)}}$ for
\begin{equation}
\ket{\Omega_{(0,1)}} = \frac{1}{\sqrt{2}}\lb\ket{0}\otimes \ket{\overline{0}_1} + \ket{1}\otimes \ket{\overline{1}_1} \rb,
\label{equ:Omega}
\end{equation}
where $\lset\ket{\overline{i}_1}\rset^1_{i=0}$ denotes the computational basis in the $7$-qubit Steane code. Note that $\omega_{(0,1)}$ can be prepared using elementary gates (Hadamard and CNOT gates) only, and we will denote this preparation circuit by $\text{Prep}\lb\omega_{(0,1)}\rb$. 
\begin{figure*}[t!]
        \center
        \includegraphics[scale=0.7]{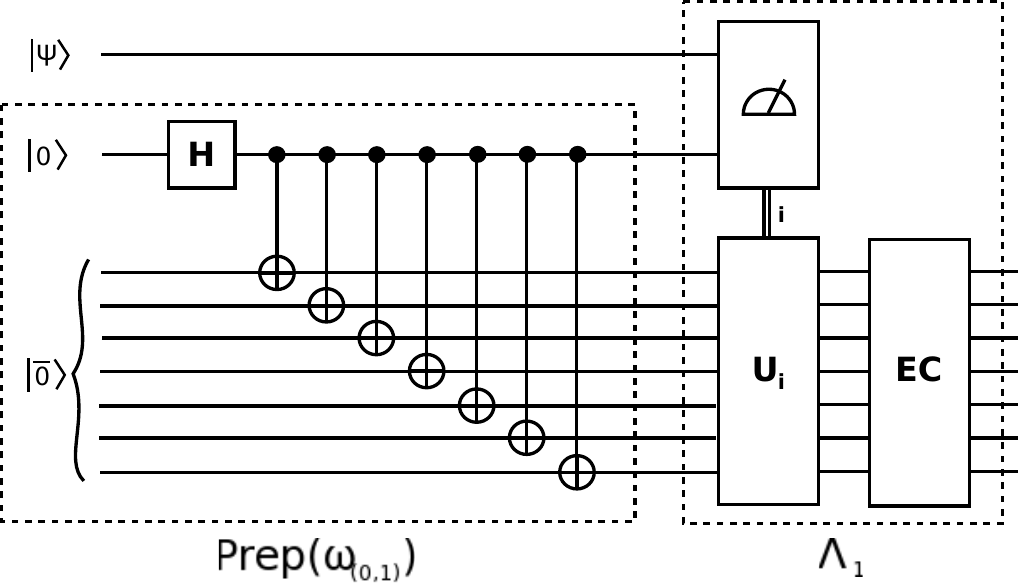}
        \caption{The circuit $\Enc_{0\ra 1}$ encoding the unknown state $\ket{\psi}$ into the first level of the $7$-qubit Steane code. Note that the logical CNOT gate is implemented in this code by applying elementary CNOT gates to each physical qubit. Here, $U_i$ denotes a certain Pauli gate (implemented in the $7$-qubit Steane code) depending on the outcome $i\in\lset 0,1,2,3\rset$ of the Bell measurement, and $\EC$ denotes the error correction of the code.}
        \label{fig:circuit1}
\end{figure*}
Now, we define
\begin{equation}
\Enc_{0\ra 1}\lb \cdot\rb := \Lambda_1\lb \cdot\otimes \text{Prep}\lb\omega_{(0,1)}\rb \rb.
\label{equ:Enc01}
\end{equation}
Here, $\Lambda_1:\M_2 \otimes \M_2 \otimes \M^{\otimes 7}_2\ra \M^{\otimes 7}_2$ denotes the teleportation protocol, i.e., measuring the two first registers in the Bell basis
\begin{align*}
\ket{\phi_1} &= \ket{\Omega_2}, \\
\ket{\phi_2} &= (\one_2\otimes \sigma_z)\ket{\Omega_2},\\
\ket{\phi_3} &= (\one_2\otimes \sigma_x)\ket{\Omega_2}, \\
\ket{\phi_4} &= (\one_2\otimes \sigma_x\sigma_z)\ket{\Omega_2},
\end{align*} 
and then depending on the measurement outcome $i\in\lset 1,2,3,4\rset$ performing the $1$-rectangle corresponding to the unitary gate $U_i\in\mathcal{U}_2$ on the $7$-qubit Steane code space, where 
\begin{align*}
U_1 = \one_2, \quad U_2 = \sigma_z, \quad U_3 = \sigma_x, \quad U_4 = \sigma_x\sigma_z .
\end{align*}
Again, we note that $\Lambda_1$ is a quantum circuit ending in an error correction. See Figure \ref{fig:circuit1} for a circuit diagram of the quantum circuit $\Enc_{0\ra 1}$.

Similarily, we denote by $\Lambda_2:\M_2 \otimes \M^{\otimes 7}_2 \otimes \M^{\otimes 7}_2\ra \M_2$ the teleportation protocol measuring the two final registers in the Bell basis (implemented in the $7$-qubit Steane code) and performing the unitary gate stated above depending on the measurement outcome on the remaining qubit system. Then, we define
\begin{equation}
\Dec_{1\ra 0}\lb \cdot\rb := \Lambda_2\lb \text{Prep}\lb\omega_{(0,1)}\rb \otimes \cdot\rb.
\label{equ:Dec10}
\end{equation}
Again, this is a quantum circuit, and the circuit diagram can be seen in Figure \ref{fig:circuit2}.

\begin{figure*}[t!]
        \center
        \includegraphics[scale=0.7]{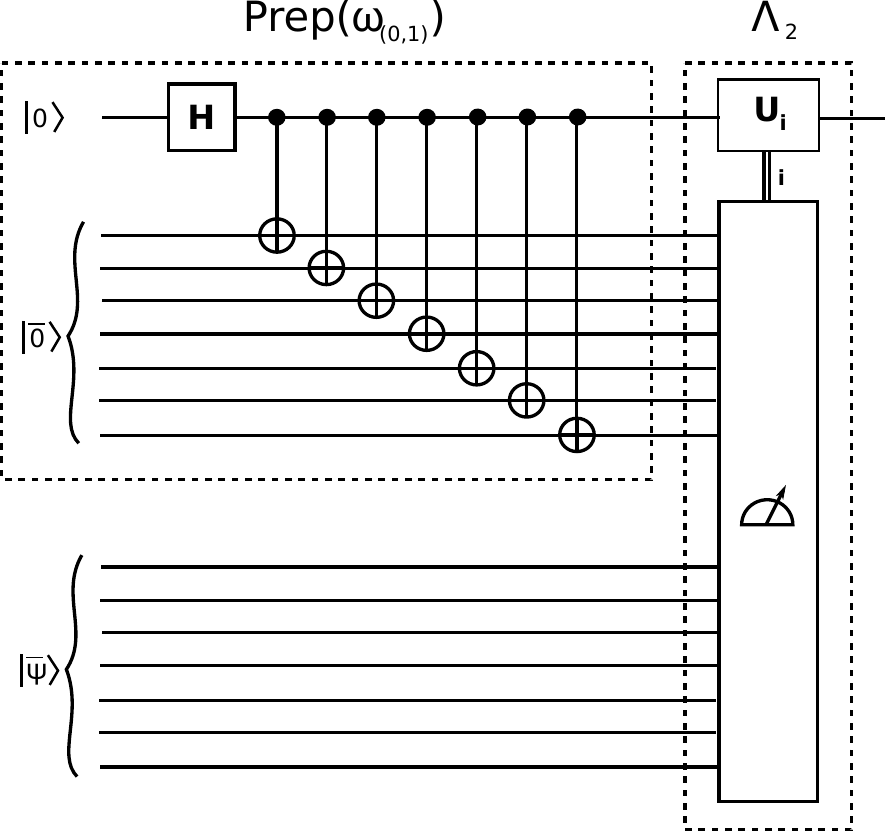}
        \caption{The circuit $\Dec_{0\ra 1}$ decoding the unknown state $\ket{\overline{\psi}}$ from the first level of the $7$-qubit Steane code to a physical qubit. Here, $U_i$ denotes a certain Pauli gate (applied to the physical qubit) depending on the outcome $i\in\lset 0,1,2,3\rset$ of the Bell measurement.}
        \label{fig:circuit2}
\end{figure*}

\section{Chasing constants in weak typicality}

To prove our main results, we need precise error bounds in the direct parts of the classical and quantum capacity theorems (cf.~Theorem \ref{thm:HSWQChannel} and Theorem \ref{thm:LSD}). Such bounds can be obtained by chasing the constants appearing in the proofs of these results, and by identifying the precise dependence of the final error on the number of channel uses. Most proofs of Theorem \ref{thm:HSWQChannel} and Theorem \ref{thm:LSD} in the literature (see e.g., \cite{devetak2005private,hayden2008decoupling,wilde2017quantum,watrous2018theory}) are based on different notions of typicality. In this appendix we summarize the neccessary results and obtain explicit error bounds needed to chase the constants in the capacity theorems. We should emphasize that the following results are well-known, and we merely extracted them from the literature (in particular from~\cite{wilde2017quantum}) making some bounds explicit.  

For each $n\in\N$ let $X^n=(X_1,\ldots ,X_n)$ be an $n$-tuple of i.i.d.~random variables $X_i\sim X$ with values in a finite set $\mathcal{A}$. We denote by $\text{supp}(X)=\lset x\in\mathcal{A}~:~p_X(x)\neq 0\rset$ and the Shannon entropy by
\[
H(X) = -\sum_{x\in \mathcal{A}} p_X(x)\log(p_X(x)).
\] 
The set of $\delta$-typical sequences of length $n\in\N$ is given by
\[
T^{X^n}_\delta := \lset x^n\in\mathcal{A}^n ~|~ \Big{|}-\frac{1}{n}\log\lb p_{X^n}(x^n)\rb - H(X)\Big{|}\leq \delta \rset,
\]
where $p_{X^n}(x^n)=\Pi^n_{i=1}p_X(x_i)$. We first note the following simple property that follows immediately from the definition:

\begin{thm}
For any $\delta>0$ we have
\[
p_{X^n}(x^n)\leq 2^{-n(H(X)-\delta)},
\]
for any $x^n\in T^{X^n}_\delta$.
\label{thm:equipartition}
\end{thm}

The following theorem shows that with high probability the random variable $X^n$ takes values in the typical set $T^{X^n}_\delta$.

\begin{thm}
For $p_{\min} = \min_{x\in\text{supp}\lb X\rb}p_X\lb x\rb$ and any $\delta >0$ we have 
\[
P\lb X^n\in T^{X^n}_{\delta}\rb \geq 1 - \exp\lb \frac{-2n\delta^2}{\log(p_{\min})^2}\rb.
\]
\label{thm:weaktypConc}
\end{thm}
\begin{proof}
For each $x\in\mathcal{A}$ we define an i.i.d.~sequence of indicator random variables as 
\[
I_x(X_i) = \begin{cases} 1 , \text{ if }X_i=x \\
0, \text{ else. }\end{cases} 
\] 
Now consider the $n$-tuples of i.i.d.~random variables $Z^n=(Z_1,\ldots ,Z_n)$ defined as
\[
Z_i = -\sum_{x\in\text{supp}\lb X\rb} I_x\lb X_i\rb\log(p_X(x)).
\]
Note that $\mathbbm{E}(Z_i) = H(X)$ and $Z_i\in\lbr 0, -\log\lb p_{\min}\rb\rbr$ almost surely. By Hoeffding's inequality we have 
\[
P\lb \left|\frac{1}{n}\sum^n_{i=1} Z_i - H(X)\right| > \delta\rb \leq \exp\lb\frac{-2n\delta^2}{\log(p_{\min})^2}\rb.
\]
Now note that $x^n\in \mathcal{A}^n$ with $p_{X^n}(x^n)>0$ satisfies $x^n\in T^{X^n}_\delta$ if and only if
\begin{align*}
-\frac{1}{n}\log\lb p_{X^n}\lb x^n\rb\rb &= -\frac{1}{n}\sum^n_{i=1}\log\lb p_X(x_i)\rb \\
&= -\frac{1}{n}\sum^n_{i=1}\sum_{x\in\text{supp}(X)} I_x(x_i)\log\lb p_X(x)\rb\in \lbr H(X)-\delta , H(X)+\delta\rbr .
\end{align*}
Therefore, we have
\begin{align*}
P\lb X^n\in T^{X^n}_{\delta}\rb &= 1-P\lb X^n\notin T^{X^n}_{\delta}\rb \\ 
&= 1-P\lb -\frac{1}{n}\sum^n_{i=1}\sum_{x\in\text{supp}(X)} I_x(X_i)\log\lb p_X(x)\rb\notin \lbr H(X)-\delta , H(X)+\delta\rbr\rb \\
&=1-P\lb \left|\frac{1}{n}\sum^n_{i=1} Z_i - H(X)\right|>\delta\rb \\
& \geq  1-\exp\lb\frac{-2n\delta^2}{\log(p_{\min})^2}\rb.
\end{align*}
\end{proof}
For each $n\in\N$ let $(X,Y)^n=\lb (X_1,Y_1),\ldots ,(X_n,Y_n)\rb$ be an $n$-tuple of i.i.d.~pairs of random variables $(X_i,Y_i)\sim (X,Y)$ with values in the finite product set $\mathcal{A}\times \mathcal{B}$. We define $\text{supp}\lb X,Y\rb=\lset (x,y)\in\mathcal{A}\times \mathcal{B}~:~p_{X,Y}(x,y)\neq 0\rset$ and the joint Shannon entropy by
\[
H(X,Y) = -\sum_{(x,y)\in \mathcal{A}\times \mathcal{B}} p_{X,Y}(x,y)\log(p_{(X,Y)}(x,y)).
\] 
The set of conditional $\delta$-typical sequences of length $n\in\N$ conditioned onto a sequence $x^n\in\mathcal{A}^n$ satisfying $p_{X^n}\lb x^n\rb>0$ is given by
\[
T^{Y^n|x^n}_\delta := \lset y^n\in\mathcal{B}^n ~|~ \Big{|}-\frac{1}{n}\log\lb\frac{ p_{X^n,Y^n}(x^n,y^n)}{p_{X^n}\lb x^n\rb}\rb - H(X,Y) + H(X)\Big{|}\leq \delta \rset.
\]
Here, $p_X$ denotes the marginal probability distribution $p_X(x) = \sum_{y\in \mathcal{B}}p_{X,Y}\lb x,y\rb$.

\begin{thm}
For $r_{\min} = \min_{(x,y)\in\text{supp}\lb X,Y\rb}\lb \frac{p_{X,Y}\lb x,y\rb}{p_X(x)}\rb$ and any $\delta >0$ we have
\[
\mathbbm{E}_{X^n} P_{Y^n|X^n}\lb Y^n\in T^{Y^n|X^n}_{\delta} \rb \geq 1 - \exp\lb\frac{-2n\delta^2}{\log(r_{\min})^2}\rb.
\]
\label{thm:weakcondTypConc}
\end{thm}

\begin{proof}
For each pair $(x,y)\in\mathcal{A}\times \mathcal{B}$ define the sequence of indicator random variables
\[
I_{x,y}\lb X_i,Y_i\rb = \begin{cases} 1, \text{ if }(X_i,Y_i) = (x,y) \\ 0, \text{ else.}\end{cases}
\]
Now define $n$-tuples of i.i.d.~random variables $Z^n=(Z_1,\ldots ,Z_n)$ given by 
\[
Z_i = -\sum_{(x,y)\in\text{supp}\lb X,Y\rb} I_{(x,y)}\lb X_i, Y_i\rb \lbr\log\lb p_{X,Y}\lb x,y\rb\rb - \log(p_X\lb x\rb)\rbr.
\]
Note that $\mathbb{E}\lb Z_i\rb = H(X,Y) - H(X)$ and $Z_i\in \lbr 0, -\log(r_{\min})\rbr$ almost surely. By Hoeffding's inequality we have
\[
P\lb \Big{|}\frac{1}{n}\sum^n_{i=1} Z_i - H(X,Y) + H(X)\Big{|} >\delta\rb\leq \exp\lb\frac{-2n\delta^2}{\log(r_{\min})^2}\rb.
\]
Note that for a given $n$-tuple $x^n\in\mathcal{A}^n$ satisfying $p_{X^n}(x^n)>0$ and $y^n\in\mathcal{B}^n$ satisfying $p_{X^n,Y^n}(x^n,y^n)>0$ we have $y^n\in T^{Y^n|x^n}_{\delta}$ if and only if
\begin{align*}
-\frac{1}{n}\log\lb \frac{p_{X^n,Y^n}\lb x^n,y^n\rb}{p_{X^n}\lb x^n\rb}\rb &= -\frac{1}{n}\sum^n_{i=1}\lbr \log(p_{X,Y}(x_i,y_i)) - \log(p_X(x_i))\rbr \\
&= -\frac{1}{n}\sum^n_{i=1}\sum_{(x,y)\in\text{supp}\lb X,Y\rb}I_{(x,y)}(x_i,y_i)\lbr\log(p_{X,Y}(x,y)) - \log(p_X(x))\rbr\\
&\in \lbr H(X,Y) -H(X) -\delta, H(X,Y) - H(X) +\delta\rbr =: I_\delta.
\end{align*}
Using the law of total probability, we obtain
\begin{align*}
&\mathbbm{E}_{X^n} P_{Y^n|X^n}\lb Y^n\in T^{Y^n|X^n}_{\delta}\rb = \sum_{x^n\in\text{supp}(X^n)} P_{X^n}(x^n)P_{Y^n|X^n}\lb Y^n\in T^{Y^n|X^n}_{\delta}\rb \\
&=P\lb Y^n\in T^{Y^n|X^n}_{\delta}\rb \\
&=1- P\lb Y^n\notin T^{Y^n|X^n}_{\delta}\rb \\
& = 1 - P\lb -\frac{1}{n}\sum^n_{i=1}\sum_{(x,y)\in\text{supp}\lb X,Y\rb}I_{(x,y)}(X_i,Y_i)\lbr\log(p_{X,Y}(x,y)) - \log(p_X(x))\rbr\notin I_\delta\rb\\
& = 1 - P\lb \Big{|}\frac{1}{n}\sum^n_{i=1} Z_i - H(X,Y) + H(X)\Big{|} >\delta\rb \geq 1-\exp\lb\frac{-2n\delta^2}{\log(r_{\min})^2}\rb. 
\end{align*}
\end{proof}

\begin{thm}[Size bound for conditional typical subset]\hfill\\
For every $n$-tuple $x^n\in\mathcal{A}^n$ and any $\delta >0$ we have
\[
\Big{|}T^{Y|x^n}_\delta\Big{|}\leq 2^{n\lb H(X,Y) - H(X) + \delta\rb}
\]
\label{thm:SizeCondTyp}
\end{thm}

\begin{proof}
Without loss of generality we can assume that $x_i\in\text{supp}(X)$ since otherwise $T^{Y|x^n}_\delta=\emptyset$. We know that for each $y^n\in T^{Y|x^n}_\delta$ we have 
\[
p_{Y^n|x^n}\lb y^n\rb = \prod^n_{i=1} p_{Y|x_i}\lb y_i\rb = \prod^n_{i=1} \frac{p_{X,Y}\lb x_i,y_i\rb}{p_X(x_i)} \geq 2^{-n\lb H(X,Y) - H(X) + \delta\rb}.
\]
Therefore, we have
\[
1\geq \sum_{y_n\in T^{Y|x^n}_\delta} p_{Y^n|x^n}\lb y^n\rb \geq \Big{|}T^{Y|x^n}_\delta\Big{|}2^{-n\lb H(X,Y) - H(X) + \delta\rb}.
\]
This finishes the proof.
\end{proof}

Now, we will use the previous theorems from this section to introduce quantum typicality. Let $\rho=\sum_{x\in\mathcal{A}}p_X(x)\proj{x}{x}$ denote a quantum state, where we introduced a random variable $X$ with values in $\mathcal{A}$ distributed according to the spectrum of $\rho$. Note that $H(X)=S(\rho)$, i.e., the von-Neumann entropy of the quantum state $\rho$. For $\delta>0$ we define the $\delta$-typical projector with respect to $\rho^{\otimes n}$ as
\[
\Pi^n_\delta = \sum_{x^n\in T^{X^n}_\delta} \proj{x^n}{x^n}.
\]
The following theorem follows easily from Theorem \ref{thm:equipartition}
\begin{thm}
For any $\delta>0$ and $n\in\N$ we have 
\[
\Pi^n_\delta\rho^{\otimes n}\Pi^n_\delta\leq 2^{-n(S(\rho)-\delta)}\Pi^n_\delta.
\]
\label{thm:QTypical1}
\end{thm}

From Theorem \ref{thm:weaktypConc} we easily get the following theorem:

\begin{thm}[Quantum typicality]
For $\lambda^*_{\min}(\rho):=\min\lset \lambda\in\text{spec}(\rho)\setminus\lset 0\rset\rset$ and any $\delta>0$ we have 
\[
\Tr\lb \Pi^n_\delta\rho^{\otimes n}\rb \geq 1-\exp\lb \frac{-2n\delta^2}{\log(\lambda^*_{\min}(\rho))^2}\rb .
\]
\label{thm:QTypical2}
\end{thm}

Now let $\lset p_X(x),\rho_x\rset_{x\in\mathcal{A}}$ denote an ensemble of quantum states, and for each $x\in\mathcal{A}$ we have the eigendecomposition $\rho_x=\sum_{y\in\mathcal{B}}p_{Y|X}(y|x)\proj{y_x}{y_x}$ defining the random variable $Y$. For $\delta>0$ and an $n$-tuple $x^n\in\mathcal{A}^n$ we define the conditional $\delta$-typical projector of the ensemble $\lset p_X(x),\rho_x\rset_{x\in\mathcal{A}}$ conditioned on $x^n$ by 
\[
\Pi^{B^n|x^n}_\delta= \sum_{y^n_{x^n}\in T^{Y^n|x^n}_\delta} \proj{y^n_{x^n}}{y^n_{x^n}}.
\]
Note that 
\[
\lbr\Pi^{B^n|x^n}_\delta,\rho_{x^n}\rbr = 0,
\]
and by Theorem \ref{thm:weakcondTypConc} we have the following result:

\begin{thm}[Conditional quantum typicality]
For $\mu^*_{\min} = \min_{x\in \text{supp}\lb X\rb}\min\lset \lambda\in\text{spec}\lb \rho_x\rb\setminus\lset 0\rset\rset$ and any $\delta >0$ we have
\[
\sum_{x^n\in \mathcal{A}^n} p_{X^n}(x^n)\Tr\lb \Pi^{B^n|x^n}_\delta \rho_{x^n}\rb \geq 1 - \exp\lb\frac{-2n\delta^2}{\log(\mu^*_{\min})^2}\rb.
\]
\label{thm:QTypical3}
\end{thm}

Finally, by Theorem \ref{thm:SizeCondTyp} we have the following theorem:

\begin{thm}
For any $\delta>0$ we have 
\[
\Tr\lb \Pi^{B^n|x^n}_\delta\rb\leq 2^{n\lb H(X,Y) - H(X) + \delta\rb}.
\]
\label{thm:SizeQTyp}
\end{thm}

\section{Explicit error bound in the HSW-theorem}
\label{app:HSWError}

For the proof of Theorem~\ref{thm:FTCCQCLowerBound} we need explicit error bounds in the direct part of the proof of the HSW-theorem (cf.~Theorem~\ref{thm:HSWQChannel}). To make our article selfcontained we derive these bounds in this appendix. We will start with the following lemma, which is a version of the well-known packing lemma (see for instance~\cite[Lemma 16.3.1]{wilde2017quantum}). Its proof combines the general strategy used in~\cite[Lemma 16.3.1]{wilde2017quantum} with some insights from \cite[Section 8.1.2]{watrous2018theory} leading to a slightly better error estimate.

\begin{lem}[Packing lemma]
Let $\lset p_i,\sigma_i\rset^L_{i=1}$ be an ensemble of quantum states $\sigma_i\in \M^+_d$ with ensemble average $\sigma=\sum^L_{i=1}p_i\sigma_i$. Let $\Pi:\C^d\ra \C^d$ and $\Pi_i :\C^d\ra\C^d$ for each $i\in\lset 1,\ldots ,L\rset$ denote projectors such that the following conditions hold for $\epsilon_1,\epsilon_2>0$ and $A,B\in\N$:
\begin{enumerate}
\item $\Tr\lbr\Pi \sigma \rbr \geq 1-\epsilon_1$. 
\item $\sum^L_{i=1} p_i\Tr\lbr\Pi_i \sigma_i \rbr \geq 1-\epsilon_2$.
\item $\lbr \Pi_i,\sigma_i \rbr = 0$ for any $i\in\lset 1,\ldots, L\rset$. 
\item $\Tr\lbr\Pi_i\rbr\leq A$ for any $i\in\lset 1,\ldots, L\rset$.
\item $\Pi \sigma\Pi \leq \frac{1}{B}\Pi$.
\end{enumerate}   
Then, for any $M\leq L$ there exists an $M$-tuple $I=\lb i_1 , \ldots ,i_M\rb\in \lset 1,\ldots ,L\rset^M$, and a POVM $\lb\Lambda_s\rb^M_{s=1}\in (\M^+_d)^M$ such that 
\[
\frac{1}{M}\sum^M_{s=1}\Tr\lb \Lambda_s \sigma_{i_s}\rb\geq 1 - 4\epsilon_1 - 2\epsilon_2 - 4 M\frac{A}{B}.
\] 
\end{lem}
\begin{proof}
Let $M\leq L$ be fixed in the following. For each $i\in\lset 1,\ldots ,L\rset$ define the operators
\[
Y_{i} = \Pi\Pi_i\Pi .
\]
For any $M$-tuple $I = \lb i_1,\ldots ,i_M\rb\in \lset 1,\ldots ,L\rset^M$ we define a POVM $\lb \Lambda_s\rb^M_{s=1}$ by 
\[
\Lambda_s = \lb\sum^M_{s'=1} Y_{i_{s'}}\rb^{-\frac{1}{2}} Y_{i_s}\lb\sum^M_{s'=1} Y_{i_{s'}}\rb^{-\frac{1}{2}}, \text{ for }s\in\lset 1,\ldots ,M\rset, 
\]
and the decoding error of symbol $s\in\lset 1,\ldots ,M\rset$ by
\[
p_e(s,I) = \Tr\lb\lb\one_d - \Lambda_s\rb \sigma_{i_s}\rb.
\]
Recall the Hayashi-Nagaoka inequality (see~\cite[Lemma 8.28]{watrous2018theory})
\[
\one_d - \lb P+Q\rb^{-\frac{1}{2}}P\lb P+Q\rb^{-\frac{1}{2}}\leq 2\lb \one_d - P\rb+4Q,
\]
for any pair of positive matrices $P,Q\in\M^+_d$ satisfying $P\leq \one_d$. Applying this inequality for $P=Y_{i_s}$ and $Q=\sum_{s'\neq s}Y_{i_s}$ gives the estimate
\[
p_e(s,I) \leq 2\lbr 1- \Tr\lb Y_{i_s} \sigma_{i_s}\rb\rbr + 4\sum_{s'\neq s}\Tr\lb Y_{i_{s'}} \sigma_{i_s} \rb.
\]
Following~\cite{watrous2018theory} we apply the operator equality 
\[
ABA = AB + BA - B +\lb \one_d -A\rb B \lb\one_d - A\rb,
\]
and using assumption 3.~from above we obtain
\begin{align}
p_e&(s,I) \nonumber\\
&= 2\lbr 1- 2\Tr\lb \Pi\Pi_{i_s} \sigma_{i_s}\rb + \Tr\lb \Pi_{i_s} \sigma_{i_s}\rb - \Tr\lb (\one_d-\Pi)\Pi_{i_s}(\one_d-\Pi)\sigma_{i_s}\rb \rbr + 4\sum_{s'\neq s}\Tr\lb Y_{i_{s'}}\sigma_{i_s}\rb\nonumber \\
&\leq 2\lbr 1- \Tr\lb (2\Pi - \one_d)\Pi_{i_s} \sigma_{i_s}\rb \rbr + 4\sum_{s'\neq s}\Tr\lb Y_{i_{s'}}\sigma_{i_s}\rb.
\label{equ:PackingLemErr1}
\end{align}
By an elementary computation and using that $2\Pi-\one_d\leq \one_d$ we find that 
\begin{align*}
1- \Tr\lb (2\Pi - \one_d)\Pi_{i_s} \sigma_{i_s}\rb & = 1 -\Tr\lb (2\Pi-\one_d)\sigma_{i_s}\rb + \Tr\lb (2\Pi-\one_d)(\one_d-\Pi_{i_s})\sigma_{i_s}\rb \\
&\leq 1 -\Tr\lb (2\Pi-\one_d)\sigma_{i_s}\rb + \Tr\lb (\one_d-\Pi_{i_s})\sigma_{i_s}\rb \\
& = 3-2\Tr\lb \Pi\sigma_{i_s}\rb - \Tr\lb \Pi_{i_s}\sigma_{i_s}\rb .
\end{align*}
Combining this with \eqref{equ:PackingLemErr1} leads to 
\[
p_e(s,I) \leq 6 - 4\Tr\lb \Pi\sigma_{i_s}\rb - 2\Tr\lb \Pi_{i_s}\sigma_{i_s}\rb + 4\sum_{s'\neq s}\Tr\lb Y_{i_{s'}} \sigma_{i_s} \rb.
\]
For fixed $I = \lset i_1,\ldots ,i_M\rset$ we define the average decoding error by 
\[
\overline{p}_e\lb I\rb = \frac{1}{M}\sum^M_{s=1}p_e(s,I).
\]
Now define a $\lset 1,\ldots ,L\rset$-valued random variable $Z$ distributed according to the probability distribution $\lset p_i\rset^L_{i=1}$. Choosing the index set $I$ at random according to $M$ i.i.d.~copies of $Z$ leads to the following upper bound on the expected value of $\overline{p}_e$:
\begin{align*}
&\mathbb{E}_{Z_1,\ldots ,Z_M}\left[\overline{p}_e\lb \lset Z_1,\ldots ,Z_M\rset\rb\right] \\
&\leq \frac{1}{M}\sum^M_{s=1}\mathbb{E}_{Z_1,\ldots ,Z_M}\lbr 6 - 4\Tr\lb \Pi\sigma_{Z_s}\rb - 2\Tr\lb \Pi_{Z_s}\sigma_{Z_s}\rb + 4\sum_{s'\neq s}\Tr\lb Y_{Z_{s'}} \sigma_{Z_s} \rb \rbr \\
&= 6 + \frac{1}{M}\sum^M_{s=1} \left[ -4\mathbb{E}_{Z_s}\Tr\lb \Pi \sigma_{Z_s}\rb   -2\mathbb{E}_{Z_s}\Tr\lb \Pi_{Z_s} \sigma_{Z_s}\rb + 4\sum_{s'\neq s}\mathbb{E}_{Z_s,Z_{s'}}\Tr\lb Y_{Z_{s'}}\sigma_{Z_s}\rb\right]
\end{align*}
Using the assumptions 1.~,2.~,4.~, and 5.~from above, and that $\mathbb{E}_{Z}\sigma_Z = \sigma$ we find 
\begin{align*}
&\mathbb{E}_{Z_1,\ldots ,Z_M}\left[\overline{p}_e\lb \lset Z_1,\ldots ,Z_M\rset\rb\right] \\
&\leq 6 + \frac{1}{M}\sum^M_{s=1} \left[ -4\Tr\lb \Pi \sigma\rb -2\mathbb{E}_{Z_s}\Tr\lb \Pi_{Z_s} \sigma_{Z_s}\rb + 4\sum_{s'\neq s}\mathbb{E}_{Z_{s'}}\Tr\lb \Pi_{Z_{s'}}\Pi \sigma\Pi\rb\right] \\
&\leq 6 + \frac{1}{M}\sum^M_{s=1} \left[ -4(1-\epsilon_1) -2(1-\epsilon_2)+ 4\sum_{s'\neq s}\mathbb{E}_{Z_{s'}}\Tr\lb \Pi_{Z_{s'}}\frac{1}{B}\Pi\rb\right] \\
&\leq 4\epsilon_1 + 2\epsilon_2 + 4M\frac{A}{B}
\end{align*}
The previous estimate shows the existence of an $M$-tuple $I=\lb i_1,\ldots ,i_M\rb$ such that 
\[
\overline{p}_e\lb I\rb \leq 4\epsilon_1 + 2\epsilon_2 +  4M\frac{A}{B} .
\]
Choosing this index set for the coding scheme gives the result of the lemma.
\end{proof}

To prove the direct part of the Holevo-Schumacher-Westmoreland theorem (following the strategy presented in~\cite{wilde2017quantum}) we can use typical projectors in the packing lemma. Specifically, let $\lset p_X(x),\sigma_x\rset^L_{x=1}$ be an ensemble of quantum states with average state $\sigma=\sum^L_{x=1} p_X(x)\sigma_x$ and such that each $\sigma_x$ has eigenvalues $p_{Y|X}(y|x)$ for $y\in\lset 1,\ldots ,L\rset$. For any $\delta>0$ let $\Pi^n_\delta$ be the $\delta$-typical projector with respect to $\sigma^{\otimes n}$ and for any $x^n\in \lset 1,\ldots ,L\rset^n$ satisfying $p_{X^n}(x^n)>0$ let $\Pi^{B^n|x^n}_\delta$ denote the conditional typical projector with respect to the ensemble $\lset p_{X^n}(x^n),\rho_{x^n}\rset$ conditioned onto $x^n$. Applying Theorem \ref{thm:QTypical1}, Theorem \ref{thm:QTypical2}, Theorem \ref{thm:QTypical3} and Theorem \ref{thm:SizeQTyp} shows that 
\begin{enumerate}
\item $\Tr\lbr\Pi^n_\delta \sigma^{\otimes n} \rbr \geq 1-\exp\lb \frac{-2n\delta^2}{\log(\lambda_{\min})^2}\rb$. 
\item $\sum_{x^n} p_{X^n}(x^n)\Tr\lbr\Pi^{B^n|x^n}_\delta \sigma_{x^n} \rbr \geq 1-\exp\lb\frac{-2n\delta^2}{\log(\mu_{\min})^2}\rb$.
\item $\lbr \Pi^{B^n|x^n}_\delta,\sigma_{x^n} \rbr = 0$ for any $x^n\in\lset 1,\ldots, L\rset^n$. 
\item $\Tr\lbr\Pi^{B^n|x^n}_\delta\rbr\leq 2^{n\lb H(X,Y) - H(X) + \delta\rb}$ for any $x^n\in\lset 1,\ldots, L\rset^n$.
\item $\Pi^n_\delta \sigma^{\otimes n}\Pi^n_\delta \leq 2^{-n(S(\sigma)-\delta)}\Pi^n_\delta$.
\end{enumerate} 
Here, we used 
\[
\lambda_{\min}:=\min\lset \lambda\in\text{spec}(\sigma)\setminus\lset 0\rset\rset,
\]
and
\[
\mu_{\min} = \min_{x\in \text{supp}\lb X\rb}\min\lset \lambda\in\text{spec}\lb \sigma_x\rb\setminus\lset 0\rset\rset
\]
Applying Lemma \ref{equ:PackingLemErr1} shows that for every $M\leq L^n$ there exist $I=\lb x(1)^n , \ldots ,x(M)^n\rb\in \lb\lset 1,\ldots ,L\rset^n\rb^M$, and a POVM $\lb\Lambda_m\rb^M_{s=1}\in ((\M^{\otimes n}_L)^+)^M$ such that 
\[
\frac{1}{M}\sum^M_{s=1}\Tr\lb \Lambda_s \sigma_{x(s)^n}\rb\geq 1 - 4\exp\lb \frac{-2n\delta^2}{\log(\lambda_{\min})^2}\rb - 2\exp\lb\frac{-2n\delta^2}{\log(\mu_{\min})^2}\rb - 4 M2^{-n\lb \chi\lb\lset p_X(x),\sigma_x\rset^L_{x=1}\rb - 2\delta\rb},
\] 
with the Holevo quantity of the ensemble $\lset p_X(x),\sigma_x\rset^L_{x=1}$ given by 
\[
\chi\lb\lset p_X(x),\sigma_x\rset^L_{x=1}\rb = S(\sigma)+H(X)-H(X,Y).
\]
Given a quantum channel $T:\M_{d_1}\ra \M_{d_2}$ and an ensemble $\lset p_{X}(x),\rho_x\rset^L_{i=1}$ we can apply the above reasoning to the ensemble $\lset p_{X}(x),T(\rho_x)\rset^L_{i=1}$ which leads to a coding scheme for the cq-channel $x\mapsto T(\rho_x)$ and the following theorem:

\begin{thm}[Error bound classical capacity]
Let $T:\M_{d_1}\ra \M_{d_2}$ be a quantum channel, $\lset p_{X}(x),\rho_x\rset^L_{i=1}$ an ensemble of quantum states on $\C^{d_1}$, and $\sigma=\sum_{x} p_{X}(x)T(\rho_x)$ the average state at the channel output. For any $M=2^{nR}$ with $R\leq \log(L)$ and any $\delta>0$, there exists $I=\lb x(1)^n , \ldots ,x(M)^n\rb\in \lb\lset 1,\ldots ,L\rset^n\rb^M$, and a POVM $\lb\Lambda_s\rb^M_{s=1}\in ((\M^{\otimes n}_L)^+)^M$ such that 
\begin{align*}
\frac{1}{M}\sum^M_{s=1}&\Tr\lb \Lambda_s T^{\otimes n}\lb\rho_{x(s)^n}\rb\rb\\
&\geq 1 - 4\exp\lb \frac{-2n\delta^2}{\log(\lambda_{\min})^2}\rb - 2\exp\lb\frac{-2n\delta^2}{\log(\mu_{\min})^2}\rb - 4\cdot 2^{n\lb R-\chi\lb\lset p_X(x),T(\rho_x)\rset\rb + 2\delta\rb} 
\end{align*}
with
\[
\lambda_{\min}=\min\lset \lambda\in\text{spec}(\sigma)\setminus\lset 0\rset\rset,
\]
and
\[
\mu_{\min} = \min_{x\in \text{supp}\lb X\rb}\min\lset \lambda\in\text{spec}\lb T(\rho_x)\rb\setminus\lset 0\rset\rset.
\]
\label{thm:HSWErrorBound}
\end{thm}

\section{Explicit error bound in the LSD-theorem}
\label{app:LSDError}

The following theorem follows from analyzing the coding scheme from~\cite{hayden2008decoupling}.

\begin{thm}[Decoupling with error bound]
For any quantum channel $T:\M_{d_1}\ra\M_{d_2}$, any pure state $\ket{\phi_{AA'}}\in \C^{d_1}\otimes \C^{d_1}$, any $m\in\N$, any $\delta>0$, and any rate $R>0$ there exists an encoder $E_m:\M^{\otimes Rm}_2\ra\M^{\otimes m}_{d_1}$ and a decoder $D_m:\M^{\otimes m}_{d_2}\ra \M^{\otimes Rm}_2$ such that 
\[
F\lb \ket{\Omega^{\otimes Rm}_2},\lb\ident^{\otimes Rm}_2\otimes D_m\circ T^{\otimes m}\circ E_m\rb\lb \omega^{\otimes Rm}_2\rb\rb \geq 1-\epsilon_m ,
\]
where $\omega_2=\proj{\Omega_2}{\Omega_2}$ denotes the maximally entangled qubit state with $\ket{\Omega_2}=(\ket{00}+\ket{11})/\sqrt{2}$, and
\[
\epsilon_m = 4\sqrt{3}\exp\lb - \frac{m\delta^2}{\log\lb \mu_{\min}\lb \phi_{A'}\rb\rb^2}\rb + 2^{-\frac{m}{2}\lb I_{\text{coh}}\lb \phi_{A'},T\rb - R - 3\delta\rb} 
\]
for
\[
\mu_{\min}(\phi_{A'}) = \min\lset \lambda>0 ~:~\lambda\in \text{spec}\lb \phi_{A'}\rb\cup \text{spec}\lb T(\phi_{A'})\rb\cup\text{spec}\lb T^c(\phi_{A'})\rb\rset,
\]
where $T^c$ denotes the complementary channel of $T$.
\label{thm:DecouplWithError}
\end{thm}

The error measure in Theorem \ref{thm:DecouplWithError} is called the entanglement generation fidelity or channel fidelity (cf.~\cite{kretschmann2004tema}). This quantity can be related to the minimum fidelity
\[
F(T)=\min\lset \bra{\psi}T(\proj{\psi}{\psi})\ket{\psi} ~:~\ket{\psi}\in\C^{d_1}, \braket{\psi}{\psi}=1\rset.
\]
Specifically, we can use~\cite[Proposition 4.5.]{kretschmann2004tema} modifying the coding scheme slightly to find the following corollary.

\begin{cor}
For any quantum channel $T:\M_{d_1}\ra\M_{d_2}$, any pure state $\ket{\phi_{AA'}}\in \C^{d_1}\otimes \C^{d_1}$, any rate $R>0$ and any $m\in\N$ such that $Rm>1$, there exists an encoder $E_m:\M^{\otimes (Rm-1)}_2\ra\M^{\otimes m}_{d_1}$ and a decoder $D_m:\M^{\otimes m}_{d_2}\ra \M^{\otimes (Rm-1)}_2$ such that 
\[
F\lb D_m\circ T^{\otimes m}\circ E_m\rb \geq 1-\tilde{\epsilon}_m
\]
with 
\[
\tilde{\epsilon}_m = 1-2\epsilon_m, 
\]
and $\epsilon_m$ as in Theorem \ref{thm:DecouplWithError}. Note the coding scheme given by $E_m$ and $D_m$ still achieves the rate $R$.
\label{cor:DecouplWithErrorChanFidelity}
\end{cor}

\bibliographystyle{IEEEtranSN}
\bibliography{Biblio.bib}

\end{document}